\documentclass[onecolumn,nofootinbib,thightenlines,notitlepage,tightenlines,longbibliography,superscriptaddress,11pt]{revtex4-1} 
\newcommand{\pagenumbaa}{1}
\usepackage{graphicx}
\usepackage{amssymb}
\usepackage{amstext}

\usepackage{enumerate}

\usepackage{tikz}
\usetikzlibrary{chains}
\usetikzlibrary{fit}
\usepackage{pgflibraryarrows}		
\usepackage{pgflibrarysnakes}		
\usepackage{xcolor}
\usepackage{epsfig}
\usetikzlibrary{shapes.symbols,patterns} 
\usepackage{pgfplots}

\usepackage{amsthm}
\usepackage{amsmath}
\usepackage{amsfonts}
\usepackage{amssymb,amstext}
\usepackage{bbm} 
\usepackage{dsfont}

\usepackage[colorlinks=true,urlcolor=blue, hyperindex, breaklinks=true] {hyperref}
\usepackage{todonotes}

\usepackage{footnote}

\theoremstyle{plain}
\newtheorem{mythm}{Theorem}[section]
\newtheorem{myprop}[mythm]{Proposition}

\newtheorem{mylem}[mythm]{Lemma}
\newtheorem{myclaim}[mythm]{Claim}

\theoremstyle{definition}
\newtheorem{mydef}[mythm]{Definition}
\newtheorem{myex}[mythm]{Example}
\newtheorem{myremark}[mythm]{Remark}
\newtheorem{myass}[mythm]{Assumption}

\newcommand{\norm}[1]{\left\lVert#1\right\rVert}

\newcommand{\brackett}[1]{ \left \lbrace \text{#1}\right \rbrace }
\newcommand{\bracket}[1]{[#1]}

\def\Hb{\ensuremath{H_{\rm b}}}
\newcommand{\W}{\mathsf{W}} 

\def\setC{\mathsf{c}}

\newcommand{\introsection}[1]{\emph{#1}.---}
\newcommand{\Hh}[1]{H\!\left({#1}\right)} 
\newcommand{\Hdiff}[1]{h\!\left({#1}\right)} 
\newcommand{\I}[2]{I\!\left({#1},{#2}\right)} 
\newcommand{\D}[2]{D\!\left({#1}\right| \!\!\left|{#2}\right)} 

\newcommand{\Prob}[1]{\,{\mathds P} \!\left[#1\right]} %
\newcommand{\E}[1]{\,{\mathds E}\!\left[#1\right]} 

\DeclareMathOperator{\st}{s.t.}

\newcommand{\drv}{\ensuremath{\mathrm{d}}}

\newcommand{\inprod}[2]{\ensuremath{\left\langle{#1}\vphantom{\big|},\vphantom{\big|}{#2}\right\rangle}}
\newcommand{\transp}{\ensuremath{^{\scriptscriptstyle{\top}}}}
\newcommand{\WW}{\mathcal{W}}
\newcommand{\A}{\mathbb{A}}
\newcommand{\B}{\mathbb{B}}

\newcommand{\Borelsigalg}[1]{\ensuremath{\mathcal{B}\!\left(#1\right)}}

\newcommand{\R}{\ensuremath{\mathbb{R}}}
\newcommand{\Rp}{\ensuremath{\R_{\geq 0}}}
\newcommand{\Rsp}{\ensuremath{\R_{> 0}}}

\newcommand{\Lp}[1]{\mathrm{L}^{#1}}

\allowdisplaybreaks

\usepackage{makecell}

\long\def\symbolfootnote[#1]#2{\begingroup\def\thefootnote{\fnsymbol{footnote}}\footnote[#1]{#2}\endgroup}

\usepackage{dsfont,amssymb,amsmath,subfigure, graphicx} 
\usepackage{amsfonts,dsfont,mathtools, mathrsfs,amsthm} 

\newcommand{\X}{\mathcal{X}}
\newcommand{\Y}{\mathcal{Y}}
\newcommand{\N}{\mathbb{N}_0}

\newcommand{\Let}{\coloneqq}
\newcommand{\e}{\mathrm{e}}
\newcommand{\meas}{\mathcal{P}}
\newcommand{\diff}{\mathrm{d}}

\begin{document}

\symbolfootnote[0]{
The material in this paper was presented in part at the IEEE International Symposium on Information Theory, June 2014. $\vspace{1mm}$}
\title{Efficient Approximation of Channel Capacities}

 \author{Tobias Sutter}
 \email[]{$\brackett{sutter,\,mohajerin,\,lygeros}$@control.ee.ethz.ch}
 \affiliation{Automatic Control Laboratory, ETH Zurich, Switzerland}
 
 \author{David Sutter}
 \email[]{suttedav@phys.ethz.ch}
 \affiliation{Institute for Theoretical Physics, ETH Zurich, Switzerland}

 \author{Peyman Mohajerin Esfahani}
 \email[]{$\brackett{sutter,\,mohajerin,\,lygeros}$@control.ee.ethz.ch}
 \affiliation{Automatic Control Laboratory, ETH Zurich, Switzerland}

 \author{John Lygeros}
 \email[]{$\brackett{sutter,\,mohajerin,\,lygeros}$@control.ee.ethz.ch}
 \affiliation{Automatic Control Laboratory, ETH Zurich, Switzerland}


\begin{abstract}
We propose an iterative method for approximately computing the capacity of discrete memoryless channels, possibly under additional constraints on the input distribution. Based on duality of convex programming, we derive explicit upper and lower bounds for the capacity. 
The presented method requires $O(M^2 N \sqrt{\log N}/\varepsilon)$ to provide an estimate of the capacity to within $\varepsilon$, where $N$ and $M$ denote the input and output alphabet size; a single iteration has a complexity $O(M N)$. We also show how to approximately compute the capacity of memoryless channels having a bounded continuous input alphabet and a countable output alphabet under some mild assumptions on the decay rate of the channel's tail.
It is shown that discrete-time Poisson channels fall into this problem class. As an example, we compute sharp upper and lower bounds for the capacity of a discrete-time Poisson channel with a peak-power input constraint. 
\end{abstract}

 \maketitle

 \setcounter{page}{\pagenumbaa}  
 \thispagestyle{plain}

\section{Introduction}
A discrete memoryless channel (DMC) comprises a finite input alphabet $\mathcal{X} = \{1,2,\hdots,N\}$, a finite output alphabet $\mathcal{Y} = \{1,2,\hdots,M\}$, and a conditional probability mass function expressing the probability of observing the output symbol $y$ given the input symbol $x$, denoted by $W(y|x)$. In his seminal 1948 paper \cite{shannon48}, Shannon proved that the channel capacity for a DMC is
\begin{equation}
C(W)=\max \limits_{p \in \Delta_N} \I{p}{W},  \label{eq:shannon48}
\end{equation} 
where $\Delta_{N}\!\!:=\!\{  x\in\R^{N} \!: \ \! x\geq 0,  \sum_{i=1}^{N} x_{i}=1\}$\! denotes the $N$-simplex and $\I{p}{W}\!\!:=\!\sum_{x \in \mathcal{X}}  p(x)$ $ \D{W(\cdot|x)}{(pW)(\cdot)}$ the mutual information. $W(y|x)=\Prob{Y=y|X=x}$ describes the channel law and $(pW)(\cdot)$ is the probability distribution of the channel output induced by $p$ and $W$, i.e., $(pW)(y):=\sum_{x \in \mathcal{X}} p(x) W(y|x)$. $\D{\cdot}{\cdot}$ denotes the relative entropy that is defined as $\D{W(\cdot|x)}{(pW)(\cdot)}:=\sum_{y \in \mathcal{Y}} W(y|x) \log\left(\tfrac{W(y|x)}{(pW)(y)}\right)$.
Shannon also showed that in case of an additional average cost constraint on the input distribution of the form $\E{s(X)}\leq S$, where $s:\mathcal{X}\to\Rp$ denotes a cost function and $S\geq 0$, the capacity is given by
\begin{equation} \label{eq:DMC_capacity_const}
C_S(W)= \left\{
\begin{array}{lll}
			&\max\limits_{p} 		& \I{p}{W} \\
			&\st					& \E{s(X)}\leq S\\
			& 						& p\in \Delta_{N}.
	\end{array} \right.
\end{equation}
For a few DMCs it is known that the capacity can be computed analytically, however in general there is no closed-form solution. It is therefore of interest to have an algorithm that solves \eqref{eq:DMC_capacity_const} in a reasonable amount of time. Since for a fixed channel the mutual information is a concave function in $p$, the optimization problem \eqref{eq:DMC_capacity_const} is a finite dimensional convex optimization problem. Solving \eqref{eq:DMC_capacity_const} with convex programming solvers, however, turned out to be computationally inefficient even for small alphabet sizes \cite{blahut72}. 

Shannon's formula for the capacity of a DMC generalizes to the case of memoryless channels with continuous input and output alphabets, i.e. $\mathcal{X}=\mathcal{Y}=\R$.  However, when considering such channels, it is essential to introduce additional constraints on the channel input to obtain physically meaningful results, more details can be found in \cite[Chapter~7]{gallager68}.
In addition to average cost type constraints, peak-power constraints are also often considered. A peak-power constraint demands that $X\in \A$ for some compact set $\A\subset\mathcal{X}$ with probability one. For such a setup, i.e., having average and peak-power constraints, the capacity is given by
\begin{equation} \label{eq:cont_DMC_capacity_const}
C_{\A,S}(W)= \left\{
\begin{array}{lll}
			&\sup\limits_{p} 		& \I{p}{W} \\
			&\st					& \E{s(X)}\leq S\\
			& 						& p\in \mathcal{P}(\A),
	\end{array} \right.
\end{equation}
where $\mathcal{P}(\A)$ denotes the set of all probability distributions on the Borel $\sigma$-algebra $\Borelsigalg{\A}$ and the mutual information is defined as $\I{p}{W}:=\int_{\A} \D{W(\cdot|x)}{(pW)(\cdot)} p(\drv x)$. The channel is described by a transition density defined by $\Prob{Y\in\drv y|X=x}=W(y|x)\drv y$ and $(pW)(\cdot)$ is the probability distribution of the channel output induced by $p$ and $W$ which is given by $(pW)(y):=\int_{\A}W(y|x)p(\drv x)$ and the relative entropy that is defined as $\D{W(\cdot|x)}{(pW)(\cdot)}:=\int_{\mathcal{Y}} W(y|x) \log\left(\tfrac{W(y|x)}{(pW)(y)}\right)\drv y$.
The optimization problem \eqref{eq:cont_DMC_capacity_const} is an infinite dimensional convex optimization problem and as such in general computationally intractable (NP-hard).

\vspace{3mm}
\introsection{Previous Work and Contributions}
Historically one of the first attempts to numerically solve \eqref{eq:DMC_capacity_const} is the so-called \emph{Blahut-Arimoto algorithm} \cite{blahut72,arimoto72}, that exploits the special structure of the mutual information and approximates iteratively the capacity of any DMC. Each iteration step has a computational complexity $O(MN)$. It was shown that this algorithm, in case of no additional input constraints has an \textit{a priori} error bound of the form $|C(W)-C_{\textnormal{approx}}^{(n)}(W)|\leq O(\tfrac{\log(N)}{n})$, where $n$ denotes the number of iterations \cite[Corollary~1]{arimoto72}. Hence, the overall computational complexity of finding an additive $\varepsilon$-solution is given by $O(\tfrac{MN\log(N)}{\varepsilon})$. As such the computational cost required for an acceptable accuracy for channels with large input alphabets can be considerable. 
This undesirable property together with the complexity per iteration prevents the algorithm from being useful for a large class of channels, e.g., a Rayleigh channel with a discrete input alphabet \cite{shamai01}.
There have been several improvements of the Blahut-Arimoto algorithm \cite{sayir00,matz04,yaming10}, which achieve a better convergence for certain channels. However, since they all rely on the original Blahut-Arimoto algorithm they inherit its overall computational complexity as well as its complexity per iteration step. Therefore, even with improved Blahut-Arimoto algorithms, approximating the capacity for channels having large input alphabets remains computationally expensive. 
Based on sequential Monte-Carlo integration methods (a.k.a.\ particle filters), the Blahut-Arimoto algorithm has been extended to memoryless channels with continuous input and output alphabets \cite{dauwels05,ref:Chen-13,ref:Chen-14-1, ref:Chen-14-2}. As shown in several examples, this approach seems to be powerful in practice, however a rate of convergence has not been proven.

Another recent approach towards approximating \eqref{eq:DMC_capacity_const} is presented in \cite{chiang04} by Mung and Boyd, where 
they introduce an efficient method to derive upper bounds on the channel capacity problem, based on geometric programming.
Huang and Meyn \cite{meyn05} developed a different approach based on cutting plane methods, where the mutual information is iteratively approximated by linear functionals and in each iteration step, a finite dimensional linear program is solved. It has been shown that this method converges to the optimal value, however no rate of convergence is provided. 

In this article, we present a new approach to solve \eqref{eq:DMC_capacity_const} that is based on its dual formulation. It turns out that the dual problem of \eqref{eq:DMC_capacity_const} has a particular structure that allows us to apply Nesterov's smoothing method \cite{nesterov05}. In the absence of input cost constraints, this leads to an a priori error bound of the order $|C(W)-C_{\textnormal{approx}}^{(n)}(W)|\leq O(\tfrac{M \sqrt{\log(N)}}{n})$, where $n$ denotes the number of iterations and each iteration step has a  computational complexity of $O\!\left(NM \right)$.
Thus, the overall computational complexity of finding an $\varepsilon$-solution is given by $O(\tfrac{M^{2} N \sqrt{\log(N)}}{\varepsilon})$.
In particular for large input alphabets our method has a computational advantage over the Blahut-Arimoto algorithm. In addition the novel method provides primal and dual optimizers leading to an \emph{a posteriori} error which is often much smaller than the a priori error. 

Due to the favorable structure of the capacity problem and its dual formulation, the presented method can be extended to approximate the capacity of memoryless channels having a bounded continuous input alphabet and a countable output alphabet, under some assumptions on the tail of $W(\cdot|x)$, i.e., problem \eqref{eq:cont_DMC_capacity_const} is addressed for a countable output alphabet. As a concrete example, this is demonstrated on the discrete-time Poisson channel with a peak-power constraint. To the best of our knowledge, for this scenario up to now only lower bounds exist \cite{lapidoth09}.

 \vspace{3mm}
\introsection{Structure} Section~\ref{sec:classicalCapacity} introduces our method for approximating the channel capacity for DMCs. We provide a priori and a posteriori bounds for the approximation error and present two numerical examples that illustrate its computational performance compared to the Blahut-Arimoto algorithm. In Section~\ref{sec:cont:Channels}, we generalize the approximation scheme to channels having bounded continuous input alphabets and countable output alphabets. We then show how the presented results can be used to compute the capacity of discrete-time Poisson channels under a peak-power constraint and possibly average-power constraints on the input. We conclude in Section~\ref{sec:conclusion} with a summary and potential subjects of further research. In the interest of readability, some of the technical proofs and details are given in the appendices.

\vspace{3mm}
\introsection{Notation}
The logarithm with basis 2 is denoted by $\log(\cdot)$ and the natural logarithm by $\ln(\cdot)$. In Section~\ref{sec:classicalCapacity} we consider DMCs with a finite input alphabet $\mathcal{X}=\{ 1,2,\hdots,N \}$ and a finite output alphabet $\mathcal{Y}=\{ 1,2,\hdots,M \}$. The channel law is summarized in a matrix $\W\in\R^{N\times M}$, where $\W_{ij}:=\Prob{Y=j|X=i}=W(j|i)$. We define the standard $n-$simplex as $\Delta_{d}:=\left\{  x\in\R^{d} : x\geq 0, \sum_{i=1}^{d} x_{i}=1\right\}$. The input and output probability mass functions are denoted by the vectors $p\in \Delta_{N}$ and $q\in\Delta_{M}$. The input cost constraint can be written as $\E{s(X)} = p\transp s\leq S$, where $s\in \Rp^{N}$ denotes the cost vector and $S\in \Rp$ is the given total cost. The binary entropy function is denoted by $\Hb(\alpha):=-\alpha \log(\alpha)-(1-\alpha)\log(1-\alpha)$, for $\alpha\in [0,1]$. For a probability mass function $p \in \Delta_{N}$ we denote the entropy by $H(p):=\sum_{i=1}^N -p_i \log(p_i)$. It is convenient to introduce an additional variable for the conditional entropy of $Y$ given $\{X=i\}$ as $r\in\R^{N}$, where $r_{i}=-\sum_{j=1}^{M}\W_{ij}\log(\W_{ij})$. For a probability density $p$ supported at a measurable set $B\subset \R$ we denote the differential entropy by $h(p)=-\int_{B} p(x) \log(p(x)) \drv x$.  For two vectors $x,y \in \R^n$, we denote the canonical inner product by $\left \langle x,y \right \rangle := x \transp y$. We denote the maximum (resp.~minimum) between $a$ and $b$ by $a \vee b$ (resp.~$a\wedge b$). For $\A\subset\R$ and $1\leq p \leq \infty$, let $\Lp{p}(\A)$ denote the space of $\Lp{p}$-functions on the measure space $(\A, \Borelsigalg{\A}\!, \drv x)$, where $\Borelsigalg{\A}$ denotes the Borel $\sigma$-algebra and $\drv x$ the Lebesgue measure.
The capacity of a channel $W$ is denoted by $C(W)$. 
For the channel law matrix  $\W\in\R^{N\times M}$ we consider the norm
$\|\W\|:= \max\limits_{\lambda\in\R^{M}, \ p\in\R^{N}}\left\{ \inprod{\W\lambda}{p} \ : \ \|\lambda\|_{2}=1, \ \|p\|_{1}=1 \right\},$
and note that an upper bound is given by 
\begin{equation} \label{eq:operator:norm}
\|\W\| 			=		\max\limits_{\|p\|_{1}=1} \max\limits_{\|\lambda\|_{2}=1} \lambda\transp \W\transp p \leq 	\max\limits_{\|p\|_{1}=1} \| \W\transp p \|_{2}\leq 	\max\limits_{\|p\|_{1}=1} \| \W\transp p \|_{1} =	\max\limits_{\|p\|_{1}=1} \| p\|_{1}=	1.
\end{equation}

\section{Discrete Memoryless Channel} \label{sec:classicalCapacity}
To keep notation simple we consider a single average-input cost constraint as the extension to multiple average-input cost constraints is straightforward. In a first step, we introduce the output distribution $q\in\Delta_{M}$ as an additional decision variable, as done in \cite{benTal88,chiang04,chiang05} and note that the mutual information $I(X;Y)$ is equal to $H(Y)-H(Y|X)$.
\begin{mylem} \label{lem:equivalent:primal:problem}
Let $\mathcal{F}:=\arg\max\limits_{p\in\Delta_{N}} \I{p}{W}$ and $S_{\max}:=\min\limits_{p\in\mathcal{F}}s\transp p$. If $S\geq S_{\max}$ the optimization problem \eqref{eq:DMC_capacity_const} has the same optimal value as
\begin{equation} \label{opt:primal:equivalent:no:power:constraints}
 	\mathsf{P}: \quad \left\{ \begin{array}{lll}
			&\max\limits_{p,q} 		&- r\transp p + H(q) \\
			&\st					& \W\transp p = q\\
			& 					& p\in \Delta_{N}, \ q\in\Delta_{M}.
	\end{array} \right.
\end{equation}
If $S<S_{\max}$ the optimization problem \eqref{eq:DMC_capacity_const} has the same optimal value as
\begin{equation} \label{opt:primal:equivalent}
 	\mathsf{P}: \quad \left\{ \begin{array}{lll}
			&\max\limits_{p,q} 		&- r\transp p + H(q) \\
			&\st					& \W\transp p = q\\
			&					& s\transp p = S \\
			& 					& p\in \Delta_{N}, \ q\in\Delta_{M}.
	\end{array} \right.
\end{equation}
\end{mylem}
\begin{proof}
The proof can be found in Appendix~\ref{ap:LemmaMung}.
\end{proof}
Note that we later add an assumption on our channel (Assumption~\ref{ass:channel}) that guarantees uniqueness of the optimizer maximizing the mutual information, i.e., $\mathcal{F}$ is a singleton. In this case the optimizer to \eqref{opt:primal:equivalent} (resp. \eqref{opt:primal:equivalent:no:power:constraints}) is also feasible for the original problem \eqref{eq:DMC_capacity_const}. Computing $S_{\max}$ is straightforward once $\mathcal{F}$ is known. The singleton $\mathcal{F}$ can be seen as the maximizer of a channel capacity problem with no additional input cost constraint and can as such be computed with the scheme we present in this article. 

For the rest of the section we restrict attention to \eqref{opt:primal:equivalent}, since the less constrained problem \eqref{opt:primal:equivalent:no:power:constraints} can be solved in a similar, more direct way.
We tackle this optimization problem through its Lagrangian dual problem. The dual function turns out to be a non-smooth function. As such, it is known that the efficiency estimate of a black-box first-order method is of the order $O\left( \tfrac{1}{\varepsilon^{2}}\right)$ if no specific problem structure is used, where $\varepsilon$ is the desired abolute accuracy of the approximate solution in function value \cite{ref:nesterov-book-04}. 
We show, however, that $\mathsf{P}$ has a certain structure that allows us to use Nesterov's approach for approximating non-smooth problems with smooth ones \cite{nesterov05} leading to an efficiency estimate of the order $O\left( \tfrac{1}{\varepsilon}\right)$.  This, together with the low complexity of each iteration step in the approximation scheme leads to a numerical method for the channel capacity problem that has a very attractive computational complexity.

\subsection{Preliminaries}
Some preliminaries are needed in order to present our capacity approximation method. We begin by recalling Nesterov's seminal work \cite{nesterov05} in the context of structural convex optimization, which is our main tool in the proposed capacity approximation scheme. 
 \subsection*{Nesterov's smoothing approach \cite{nesterov05}} 
Consider finite-dimensional real vector spaces $E_i$ endowed with a norm $\|\cdot\|_i$ and denote its dual space by $E^\star_i$ for $i=1,2$. Each dual pair of vector spaces comes with a bilinear form $\inprod{\cdot}{\cdot}_i:E^\star_i\times E_i\to \R$. For a linear operator $A:E_1\to E_2^\star$ the operator norm is defined as $\|A\|_{1,2}=\max_{x,u}\{ \inprod{Ax}{u}_2 \ : \ \|x\|_1=1, \|u\|_2=1 \}$. We are interested in the following optimization problem
\begin{equation}\label{eq:Nesterov:1}
\min_x\{ f(x) \ : \ x\in Q_1\},
\end{equation}
where $Q_1\subset E_1$ is a compact convex set and $f$ is a continuous convex function on $Q_1$. We assume that the objective function has the following structure
\begin{equation} \label{eq:Nesterov:2:structure}
f(x) = \hat{f}(x) + \max_u \{ \inprod{Ax}{u}_2 -\hat{\phi}(u) \ : \ u\in Q_2\},
\end{equation}
where $Q_2\subset E_2$ is a compact convex set, $\hat{f}$ is a continuously differentiable convex function whose gradient is Lipschitz continuous with constant $L$ on $Q_1$ and $\hat{\phi}$ is a continuous convex function on $Q_2$. It is assumed that $\hat{\phi}$ and $Q_2$ are simple enough such that the maximization in \eqref{eq:Nesterov:2:structure} is available in closed form.
The dual program to \eqref{eq:Nesterov:1} can be given as
\begin{equation}\label{eq:Nesterov:dual}
\max_u \left\{ -\hat{\phi}(u) + \min_x \{ \inprod{Ax}{u}_2 + \hat{f}(x) \ : \ x\in Q_1 \}  \ : \ u\in Q_2 \right\}.
\end{equation}
The main difficulty in solving \eqref{eq:Nesterov:1} efficiently is its non-smooth objective function. Without using any specific problem structure the complexity for subgradient-type methods is $O\left( \tfrac{1}{\varepsilon^{2}}\right)$, where $\varepsilon$ is the desired abolute accuracy of the approximate solution in function value. Nesterov's work suggests that when approximating problems with the particular structure \eqref{eq:Nesterov:2:structure} by smooth ones, a solution to the 
non-smooth problem can be constructed with complexity in order of $O\left( \tfrac{1}{\varepsilon}\right)$. In addition, Nesterov shows that when solving the smooth problem, a solution to the dual problem \eqref{eq:Nesterov:dual} can be obtained, and as such an a posteriori statement about the duality gap is available that often is significantly tighter than the $O\left( \tfrac{1}{\varepsilon}\right)$ complexity bound.
Consider the the smooth approximation to problem \eqref{eq:Nesterov:1} given by
\begin{equation} \label{eq:Nesterov:smooth:problem}
\min_x\{ f_\nu (x) \ : \ x\in Q_1\},
\end{equation}
where $\nu>0$ and the objective function is given by
\begin{equation}\label{eq:Nesterov:2:smooth:objective}
f_\nu(x) = \hat{f}(x) + \max_u \{ \inprod{Ax}{u}_2 -\hat{\phi}(u)-\nu d(u) \ : \ u\in Q_2\},
\end{equation}
where $d:Q_2\to\R$ is continuous and strongly convex with convexity parameter $\sigma$. It can be shown that $f_\nu$ has a Lipschitz continuous gradient with Lipschitz constant $L+\tfrac{\| A \|_{1,2}^2}{\nu \sigma}$ \cite[Theorem~1]{nesterov05}. In this light, the optimization problem \eqref{eq:Nesterov:smooth:problem} belongs to a class of problems that can be solved in  
 $O\left( \tfrac{1}{\sqrt{\varepsilon}}\right)$ using a fast gradient method. The result \cite[Theorem~3]{nesterov05} explicitly details how, having solved the smooth problem \eqref{eq:Nesterov:smooth:problem}, primal and dual solutions to the non-smooth problems \eqref{eq:Nesterov:1} and \eqref{eq:Nesterov:dual} can be obtained and how good they are.

 \subsection*{Entropy maximization} 
As a second preliminary result for some $c\in\R^{N}$ we consider the following optimization problem, that, if feasible, has an analytical solution
\begin{equation} \label{opt:cover}
 	\left\{ \begin{array}{lll}
			&\underset{p}{\max} 		&H(p) - c\transp p \\
			&\text{s.t. } 			& s\transp p = S\\
			& 					& p\in \Delta_{N}.
	\end{array} \right.
\end{equation}
\begin{mylem} \label{lem:cover}
Let $p^\star=[p_1^\star, \ldots, p_N^\star]$ with $p_i^\star =2^{\mu_1 - c_i + \mu_2 s_i}$, where $\mu_1$ and $\mu_2$ are chosen such that $p^\star$ satisfies the constraints in \eqref{opt:cover}. Then $p^\star$ uniquely solves \eqref{opt:cover}.
\end{mylem}
\begin{proof}
See Appendix~\ref{ap:cover}.
\end{proof}

\subsection{Capacity Approximation Scheme}
In the following we focus on the input constrained channel capacity problem \eqref{opt:primal:equivalent} and the scenario of no input constraints \eqref{opt:primal:equivalent:no:power:constraints} is discussed as a special case within this section.
Consider the convex optimizaton problem \eqref{opt:primal:equivalent}, whose optimal value, according to Lemma~\ref{lem:equivalent:primal:problem} is the capacity $C_{S}(W)$.
The Lagrange dual program to \eqref{opt:primal:equivalent} is
\begin{align}\label{Lagrange:Dual:Program}
\mathsf{D}:\quad \left\{ \begin{array}{ll}
			\underset{\lambda}{\min} 		&G(\lambda) + F(\lambda) \\
			\text{s.t. } 				& \lambda\in\R^{M},
	\end{array}\right. 
\end{align}
where $F, G: \R^{M}\to\R$ are given by
\begin{align} \label{equation:F:and:G} 
G(\lambda)= \left\{ \begin{array}{ll}
			\underset{p}{\max} 		&-r\transp p + \lambda\transp \W\transp p \\
			\text{s.t. } 				&s\transp p=  S \\
								& p\in\Delta_{N}
	\end{array} \right.
	\quad \textnormal{and} \qquad
	F(\lambda)= \left\{ \begin{array}{ll}
			\underset{q}{\max} 		&H(q)-\lambda\transp q \\
			\text{s.t. } 				& q\in\Delta_{M}.
	\end{array}\right. 
\end{align}
Note that since the coupling constraint $ \W\transp p = q$ in the primal program \eqref{opt:primal:equivalent} is affine, the set of optimal solutions to the dual program \eqref{Lagrange:Dual:Program} is nonempty \cite[Proposition~5.3.1]{ref:Bertsekas-09} and as such the optimum is attained. 
It can be seen that the dual program \eqref{Lagrange:Dual:Program} structurally resembles the problem \eqref{eq:Nesterov:1} with \eqref{eq:Nesterov:2:structure}, without a bounded feasible set, however. 
To ensure that the set of dual optimizers is compact, we need to impose the following assumption on the channel matrix $\W$, that we will maintain for the remainder of Section~\ref{sec:classicalCapacity}.

\begin{myass} \label{ass:channel}
$\gamma:=\min\limits_{i,j}\W_{ij}>0$
\end{myass}
Assumption~\ref{ass:channel} excludes situations where the channel matrix has zero entries. Even though this may seem restrictive at first glance, it holds for a large class of channels. Moreover, in a finite dimensional setting, for a fixed input distribution, the mutual information is well known to be continuous in the channel matrix entries. Therefore, singular cases where the channel matrix contains zero entries can be avoided by slight perturbations of those entries. (This is discussed in more detail in Remark~\ref{rmk:perturbation}.)
Under Assumption~\ref{ass:channel} for a fixed channel, the mutual information can be seen to be a strictly concave function in the input distribution. Therefore, the capacity achieving input distribution is unique.
With Assumption~\ref{ass:channel} one can derive an explicit bound on the norm of the dual optimizers, which is crucial in the subsequent derivation of the main result in this section, namely Theorem~\ref{thm:error:bound:capacity}.   
\begin{mylem} \label{lem:compact:set}
Under Assumption~\ref{ass:channel}, the dual program \eqref{Lagrange:Dual:Program} is equivalent to 
\begin{equation}\label{eq:dual:finite:compact}
\begin{aligned}
\left\{ \begin{array}{ll}
			\underset{\lambda}{\min} 		&G(\lambda) + F(\lambda) \\
			\textnormal{s.t. } 				& \lambda\in Q,
	\end{array}\right. 
\end{aligned}
\end{equation}
where $Q:= \left\{ \lambda\in\R^{M} \ : \ \norm{\lambda}_{2}\leq M \left( \log(\gamma^{-1}) \vee \tfrac{1}{\ln 2} \right) \right\}$. 
\end{mylem}
\begin{proof}
See Appendix~\ref{ap:bounding:lambda}.
\end{proof}
\begin{mylem} \label{lem:strong:duality:finite}
Strong duality holds between \eqref{opt:primal:equivalent} and \eqref{Lagrange:Dual:Program}.
\end{mylem}
\begin{proof}
The proof follows by a standard strong duality result of convex optimization, see \cite[Proposition~5.3.1, p.~169]{ref:Bertsekas-09}.
\end{proof}
Note that the optimization problem defining $F(\lambda)$ is of the form given in \eqref{opt:cover}. Hence, according to Lemma~\ref{lem:cover}, $F(\lambda)$ has a unique optimizer $q^{\star}$ with components
$q_{j}^{\star} = 2^{\mu-\lambda_{j}}$,
where $\mu\in\R$ needs to be chosen such that $q^{\star}\in\Delta_{M}$, i.e.,
\begin{equation*}
\mu = - \log\left(\sum_{j=1}^M 2^{-\lambda_j} \right).
\end{equation*}
Therefore, 
\begin{equation} \label{analytical:solution:F}
\begin{aligned}
F(\lambda) 	&=\sum_{j=1}^{M} \left(- q_{j}^{\star}\log(q_{j}^{\star}) - \lambda_{j}q_{j}^{\star} \right) =	-\sum_{j=1}^{M} \mu \, 2^{\mu-\lambda_{j}} = -\mu \, 2^{\mu}\sum_{j=1}^{M}2^{-\lambda_{j}} =    \log\left(\sum_{j=1}^M 2^{-\lambda_j} \right).
\end{aligned}
\end{equation}
$F(\lambda)$ is a smooth function with gradient
\begin{align} \label{eq:gradient:F}
(\nabla F(\lambda))_{i} = \frac{-2^{-\lambda_{i}}}{\sum_{j=1}^M 2^{-\lambda_j} }.
\end{align}
According to \cite[Theorem~1]{nesterov05} and the fact that the negative entropy is strongly convex with convexity parameter 1 \cite[Lemma~3]{nesterov05}, $\nabla F(\lambda)$ is Lipschitz continuous with Lipschitz constant $1$.
The main difficulty in solving \eqref{eq:dual:finite:compact} efficiently is that $G(\cdot)$ is non-smooth. Following Nesterov's smoothing technique \cite{nesterov05}, we alleviate this difficulty by approximating $G(\cdot)$ by a function with a Lipschitz continuous gradient. This smoothing step is efficient in our case because of the particular structure of $\eqref{eq:dual:finite:compact}$.
Following \cite{nesterov05} and \eqref{eq:Nesterov:2:smooth:objective}, consider
\begin{align} \label{eq:discrete:G:nu}
G_{\nu}(\lambda) 	=    \left\{ \begin{array}{ll}
	\max\limits_{p}				&  \lambda\transp \W\transp p - r\transp p + \nu H( p) -\nu\log(N)\\
			\text{s.t. } 			&  s\transp p=  S \\
							&  p\in\Delta_{N},
	\end{array}\right. 
\end{align}
with smoothing parameter $\nu\in\R_{>0}$ and denote by $p_{\nu}(\lambda)$ the optimizer to \eqref{eq:discrete:G:nu}, which is unique because the objective function is strictly concave. Clearly for any $\lambda \in Q$, $G_{\nu}(\lambda)$ is a uniform approximation of the non-smooth function $G(\lambda)$, since $G_{\nu}(\lambda)\leq G(\lambda)\leq G_{\nu}(\lambda) + \nu \log(N)$.
Using Lemma~\ref{lem:cover}, the optimizer $p_{\nu}(\lambda)$ to \eqref{eq:discrete:G:nu} is analytically given by
\begin{align} \label{eq:finite:optimizer:pmu}
p_{\nu}(\lambda,\mu)_{i} = 2^{\mu_{1} + \tfrac{1}{\nu} (\W \lambda \, -\,  r)_{i} + \mu_{2}s_{i}},
\end{align}
where $\mu_{1},\mu_{2}\in\R$ have to be chosen so that $s\transp p_{\nu}(\lambda,\mu)=S$ and $p_{\nu}(\lambda,\mu)\in\Delta_{N}$; for this choice of $\mu_1$, $\mu_2$ we denote the solution by $p_{\nu}(\lambda)$. 
\begin{myremark}\label{rmk:stabilization:optimizer}
In case of no input constraints, the unique optimizer to \eqref{eq:discrete:G:nu} is given by
\begin{equation*}
p_{\nu}(\lambda)_{i} = \frac{2^{ \tfrac{1}{\nu} (\W\lambda -  r)_{i}} }{\sum_{i=1}^{N}2^{ \tfrac{1}{\nu} (\W\lambda -  r)_{i}}}\quad \text{for }i=1,\hdots,N,
\end{equation*}
whose straightforward evaluation is numerically difficult for small $\nu$. One can circumvent this problem, however, by following the numerically stable technique that we present in Remark~\ref{rmk:finite:no:input:const}.
By Dubin's theorem it can be shown that the capacity of a memoryless channel with a discrete output alphabet of size $M$ and input alphabet size $N\geq M$, is achieved by a discrete input distribution with $M$ mass points \cite{gallager68,witsenhausen}. Computing the exact positions and weights of this optimal input distribution may be difficult, though it is worth noting that our analytical solution in \eqref{eq:finite:optimizer:pmu} converges to this optimal input distribution as $\nu$ tends to $0$.
\end{myremark}
\begin{myremark}[Additional input constraints]\label{rmk:finite:constraint:optimizer}
In case of additional input constraints, we need an efficient method to find the coefficients $\mu_{1}$ and $\mu_{2}$ in \eqref{eq:finite:optimizer:pmu}. In particular if there are multiple input constraints (leading to multiple $\mu_{i}$) the efficiency of the method computing them becomes important. Instead of solving a system of nonlinear equations, one can show (\cite[Theorem~4.8]{ref:Borwein-91}, \cite[p.~257 ff.]{ref:Lasserre-11}) that the coefficients $\mu_{i}$ are the unique maximizers to the following convex optimization problem
\begin{equation} \label{eq:opt:problem:find:mu:finite}
\max\limits_{\mu\in\R^{2}}\left\{ y\transp \mu - \sum_{i=1}^{N}p_{\nu}(\lambda,\mu)_{i} \right\}, 
\end{equation}
where $y:=(1,S)$. Notice that $\eqref{eq:opt:problem:find:mu:finite}$ is an unconstrained maximization of a strictly concave function, whose gradient and Hessian can be directly computed as
\begin{equation*}
\left( \begin{array}{c} y_1 - \ln 2 \sum_{i=1}^N p_{\nu}(\lambda,\mu)_i \\ y_2 - \ln 2 \sum_{i=1}^N s_i p_{\nu}(\lambda,\mu)_i
\end{array} \right) \quad \text{and}\quad
\left( \begin{array}{cc} - (\ln 2)^2 \sum_{i=1}^N p_{\nu}(\lambda,\mu)_i & -(\ln 2)^2 \sum_{i=1}^N s_i p_{\nu}(\lambda,\mu)_i \\ -(\ln 2)^2 \sum_{i=1}^N s_i p_{\nu}(\lambda,\mu)_i & -(\ln 2)^2 \sum_{i=1}^N s^2_i p_{\nu}(\lambda,\mu)_i
\end{array} \right),
\end{equation*}
which allows the use of efficient second-order methods such as Newton's method. This method directly extends to multiple input constraints. 
Let us point out that Theorem~\ref{thm:error:bound:capacity}, quantifying the approximation error of the presented algorithm, is based on the assumption that the maximum entropy solution \eqref{eq:finite:optimizer:pmu} is available, meaning that one can solve \eqref{eq:opt:problem:find:mu:finite} for optimality. In the case of a finite input alphabet this assumption is not restrictive as we have argued that \eqref{eq:opt:problem:find:mu:finite} is easy to solve. For a continuous input alphabet, that we shall discuss in the subsequent section, however, finding the maximum entropy solution is numerically difficult as it involves integration problems. Therefore, in Remark~\ref{rmk:finite:constraint:optimizer:cts}, we comment on how the presented channel capacity algorithm behaves, when having access only to an approximate solution to the mentioned maximum entropy problem. 
\end{myremark}

Finally, we can show that the uniform approximation $G_{\nu}(\lambda)$ is smooth and has a Lipschitz continuous gradient, with known Lipschitz constant.
\begin{myprop} \label{prop:Lipschitz:continuity:DMC}
$G_{\nu}(\lambda)$ is well defined and continuously differentiable at any $\lambda\in Q$. Moreover, it is convex and its gradient $\nabla G_{\nu}(\lambda)=\W\transp p_{\nu}(\lambda)$ is Lipschitz continuous with Lipschitz constant $\tfrac{1}{\nu}$.
\end{myprop}
\begin{proof}
The proof follows directly from the proof of Theorem 1 and Lemma 3 in \cite{nesterov05} together with \eqref{eq:operator:norm}.
\end{proof}
We consider the smooth, convex optimization problem
\begin{align}\label{Lagrange:Dual:Program:smooth}
 \mathsf{D}_{\nu}:\quad \left\{ \begin{array}{ll}
	\min\limits_{\lambda} 		& F(\lambda) + G_{\nu}(\lambda) \\
			\text{s.t. } 					& \lambda\in Q,
	\end{array}\right.
\end{align}
whose objective function has a Lipschitz continuous gradient with Lipschitz constant $1+\tfrac{1}{\nu}$. As such $\mathsf{D}_{\nu}$ can be  be approximated with Nesterov's optimal scheme for smooth optimization \cite{nesterov05}, which is summarized in Algorithm~\hyperlink{algo:1}{1}, where $\pi_{Q}(x)$ denotes the projection operator of the set $Q$, defined in Lemma~\ref{lem:compact:set}, with $R:=M \left( \log(\gamma^{-1}) \vee \tfrac{1}{\ln 2}\right)$
\begin{equation*}
\pi_{Q}(x):=\left\{ \begin{array}{ll} R \tfrac{x}{\norm{x}_{2}}, & \norm{x}_{2}>R \\ x, & \text{otherwise.} \end{array} \right. 
\end{equation*}

 \begin{table}[!htb]
\centering 
\begin{tabular}{c}
  \Xhline{3\arrayrulewidth}  \hspace{1mm} \vspace{-3mm}\\ 
\hspace{35mm}{\bf{\hypertarget{algo:1}{Algorithm 1: } }} Optimal scheme for smooth optimization \hspace{35mm} \\ \vspace{-3mm} \\ \hline \vspace{-0.5mm}
\end{tabular} \\
\vspace{-5mm}
 \begin{flushleft}
  {\hspace{1.7mm}Choose some $\lambda_0 \in Q$}
 \end{flushleft}
 \vspace{-6mm}
 \begin{flushleft}
  {\hspace{1.7mm}\bf{For $k\geq 0$ do$^{*}$}}
 \end{flushleft}
 \vspace{-7mm}
 
  \begin{tabular}{l l}
{\bf Step 1: } & Compute $\nabla F(x_{k})+\nabla G_{\nu}(x_{k})$ \\
{\bf Step 2: } & $y_k = \pi_{Q}\left(-\frac{1}{L_{\nu}}\left( \nabla F(x_k)+\nabla G_{\nu}(x_k) \right) + x_k\right)$\\
{\bf Step 3: } &   $z_k=\pi_{Q}\left(-\frac{1}{L_{\nu}} \sum_{i=0}^{k} \frac{i+1}{2} \left(  \nabla F(x_i)+\nabla G_{\nu}(x_i) \right)\right)$\\
{\bf Step 4: } & $x_{k+1}=\frac{2}{k+3}z_{k} + \frac{k+1}{k+3}y_{k}$\\ 
  \end{tabular}
   \begin{flushleft}
  {\hspace{1.7mm}[*The stopping criterion is explained in Remark~\ref{remark:stopping}]}
  \vspace{-10mm}
 \end{flushleft}  
\begin{tabular}{c}
\hspace{36.7mm} \phantom{ {\bf{Algorithm:}} Optimal Scheme for Smooth Optimization}\hspace{36.7mm} \\ \vspace{-1.0mm} \\\Xhline{3\arrayrulewidth}
\end{tabular}
\end{table}
The following theorem provides explicit error bounds for the solution provided by Algorithm~\hyperlink{algo:1}{1} after $n$ iterations. Define the constants $D_{1}:=\tfrac{1}{2}(M \log(\gamma^{-1})\vee\tfrac{1}{\ln 2})^2 $ and $D_{2}:=\log(N)$.
\begin{mythm}[\cite{nesterov05}] \label{thm:error:bound:capacity}
Under Assumption~\ref{ass:channel}, for $n\in\mathbb{N}$ consider a smoothing parameter
\begin{align*}
\nu = \nu(n) = \frac{2}{n+1}\sqrt{\frac{D_{1}}{D_{2}}}.
\end{align*}
Then after $n$ iterations of Algorithm~\hyperlink{algo:1}{1} we can generate the approximate solutions to the problems \eqref{Lagrange:Dual:Program} and \eqref{eq:DMC_capacity_const}, namely,
\begin{align}
\hat{\lambda} = y_{n} \in Q\qquad \textnormal{and} \qquad \hat{p}=\sum_{k=0}^{n}\frac{2(k+1)}{(n+1)(n+2)} p_{\nu}(x_{k})\in \Delta_{N}, \label{eq:optInPut}
\end{align}
which satisfy 
\begin{align}
0\leq F(\hat{\lambda}) + G(\hat{\lambda}) - \I{\hat{p}}{W} \leq \frac{4}{n+1} \sqrt{D_{1} D_{2}} + \frac{4 D_{1}}{(n+1)^{2}}. \label{eq:EBB}
\end{align}
Thus, the complexity of finding an $\varepsilon$-solution to the problems \eqref{Lagrange:Dual:Program} and \eqref{eq:DMC_capacity_const} does not exceed
\begin{align}\label{eq:finite:complexity:thm}
4  \sqrt{D_{1} D_{2}} \ \frac{1}{\varepsilon} + 2 \sqrt{\frac{D_{1}}{\varepsilon}}.
\end{align}
\end{mythm}
\begin{proof}
The proof follows along the lines of \cite[Theorem~3]{nesterov05} and in particular requires Lemma~\ref{lem:compact:set}, Lemma~\ref{lem:strong:duality:finite} and Proposition~\ref{prop:Lipschitz:continuity:DMC}.
\end{proof}
Note that Theorem~\ref{thm:error:bound:capacity} provides an explicit error bound \eqref{eq:EBB}, also called \emph{a priori error}. In addition this theorem gives an approximation to the optimal input distribution \eqref{eq:optInPut}, i.e., the optimizer of the primal problem. Thus, by comparing the values of the primal and the dual optimization problem, one can also compute an \emph{a posteriori error} which is the difference of the dual and the primal problem, namely $F(\hat{\lambda}) + G(\hat{\lambda}) - \I{\hat{p}}{W}$.

\begin{myremark}[Stopping criterion of Algorithm~\hyperlink{algo:1}{1}] \label{remark:stopping}
There are two immediate approaches to define a stopping criterion for Algorithm~\hyperlink{algo:1}{1}.
\begin{enumerate}[(i)]
\item \emph{A priori stopping criterion}: Choose an a priori error $\varepsilon>0$. Setting the right hand side of \eqref{eq:EBB} equal to $\varepsilon$ defines a number of iterations $n_{\varepsilon}$ required to ensure an $\varepsilon$-close solution.
\item \emph{A posteriori stopping criterion}: Choose an a posteriori error $\varepsilon>0$. Choose the smoothing parameter $\nu(n_{\varepsilon})$ for $n_{\varepsilon}$ as defined above in the a priori stopping criterion. Fix a (small) number of iterations $\ell$ that are run using Algorithm~\hyperlink{algo:1}{1}. Compute the a posteriori error $\mathrm{e}_{\ell}:= F(\hat{\lambda}) + G(\hat{\lambda}) - \I{\hat{p}}{\rho}$ according to Theorem~\ref{thm:error:bound:capacity}. If $\mathrm{e}_{\ell}\leq \varepsilon$ terminate the algorithm otherwise continue with another $\ell$ iterations. Continue until the a posteriori error is below $\varepsilon$.
\end{enumerate}
\end{myremark}

\begin{myremark}[Computational stability]\label{rmk:finite:no:input:const}
In the special case of no input cost constraints, one can derive an analytical expression for $G_{\nu}(\lambda)$ and its gradient as
\begin{align}
G_{\nu}(\lambda) &= \nu \log\left( \sum_{i=1}^{N} 2^{\frac{1}{\nu}( \W \lambda\, - \, r)_{i}} \right) -\nu \log(N) \nonumber \\
\nabla G_{\nu}(\lambda) &=  \frac{1}{S(\lambda)}\sum_{i=1}^{N} 2^{\frac{1}{\nu}( \W \lambda\, - \, r)_{i}}\W_{i,\cdot}, \label{eq:gradient:G}
\end{align}
where $S(\lambda):=\sum_{i=1}^{N} 2^{\frac{1}{\nu}( \W \lambda\, - \, r)_{i}}$. In order to achieve an $\varepsilon$-precise solution the smoothing factor $\nu$ has to be chosen in the order of $\varepsilon$, according to Theorem~\ref{thm:error:bound:capacity}. A straightforward computation of $\nabla G_{\nu}(\lambda)$ via \eqref{eq:gradient:G} for a small enough $\nu$ is numerically difficult. In the light of  \cite[p.~148]{nesterov05}, we present a numerically stable technique for computing $\nabla G_{\nu}(\lambda)$. By considering the functions $\R^{M}\ni \lambda \mapsto f(\lambda)=\W\lambda -r \in \R^{N}$ and $\R^{N}\ni x \mapsto R_{\nu}(x)=\nu \log \left( \sum_{i=1}^{N}2^{\tfrac{x_{i}}{\nu}} \right)\in\R$ it is clear that $\nabla_{\lambda} R_{\nu}(f(\lambda))=\nabla G_{\nu}(\lambda)$. The basic idea is to define $\bar{f}(\lambda):=\max_{1\leq i \leq N} f_{i}(\lambda)$ and then consider a function $g:\R^{M}\to \R^{N}$ given by $g_{i}(\lambda)=f_{i}(\lambda)-\bar{f}(\lambda)$, such that all components of $g(\lambda)$ are non-positive. One can show that
\begin{equation*} \label{eq:numerical:stability}
\nabla_{\lambda} R_{\nu}(f(\lambda))=\nabla_{\lambda} R_{\nu}(g(\lambda))+ \nabla \bar{f}(\lambda),
\end{equation*}
where the term on the right-hand side can be computed with a small numerical error.
\end{myremark}

\begin{myremark}[Computational complexity]\label{rmk:finite:complexity}
In case of no input cost constraint, one can see by \eqref{eq:gradient:G} that the computational complexity of a single iteration step of Algorithm~\hyperlink{algo:1}{1} is $O(MN)$. Furthermore, according to \eqref{eq:finite:complexity:thm}, the complexity in terms of number of iterations to achieve an $\varepsilon$-precise solution is $O\left( \tfrac{M\sqrt{\log N}}{\varepsilon} \right)$. This finally gives a computational complexity for finding an additive $\varepsilon$-solution of $O(\tfrac{M^2 N \sqrt{\log N}}{\varepsilon})$. 
Let us point out that that the constants in the computational complexity, explicitly given in \eqref{eq:finite:complexity:thm} and in particular the dependency on the parameter $\gamma$, can have a significant impact on the runtime of the proposed approximation method in practice. In the following remark, however, we presents a way to circumvent ill-conditioned channels with very small (or even vanishing) $\gamma$ parameter.
\end{myremark}

\begin{myremark}[Removing Assumption~\ref{ass:channel}]\label{rmk:perturbation}
The continuity of the channel capacity can be used to remove Assumption~\ref{ass:channel}. Let $\W_1\in\R^{N\times M}$ be an channel transition matrix that does not satisfy Assumption~\ref{ass:channel}, i.e., that contains zero entries. Define a new channel matrix $\W_2\in\R^{N\times M}$ by adding a perturbation $\varepsilon>0$ to all zero entries of $\W_1$ and then normalizing the rows. According to \cite{leung09}
\begin{equation} \label{eq:continuity:smith}
|C(\W_1) - C(\W_2)| \leq 3 \norm{\W_1-\W_2}_{\triangleright} \log(M\vee N) + 2 \eta( \norm{\W_1-\W_2}_{\triangleright} ),
\end{equation}
where $\eta(t)=-t\log t$ and the norm $\norm{\cdot}_{\triangleright}$ on $\R^{N\times M}$ is defined as $\norm{A}_{\triangleright}:=\max_{b\in\Delta_N} \norm{b b\transp A}_{\text{tr}}$.
 Since $\W_2$ by construction satisfies Assumption~\ref{ass:channel}, we can run Algorithm~\hyperlink{algo:1}{1} for channel $\W_2$ and as such get the following upper and lower bounds for the capacity of the singular channel $\W_1$
\begin{align*}
C_{\text{LB}}(\W_1) &:= C_{\text{LB}}(\W_2) - 3 \norm{\W_1-\W_2}_{\triangleright} \log(M\vee N) - 2 \eta( \norm{\W_1-\W_2}_{\triangleright} ) \\
C_{\text{UB}}(\W_1) &:= C_{\text{UB}}(\W_2) + 3 \norm{\W_1-\W_2}_{\triangleright} \log(M\vee N) + 2 \eta( \norm{\W_1-\W_2}_{\triangleright} ).
\end{align*}
See in Example~\ref{ex:BEC} how this perturbation method behaves numerically.
\end{myremark}

\subsection{Simulation Results}
This section presents two examples to illustrate the theoretical results developed in the preceding sections and their performance. All the simulations in this section are performed on a 2.3 GHz Intel Core i7 processor with 8 GB RAM.
\begin{myex} \label{ex:one}
Consider a DMC $W$ having a channel matrix $\W \in \R^{N \times M}$ with $N=10000$ and $M=100$, such that $\W_{ij}=\tfrac{V_{ij}}{\sum_{j=1}^M V_{ij}}$, where $V_{ij}$ is chosen i.i.d. uniformly distributed in $[0,1]$ for all $1\leq i \leq N$ and $1 \leq j \leq M$. The parameter $\gamma$ happens to be $1.0742\cdot 10^{-8}$.
Figure~\ref{fig:ex1} and Table~\ref{tab:ex1} compare the performance of the Blahut-Arimoto algorithm with that of Algorithm~\hyperlink{algo:1}{1}, which has the a priori error bound predicted by Theorem~\ref{thm:error:bound:capacity}, namely
\begin{equation*}
C_{\textnormal{UB}}(W)-C_{\textnormal{LB}}(W) \leq  \frac{2M\sqrt{2\log(N)}}{n+1}\left(\log(\gamma^{-1})\vee\tfrac{1}{\ln 2})\right)+\frac{2 M^{2}}{(n+1)^{2}}\left(\log(\gamma^{-1})\vee\tfrac{1}{\ln 2})\right)^{2},
\end{equation*}
where $n$ denotes the number of iterations and $\gamma$ is equal to the smallest entry in the channel matrix $W$. Recall that the Blahut-Arimoto algorithm has an a priori error bound of the form $C(W)-C_{\textnormal{LB}}(W) \leq  \tfrac{\log(N)}{n}$ \cite[Corollary~1]{arimoto72}. Moreover, the new method provides us with an a posteriori error, which the Blahut-Ariomoto algorithm does not. 
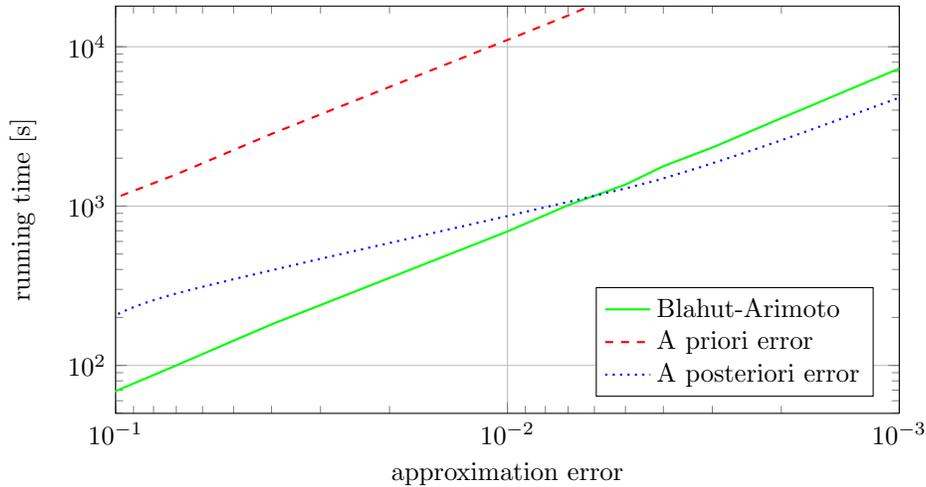
\begin{figure}[!htb]
\centering
  \begin{tikzpicture}
	\begin{axis}[
		height=7cm,
		width=12cm,
		grid=major,
		xlabel=approximation error,
		ylabel=running time \bracket{s},
		xmode=log,
		ymode=log,
		xmin=0.001,
		xmax=0.1,
		ymax=18000,
		ymin=50,
		x dir=reverse,
		legend style={at={(0.790,0.310)},anchor=north,legend cell align=left} 
	]


	\addplot[green,thick] coordinates {
	     (1,7.4)	
		(0.7,10)
		(0.4,18)	
		(0.1,69)	
		(0.07,100)	
		(0.04,181)
		(0.01,693)			
		(0.007,1013)
		(0.005,1364)
		(0.004,1779)	
		(0.003, 2331)	
		(0.002, 3564)
		(0.001,7306)																				
	};
	\addlegendentry{Blahut-Arimoto}

		\addplot[red,thick,dashed] coordinates {
	     (1,114)	
		(0.7,162)
		(0.4,282)	
		(0.1,1127)	
		(0.07,1581)	
		(0.04,2837)
		(0.01,11036)			
		(0.007,15718)	
		(0.004,28456)	
	};
		\addlegendentry{A priori error}	

		\addplot[blue,thick,smooth,dotted] coordinates {
	     (0.1325,114)	
		(0.1154,162)
		(0.0702,282)	
		(0.0063,1127)	
		(0.0037,1581)	
		(0.0018,2837)
		(0.0003976,11036)		
	};
	\addlegendentry{A posteriori error}

	\end{axis}  
\end{tikzpicture}
\label{fig:ex1}

\caption[]{For Example~\ref{ex:one}, this plot depicts the runtime of Algorithm~\hyperlink{algo:1}{1} with respect to the a priori and a posteriori stopping criterion, as explained in Remark~\ref{remark:stopping}. As a reference, the runtime of the Blahut-Arimoto algorithm is shown.}

\end{figure}
 \begin{table}[!htb]
\centering 
\caption{Some specific simulation points of Example~\ref{ex:one}. }
\label{tab:ex1}

\hspace{13mm} Blahut-Arimoto Algorithm \hspace{25mm} Algorithm~\hyperlink{algo:1}{1}
\vspace{3mm} \phantom{..}
  \begin{tabular}{c@{\hskip 3mm} | c@{\hskip 2mm} c@{\hskip 2mm} c@{\hskip 2mm} c  | c@{\hskip 2mm} c@{\hskip 3mm} c@{\hskip 3mm} c  }
 A priori error \hspace{1mm}  & \hspace{1mm}    1  &$0.1$ & $0.01$ & $0.001$ \hspace{1mm}   &\hspace{1mm}  1  &$0.1$ & $0.01$ & $0.001$   \\ 
 $C_{\textnormal{UB}}(W)$ & \hspace{1mm} --- & --- & ---  & --- \hspace{1mm}   &\hspace{1mm}  0.4419 & 0.4131 & 0.4092 & 0.4088   \\
 $C_{\textnormal{LB}}(W)$ & \hspace{1mm} 0.2930 & 0.4008 & 0.4088  & 0.4088 \hspace{1mm}   & \hspace{1mm}   0.3094 & 0.4069 & 0.4088 & 0.4088 \\
  A posteriori error & \hspace{1mm} --- & --- & ---  & --- \hspace{1mm}  & \hspace{1mm}   0.1325 & 0.0063 & 4.0$\cdot 10^{-4}$ & 3.7$\cdot 10^{-5}$ \\
 Time [s] & \hspace{1mm} 7.4  &69 &693  & 7306 \hspace{1mm} & \hspace{1mm}  114 & 1127 & 11\,036 & 110\,987  \\
  Iterations & \hspace{1mm} 14  &133 &1329  & 13\,288 \hspace{1mm} & \hspace{1mm}  27\,797 & 273\,447 & 2\,729\,860 & 27\,294\,000
  \end{tabular}
\end{table}
\end{myex}

%

\begin{myex} \label{ex:BEC}
Consider a binary erasure channel with erasure probability $\alpha$ whose channel transition matrix is given by 
$\W =  \left(\begin{smallmatrix}
1-\alpha & \alpha & 0 \\ 0 & \alpha & 1-\alpha
\end{smallmatrix} \right)$ and as such does not satisfy Assumption~\ref{ass:channel}. We use the perturbation method introduced in Remark~\ref{rmk:perturbation} to approximate its capacity that is analytically known to be $1-\alpha$ \cite[p.~189]{cover}. Table~\ref{tab:BEC} shows the performance of this perturbation method and Algorithm~\hyperlink{algo:1}{1}.

 \begin{table}[!htb]
\centering 
\caption{Some specific simulation points of Example~\ref{ex:BEC} for $\alpha = 0.4$ }
\label{tab:BEC}

\vspace{3mm} \phantom{..}
  \begin{tabular}{c@{\hskip 4mm} | c@{\hskip 4mm} c@{\hskip 4mm} c@{\hskip 4mm} c}
  \hspace{1mm} Perturbation $\varepsilon$  & \hspace{1mm}    $10^{-4}$  &$10^{-5}$ & $10^{-6}$ & $10^{-7}$  \\ 
A priori error & \hspace{1mm} 0.01 & 0.01 & 0.01  & 0.01 \hspace{1mm}      \\
 $C_{\textnormal{UB}}(W)$ & \hspace{1mm} 0.6024 & 0.6003 & 0.6000  & 0.6000 \hspace{1mm}      \\
 $C_{\textnormal{LB}}(W)$ & \hspace{1mm} 0.5949 & 0.5994 & 0.5999  & 0.6000 \hspace{1mm}    \\
  A posteriori error & \hspace{1mm} 0.0075 & $9.2\cdot 10^{-4}$ & $1.1\cdot 10^{-4}$  & $1.2\cdot 10^{-5}$ \hspace{1mm}   \\
 Time [s] & \hspace{1mm} 0.70  & 0.54 & 0.66 & 0.78 \hspace{1mm}  \\
  Iterations & \hspace{1mm} 9056  & 7402 & 8523  & 9896 \hspace{1mm} 
  \end{tabular}
\end{table}
\end{myex}

\section{Channels with Continuous Input and Countable Output Alphabets}   \label{sec:cont:Channels}
	In this section we generalize the approximation scheme introduced in Section \ref{sec:classicalCapacity} to memoryless channels with continuous input and countable output alphabets. The class of discrete-time Poisson channels is an example of such channels with particular interest in applications, for example to model direct detection optical communication systems \cite{moser_phd,shamai90,ref:Chen-13}. Consider $\X \subseteq \R$ as the input alphabet set and $\Y = \N$ as the output alphabet set. The channel is described by the conditional probability $W(i|x) := \Prob{Y = i ~|~ X=x}$. 
	Given a channel $W$ and an integer $M$,  we introduce an $M$-\emph{truncated} version of the channel by 
		\begin{align} \label{W_M} 
			W_M(i|x)&:=\left \lbrace \begin{array}{ll}W(i|x) + \frac{1}{M}\sum\limits_{j \ge M}W(j|x) , & i\in \{0,1,\ldots,M-1\} \\ 0, & i\geq M.
				\end{array} \right.
		\end{align}
	$W_M$ can be seen as a channel with input alphabet $\X$ and output alphabet $\{0,1,\ldots,M-1\}$. Figure~\ref{fig:inf:channel:plot} shows a pictorial representation of a channel and its $M$-truncated counterpart. 
The finiteness of the output alphabet of $W_{M}$ allows us to deploy an approximation scheme similar to the one developed in Section \ref{sec:classicalCapacity} to numerically approximate $C(W_{M})$.  

	 	\begin{figure}[htb!] 
 			  \centering 
 			  \includegraphics[scale = 0.82]{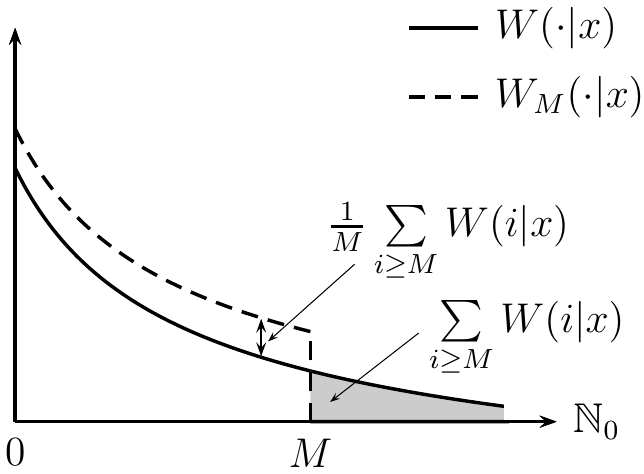} 
 			  \caption{Pictorial representation of the $M$-truncated channel counterpart.} 
 			  \label{fig:inf:channel:plot} 
 	 	\end{figure} 	
The following definition is a key feature of the channel required for the theoretical results developed in this section which, roughly speaking, imposes a certain decay rate for the output distribution uniformly in the input alphabet.
	\begin{mydef} [Polynomial tail]
	\label{def:tail}
		The channel $W$ features a \emph{$k$-ordered polynomial tail} if for $M\in\mathbb{N}_{0}$ and $k\in\Rp$
		\begin{align}\label{R}
			R_k(M) \Let \sum\limits_{i\ge M} \big(\sup_{x \in \X}W(i|x)\big)^k  < \infty.
		\end{align}
	\end{mydef}
	The following assumptions hold throughout this section. 
	\begin{myass} \label{a:channel}  \ 
	\begin{enumerate}[(i)]
		\item The channel $W$ has a $k$-ordered polynomial tail for some $k \in (0,1)$ in the sense of Definition \ref{def:tail}. \label{ass:channel:i}
		\item The mapping $x\mapsto W(i|x)$ is Lipschitz continuous for any $i\in\mathbb{N}_{0}$ with Lipschitz constant $L$.	\label{ass:channel:ii}
	\end{enumerate}
	\end{myass}

Assumption~\ref{a:channel} allows us to relate the capacity of the original channel to that of its truncated counterpart. 
	
	\begin{mythm} \label{thm:C}
	Suppose channel $W$ satisfies Assumption~\ref{a:channel}\eqref{ass:channel:i} with the order $k \in (0,1)$. Then, for any $M \in \N$ and for any probability distribution $p\in\mathcal{P}(\mathcal{X})$ we have
		\begin{align*} 
			\big|  I(p &, W)  - I(p,W_{M}) \big|  \leq \frac{2\log(\e)}{\e(1-k)} \Big[ M^{1-k}\big(R_1(M)\big)^k + R_k(M) \Big],
		\end{align*}
	where $R_k(M)$ is as defined in \eqref{R}. 
	\end{mythm}
\begin{proof}	
See Appendix~\ref{ap:cutting:thm}.
\end{proof}
Note that Theorem~\ref{thm:C} directly implies an upper bound to the capacity since
	\begin{align*}
		|C(W)-C(W_M)| =  \big| \sup\limits_{p\in \meas(\X)} I(p,W)- \sup \limits_{p \in \meas(\X)} I(p,W_{M}) \big| \leq \sup \limits_{p \in \meas(\X)} \big| I(p,W) - I(p,W_{M}) \big|.
	\end{align*} 
We consider two types of input cost constraints: a peak-power constraint $\Prob{X\in \A}=1$ for a compact set $\A\subseteq\mathcal{X}$ and an average-power constraint $\E{s(X)}\leq S$ for $S\in\Rp$ and a continuous function $s$ on $\mathcal{X}$.  
The primal capacity problem for the channel $W_{M}$ is given by
\begin{equation} \label{eq:inf:channel:primal}
C_{\A,S}(W_{M})= \left\{
\begin{array}{lll}
			&\sup\limits_{p} 		& \I{p}{W_{M}} \\
			&\st					& \E{s(X)}\leq S\\
			& 						& p\in \mathcal{P}(\A),
	\end{array} \right.
\end{equation}
where $\mathcal{P}(\A)$ denotes the space of all probability distributions supported on $\A$, cf.~\eqref{eq:cont_DMC_capacity_const}. Our method always requires a peak-power constraint, whereas the average-power constraint is optimal.
The following proposition allows us to restrict the optimization variables from probability distributions to probability densites.
\begin{myprop} \label{lem:density:dense}
The optimization problem \eqref{eq:inf:channel:primal} is equivalent to
\begin{equation*} 
C_{\A,S}(W_{M})= \left\{
\begin{array}{lll}
			&\sup\limits_{p} 		& \I{p}{W_{M}} \\
			&\st					& \E{s(X)}\leq S\\
			& 						& p\in \mathcal{D}(\A),
	\end{array} \right.
\end{equation*}
where $\mathcal{D}(\A)$ is the set of probability densities functions,~i.e., $\mathcal{D}(\A):=\{ f\in\Lp{1}(\A) \, : \, f\geq 0, \, \int_{\A}f(x)\drv x = 1 \}$. 
\end{myprop}
\begin{proof} 
See Appendix~\ref{app:proof:density:restriction}.
\end{proof}
We consider the pair of vector spaces
$(\Lp{1}(\A),\Lp{\infty}(\A))$ together with the bilinear form
\begin{align*}
\inprod{f}{g}:=\int_{\mathcal{X}}f(x)g(x)\drv x.
\end{align*}
In the light of \cite[Theorem~243G]{ref:fremlin-03} this is a dual pair of vector spaces; we refer to \cite[Section~3]{anderson87} for the details of the definition of dual pairs of vector spaces. 
Considering the standard inner product as a bilinear form on the dual pair $(\R^{M},\R^{M})$, we define the linear operator $\WW: \R^{M}\to \Lp{\infty}(\A)$ and its adjoint operator $\WW^{\star}:\Lp{1}(\A)\to\R^{M}$, given by
\begin{align*}
\WW \lambda (x) := \sum_{i=1}^{M}W_{M}(i-1|x)\lambda_{i}, \qquad  (\WW^{\star}p)_{i} := \int_{X}W_{M}(i-1|x)p(x)\drv x.
\end{align*}
Let $S_{\max}:=\inf_{p\in\mathcal{D}(\A)}\{ \inprod{p}{s} \ : \ I(p,W_{M}) = \sup_{q\in\mathcal{D}(\A)} I(q,W_{M})\}$. Following similar lines as in Lemma~\ref{lem:equivalent:primal:problem}, one can deduce that in problem \eqref{eq:inf:channel:primal} the inequality input constraint can be replaced by equality (resp. removed) is $S<S_{\max}$ (resp. $S\geq S_{\max}$). That is, in view of  Proposition~\ref{lem:density:dense}, Lemma~\ref{lem:equivalent:primal:problem} and the discussion there, problem \eqref{eq:inf:channel:primal} (under Assumption~\ref{ass:channel:Poisson}, that we require later) is equivalent to
\begin{equation} \label{opt:primal:equivalent:Poisson}
 	\mathsf{P}: \quad \left\{ \begin{array}{lll}
			&\sup\limits_{p,q} 		&- \inprod{p}{r} + H(q) \\
			&\st					& \WW^{\star}p = q\\
			&					& \inprod{p}{s} = S\\
			& 					& p\in\mathcal{D}(\A), \ q\in\Delta_{M},
	\end{array} \right.
\end{equation}
where $r(\cdot): = -\sum_{j=0}^{M-1}W_{M}(j|\cdot)\log(W_{M}(j|\cdot))$ is an element in $\Lp{\infty}(\A)$ by Assumption~\ref{a:channel}\eqref{ass:channel:ii}. For the rest of the section we restrict attention to \eqref{opt:primal:equivalent:Poisson}, since unconstrained problem can be solved in a similar way.
We call \eqref{opt:primal:equivalent:Poisson} the primal program.
Thanks to the dual vector space framework, the Lagrange dual program of $\mathsf{P}$ is given by
\begin{align} \label{eq:Dual:Program:Poisson}\mathsf{D}: \quad\left\{ \begin{array}{ll}
	\underset{\lambda}{\inf} 		&G(\lambda) + F(\lambda) \\
			\text{s.t. } 					& \lambda\in \R^{M},
	\end{array}\right.
\end{align}
where 
\begin{align*} 
G(\lambda)= \left\{ \begin{array}{ll}
			\underset{p}{\sup} 		& \inprod{p}{\WW \lambda}-\inprod{p}{r} \\
			\text{s.t. } 				& \inprod{p}{s} =  S \\
								& p\in \mathcal{D}(\A)
	\end{array} \right. \text{ and}
	\qquad 
	F(\lambda)= \left\{ \begin{array}{ll}
			\underset{q}{\max} 		&H(q)-\lambda\transp q \\
			\text{s.t. } 				& q\in\Delta_{M}.
	\end{array}\right. 
\end{align*}
\begin{mylem} \label{lem:strong:duality:poisson}
Strong duality holds between \eqref{opt:primal:equivalent:Poisson} and \eqref{eq:Dual:Program:Poisson}.
\end{mylem}
\begin{proof}
Note that the dualized constraint is a linear equality constraint. Therefore the conditions of (1) in \cite[Theorem~5]{mitter08} holds and as such strong duality follows by \cite[Theorem~4]{mitter08}.
\end{proof}
In the remainder of this article we impose the following assumption on the channel.
\begin{myass} \label{ass:channel:Poisson}
$\gamma_{M}:=\min\limits_{y\in\{0,1,\hdots,M-1\}}\min\limits_{x\in\A}W_{M}(y|x)>0$
\end{myass}
In case $\sum_{j\geq M} W(j|x) > 0$ for all $x$, Assumption~\ref{ass:channel:Poisson} holds according to \eqref{W_M} and a lower bound can be given by $\gamma_{M} \geq \tfrac{1}{M} \min_{x} \sum_{j\geq M} W(j|x)$.
Under Assumption~\ref{ass:channel:Poisson} we can show that we can again assume without loss of generality that $\lambda$ takes values in a compact set.
\begin{mylem} \label{lem:compact:set:Poisson}
Under Assumption~\ref{ass:channel:Poisson}, the dual program \eqref{eq:Dual:Program:Poisson} is equivalent to 
\begin{align*}
\left\{ \begin{array}{ll}
			\underset{\lambda}{\min} 		&G(\lambda) + F(\lambda) \\
			\textnormal{s.t. } 				& \lambda\in Q,
	\end{array}\right. 
\end{align*}
where $Q:= \left\{ \lambda\in\R^{M} \ : \ \norm{\lambda}_{2}\leq M \left( \log(\gamma_{M}^{-1}) \vee \tfrac{1}{\ln 2} \right) \right\}$.
\end{mylem}
\begin{proof}
The proof follows the same lines as in the proof of Lemma~\ref{lem:compact:set}.
\end{proof}

Note that $F(\lambda)$ is the same as in Section~\ref{sec:classicalCapacity} and therefore given by \eqref{analytical:solution:F} and its gradient by \eqref{eq:gradient:F}.
As in Section~\ref{sec:classicalCapacity}, we consider the smooth approximation
\begin{align} \label{eq:Gnu:cts}
G_{\nu}(\lambda)= \left\{ \begin{array}{ll}
			\underset{p}{\sup} 		& \inprod{p}{\WW \lambda}-\inprod{p}{r} + \nu \Hdiff{p} - \nu \log(\rho) \\
			\text{s.t. } 				& \inprod{p}{s} =  S \\
				 				& p\in\mathcal{D}(\A),
	\end{array} \right. 
\end{align}
with smoothing parameter $\nu\in\R_{>0}$ and $\rho$ denoting the Lebesgue measure of $\A$. 
To analyze the properties of $G_{\nu}(\lambda)$ we need one more auxiliary lemma. 
\begin{mylem}\label{lem:strong:convexity:cts}
The function $\mathcal{D}(\A)\ni p \mapsto  -\Hdiff{p}+\log(\rho)\in\Rp$ is strongly convex with convexity parameter $\sigma = 1$.
\end{mylem}
\begin{proof}
See Appendix~\ref{app:proof:strong:convexity:poisson}.
\end{proof}
Furthermore, we can show that the uniform approximation $G_{\nu}(\lambda)$ is smooth and has a Lipschitz continuous gradient, with known constant.
The following result is a generalization of Proposition~\ref{prop:Lipschitz:continuity:DMC}.
\begin{myprop} \label{prop:Lipschitz:cts:channel}
$G_{\nu}(\lambda)$ is well defined and continuously differentiable at any $\lambda\in \R^{M}$. Moreover, this function is convex and its gradient $\nabla G_{\nu}(\lambda)=\WW^{\star} p_{\nu}^{\lambda}$ is Lipschitz continuous with constant $L_{\nu} =\tfrac{1}{\nu}$.
\end{myprop}
\begin{proof}
See Appendix~\ref{ap:lipschitz:cts:gradient:cts}.
\end{proof}
We denote by $p_{\nu}^{\lambda}$ the optimizer to \eqref{eq:Gnu:cts}, that is unique since the objective function is strictly concave. 
To analyze the solution to \eqref{eq:Gnu:cts} we consider the following optimization problem, that, if feasible, has a closed form solution
\begin{equation} \label{opt:cover:cont}
 	\left\{ \begin{array}{lll}
			&\underset{p}{\sup} 		&\Hdiff{p} + \inprod{p}{c} \\
			&\text{s.t. } 			& \inprod{p}{s} = S\\
			& 					& p\in\mathcal{D}(\A),
	\end{array} \right.
\end{equation}
with $c,s\in \Lp{\infty}(\A)$.
\begin{mylem} \label{lem:cover:cont}
Let $p^{\star}(x)=2^{\mu_1 + c(x) + \mu_2 s(x)}$, where $\mu_1, \mu_2\in\R$ are chosen such that $p^\star$ satisfies the constraints in \eqref{opt:cover:cont}. Then $p^\star$ uniquely solves \eqref{opt:cover:cont}.
\end{mylem}
The proof directly follows from \cite[p.~409]{cover} and the proof of Lemma~\ref{lem:cover}. Hence, $G_{\nu}(\lambda)$ has a (unique) analytical optimizer
\begin{align} \label{eq:finite:optimizer:pmu:cts}
p_{\nu}^{\lambda}(x,\mu) = 2^{\mu_{1} + \tfrac{1}{\nu}\left( \WW \lambda(x) - r(x)\right) + \mu_{2}s(x)}, \quad x\in \mathcal{X},
\end{align}
where $\mu_{1},\mu_{2}\in\R$ have to be chosen such that $\inprod{p_{\nu}^{\lambda}(\cdot,\mu)}{s}=S$ and $p_{\nu}^{\lambda}(\cdot,\mu)\in\mathcal{D}(\A)$; for this choice of $\mu_1$, $\mu_2$ we denote the solution by $p_{\nu}^{\lambda}(\cdot)$. 
\begin{myremark}[No input constraints]\label{rmk:stabilization:optimizer:cts}
In case of no input constraints, the unique optimizer to \eqref{eq:Gnu:cts} is given by
\begin{equation*}
p_{\nu}^{\lambda}(x) = \frac{2^{ \tfrac{1}{\nu} (\WW\lambda(x) -  r(x))} }{\int_{\A} 2^{ \tfrac{1}{\nu} (\WW\lambda(x) -  r(x))} \drv x},
\end{equation*}
whose numerical evaluation can be done in a stable way by following Remark~\ref{rmk:finite:no:input:const}.
\end{myremark}
\begin{myremark}[Additional input constraints]\label{rmk:finite:constraint:optimizer:cts}
As in Remark~\ref{rmk:finite:constraint:optimizer}, in case of additional input constraints we need an efficient method to find the coefficients $\mu_{i}$ in \eqref{eq:finite:optimizer:pmu:cts}. This problem can again be reduced to a finite dimensional convex optimization problem (\cite[Theorem~4.8]{ref:Borwein-91},\cite[p.~257 ff.]{ref:Lasserre-11}), in the sense that the coefficients $\mu_i$ are the unique maximizers to
\begin{equation} \label{eq:opt:problem:find:mu:cts}
\max\limits_{\mu\in\R^{2}}\left\{ y\transp \mu -\int_{\A}p_{\nu}^{\lambda}(x,\mu)\drv x \right\},
\end{equation}
where $y:=(1,S)$. Note that $\eqref{eq:opt:problem:find:mu:cts}$ is an unconstrained maximization of a striclty concave function. The evalutation of the gradient and the Hessian of this objective function involves computing moments of the measure $p_{\nu}^{\lambda}(x,\mu)\drv x$, which unlike to the finite input alphabet case (Remark~\ref{rmk:finite:constraint:optimizer}) is numerically difficult. In \cite[p.~259 ff.]{ref:Lasserre-11}, an efficient approximation of the mentioned gradient and Hessian in terms of two single semidefinite programs involving two linear matrix inequalities (LMI) is presented, where the desired accuracy is controlled by the size of the LMI constraints. As mentioned in Remark~\ref{rmk:finite:constraint:optimizer}, this will provide a suboptimal solution to the maximum entropy problem \eqref{eq:Gnu:cts} and as such the error bounds of Theorem~\ref{thm:error:bound:capacity:continuous:channel} do not hold. By following \cite{ref:Devolver-13}, however, one can quantify the approximation error of Algorithm~\hyperlink{algo:1}{1} in case of an inexact gradient. We also refer the interested reader to \cite{DavidSutter14}, for a related work on channel capacity approximation under inexact first-order information.
\end{myremark}

Note that the differential entropy $\Hdiff{p}\leq \log(\rho) $ for all $p\in\mathcal{D}(\A)$ and that there exists a function $\iota:\Rsp\to\Rp$ such that
 \begin{equation} \label{eq:uniform:bound:cts}
 G_{\nu}(\lambda)\leq G(\lambda)\leq G_{\nu}(\lambda) + \iota(\nu) \ \text{for all }\lambda\in Q,
 \end{equation}
 i.e., $G_{\nu}(\lambda)$ is a uniform approximation of the non-smooth function $G(\lambda)$.  
The following lemma, Lemma~\ref{lem:iota:new}, provides an explicit expression for the function $\iota$ in \eqref{eq:uniform:bound:cts} under some Lipschitz continuity assumptions, implying in particular that $\iota(\nu)\to 0$ as $\nu\to 0$. 
\begin{mylem}\label{lem:Lf}
Under Assumption~\ref{a:channel}\eqref{ass:channel:ii} and Assumption~\ref{ass:channel:Poisson} the function $f_{\lambda}(\cdot):= \WW\lambda(\cdot) -  r(\cdot)$ is Lipschitz continuous uniformly in $\lambda\in Q$ with constant $L_{f}=LM^2 (\log \tfrac{1}{\gamma_M} \vee \tfrac{1}{\ln 2})+M L |\log \tfrac{1}{\gamma_M} - \frac{1}{\ln 2} |$.
\end{mylem}
\begin{proof}
See Appendix~\ref{ap:proof:Lf}.
\end{proof}
\begin{myass}[Lipschitz continuity of the average-power constraint function] \label{a:constraint_fct:s}	
The average-power constraint function $s(\cdot)$ is Lipschitz continuous with constant $L_{s}$. 
\end{myass}
\begin{mylem} \label{lem:iota:new}
Under Assumptions~\ref{a:channel}\eqref{ass:channel:ii}, \ref{ass:channel:Poisson} and \ref{a:constraint_fct:s} a possible choice of the function $\iota$ in \eqref{eq:uniform:bound:cts} is given by
\begin{equation*}
\iota(\nu) = \left\{ 
  \begin{array}{l l}
    \nu \left( \log\left( \frac{T_{1}}{\nu}+T_{2} \right) +1 \right), & \quad \nu<\tfrac{T_{1}}{1-T_{2}} \text{ or } T_{2}>1\\
    \nu, & \quad \text{otherwise},
  \end{array} \right.
\end{equation*}
where $T_{1}\!:=L_{f}\rho + 2L_{f}L_{s}\rho^{2}\!\left( \tfrac{1}{-\underline{s}} \vee \tfrac{1}{\overline{s}} \right)$, $T_{2}\!:= L_{s}\rho (\underline{\mu} \, \vee \, \overline{\mu})$, $\underline{\mu}\!:=\tfrac{2}{-\underline{s}}\log\!\left( \tfrac{2L_{s}\rho}{-\underline{s}}\vee 1 \right)$, $\overline{\mu}\!:=\tfrac{2}{\overline{s}}\log\!\left( \tfrac{2L_{s}\rho}{\overline{s}}\vee 1 \right)$, $\rho:=\int_{\A}\drv x$, $\underline{s} := -S + \min_{x\in\A} s(x)$ and $\overline{s} := -S + \max_{x\in\A} s(x)$.
\end{mylem}
\begin{proof} 
See Appendix~\ref{ap:proof:iota}. 
\end{proof}
We consider the smooth, finite dimensional, convex optimization problem
\begin{align}\label{Lagrange:Dual:Program:smooth:poisson}
 \mathsf{D}_{\nu}: \left\{ \begin{array}{ll}
	\underset{\lambda}{\inf} 		& F(\lambda) + G_{\nu}(\lambda) \\
			\text{s.t. } 					& \lambda\in Q,
	\end{array}\right.
\end{align}
whose solution can be approximated with Algorithm~\hyperlink{algo:1}{1} presented in Section~\ref{sec:classicalCapacity}, as follows. Define the constant $D_{1}:=\tfrac{1}{2}(M \log(\gamma^{-1})\vee\tfrac{1}{\ln 2})^2$.
\begin{mythm}\label{thm:error:bound:capacity:continuous:channel}
Under Assumptions~\ref{a:channel}\eqref{ass:channel:ii}, \ref{ass:channel:Poisson} and \ref{a:constraint_fct:s}, let $\alpha := 2(T_1 + T_2 + 1)$ where $T_1$ and $T_2$ are as defined in Lemma~\ref{lem:iota:new}. Given precision $\varepsilon \in (0, \tfrac{\alpha}{4})$, we set the smoothing parameter $\nu = \tfrac{\varepsilon / \alpha}{\log \left( \alpha / \varepsilon \right)}$ and number of iterations $n\geq \tfrac{1}{\varepsilon} \sqrt{8 D_{1}\alpha}\sqrt{\log(\varepsilon^{-1}) + \log(\alpha) + \tfrac{1}{4}}$.  Consider 
\begin{align} 
\hat{\lambda} = y_{n} \in Q\qquad \text{and} \qquad \hat{p}=\sum_{k=0}^{n}\frac{2(i+1)}{(n+1)(n+2)} p_{\nu}^{x_{k}}\in \mathcal{D}(\A), \label{eq:optInPut:cts}
\end{align}
where $y_k$ computed at the $\text{k}^{\text{th}}$ iteration of Algorithm~\hyperlink{algo:1}{1} and $p_{\nu}^{x_{k}}$ is the analytical solution in \eqref{eq:finite:optimizer:pmu:cts}. Then, $\hat{\lambda}$ and $\hat{p}$ are the approximate solutions to the problems \eqref{eq:Dual:Program:Poisson} and \eqref{eq:inf:channel:primal}, i.e., 
\begin{align}
0\leq F(\hat{\lambda}) + G(\hat{\lambda}) - \I{\hat{p}}{W_{M}} \leq \varepsilon. \label{eq:EBB:cts}
\end{align}
Therefore, Algorithm~\hyperlink{algo:1}{1} requires $O\left( \tfrac{1}{\varepsilon}\sqrt{\log\left( \varepsilon^{-1} \right)} \right)$ iterations to find an $\varepsilon$-solution to the problems \eqref{eq:Dual:Program:Poisson} and \eqref{eq:inf:channel:primal}.
\end{mythm}
\begin{proof}
See Appendix~\ref{ap:thm:error:bound:capacity:continuous:channel}.
\end{proof}
Hence, under Assumption~\ref{a:constraint_fct:s} we can quantify the approximation error of the presented method to find the capacity of any channel $W$, satisfying Assumptions~\ref{a:channel} and \ref{ass:channel:Poisson}, by
\begin{align*}
\left| C(W) -  C_{\text{approx}}^{(n)}(W_{M}) \right| 	&\leq   \underbrace{\left| C(W) -  C(W_{M}) \right|}_{(\star)} + \underbrace{\left| C(W_{M}) -  C_{\text{approx}}^{(n)}(W_{M}) \right|}_{(\star \star)},
\end{align*}
where $(\star)$ and $(\star \star)$ are addressed by Theorem~\ref{thm:C} and Theorem~\ref{thm:error:bound:capacity:continuous:channel}, respectively. Let us highlight that for the term $(\star \star)$ we have two different quantitative bounds: First, the \textit{a priori} bound $\varepsilon$ for which Theorem~\ref{thm:error:bound:capacity:continuous:channel} prescribes a lower bound for the required number of iterations; second, the \textit{a posteriori} bound $F(\hat{\lambda}) + G(\hat{\lambda}) - I (\hat p, W_{M})$
which can be computed after a number of iterations have been executed. In practice, the a posteriori bound often approaches $\varepsilon$ much faster than the a priori bound. Note also that 
by \eqref{eq:uniform:bound:cts} and Theorem~\ref{thm:error:bound:capacity:continuous:channel}
\begin{align*}		
0\leq F(\hat{\lambda}) + G_{\nu}(\hat{\lambda}) +\iota(\nu) - \I{\hat{p}}{W_{M}} \leq \iota(\nu) + \varepsilon,
\end{align*}
which shows that $ F(\hat{\lambda}) + G_{\nu}(\hat{\lambda}) +\iota(\nu)$ is an upper bound for the channel capacity with a priori error $\iota(\nu) + \varepsilon$. This bound can be particularly helpful in cases where an evaluation of $G(\lambda)$ for a given $\lambda$ is hard.
\begin{myremark}[Optimal tail truncation] \label{rmk:optimal:truncation}
Given a fixed number of iterations, the term $(\star\star)$ above is effected by the truncation level $M$ for two reasons: the higher $M$ the larger the size of the output as well as the lower the parameter $\gamma_{M}$. Therefore, term $(\star\star)$ increases as M increases, which can be quantified by \eqref{eq:proof:rate:error:term}. On the other hand, term $(\star)$ obviously has the opposite behavior. Namely, the higher M leads to the better approximation of the channel W by the truncated version $W_M$ as quantified in Theorem~\ref{thm:C}. Hence, given a channel $W$ with the polynomial tail order $k$, there is an optimal value for the truncation parameter $M$, which thanks to the monotonicity explained above can be effectively computed in practice by techniques such as bisection.   \\
Note, that this truncation procedure could also be applied to a finite output alphabet, given that the channel satisfies Assumption~\ref{a:channel}\eqref{ass:channel:i}, and for example improve the performance of the method presented in Section~\ref{sec:classicalCapacity}.
\end{myremark}

\begin{myremark}[Without average-power constraint] \
In case of considering only a peak-power constraint and no average-power constraint, our proposed methodology allows us to access a closed form expression for $G_{\nu}(\lambda)$ and its gradient, 
\begin{align} 
G_{\nu}(\lambda) 	&= 	\nu \log\left( \int_{\A} 2^{\frac{1}{\nu}\left( \WW\lambda(x) -  r(x) \right)} \drv x \right) - \nu \log(\rho) \label{eq:poiss:closed:form:Gnu} \\
\nabla G_{\nu}(\lambda) 	&= 	\frac{ \int_{\A} 2^{\frac{1}{\nu}\left( \WW\lambda(x) -  r(x) \right)}W_{M}(\cdot|x) \drv x}{\int_{\A} 2^{\frac{1}{\nu}\left( \WW\lambda(x) -  r(x) \right)} \drv x}. \nonumber
\end{align}
\end{myremark}

 \subsection*{Discrete-Time Poisson Channel} 
 The discrete-time Poisson channel is a mapping from $\Rp$ to $\mathbb{N}_0$, such that conditioned on the input $x\geq 0$ the output is Poisson distributed with mean $x + \eta$, i.e.,
\begin{equation}
W(y|x)=e^{-(x+\eta)}\frac{(x+\eta)^y}{y!}, \quad y\in \mathbb{N}_0,\, x\in \Rp,\label{eq:WpoissonIntro}
\end{equation}
where $\eta\geq0$ denotes a constant sometimes referred to as \emph{dark current}. A peak-power constraint on the transmitter is given by the peak-input constraint $X\leq A$ with probability one, i.e., $\A=[0,A]$ and an average-power constraint on the transmitter is considered by $\E{X}\leq S$.

Up to now, no analytic expression for the capacity of a discrete-time Poisson channel is known. However, for different scenarios lower and upper bounds exist. Brady and Verd\'u derived a lower and upper bound in the presence of only an average-power constraint \cite{brady90}. Later, for $\eta=0$ and only an average-power constraint, Martinez introduced better upper and lower bounds \cite{martinez07}.
Lapidoth and Moser derived a lower bound and an asymptotic upper bound, which is valid only when the available peak and average power tend to infinity with their ratio held fixed, for the presence of a peak and average-power constraint \cite{lapidoth09}. Lapidoth \emph{et al.} computed the asymptotic capacity of the discrete-time Poisson channel when the allowed average-input power tends to zero with the allowed peak power --- if finite --- held fixed and the dark current is constant or tends to zero proportionally to the average power \cite{lapidoth11}. 

In \cite{ref:Chen-14-1} a numerical algorithm is presented, where the Blahut-Ariomoto algorithm is incorporated into the deterministic annealing method, that allows the computation of both the channel capacity under peak and average power constraints and its associated optimal input distribution. Furthermore, the works \cite{ref:Chen-14-1, ref:Chen-14-2} derive several fundamental properties of capacity achieving input distributions for the discrete-time Poisson channel. 

Here, we numerically approximate the capacity of a discrete-time Poisson channel using the proposed algorithm. For simplicity, we consider the case where only a peak power constraint is imposed; the case where an additional average power constraint is present can be treated similarly. It was shown in \cite{shamai90} that in the case of a peak power constraint (with or without average power constraint), the capacity achieving input distribution is discrete. This, in the limit as the number of iterations in the proposed approximation method goes to infinity, is consistent with the optimal input distribution given in Remark~\ref{rmk:stabilization:optimizer:cts}.

The following proposition provides an upper bound for the $k$-polynomial tail for the Poisson channel $W$ as defined in \eqref{eq:WpoissonIntro}.
	
	\begin{myprop}[Poisson tail]
	\label{prop:poisson_tail}
		The Poisson channel \eqref{eq:WpoissonIntro} having a bounded input alphabet $\X = [0,A]$ and \emph{dark current} parameter $\eta$ has a $k$-polynomial tail for any $k \in (0,1]$ in the sense of Definition \ref{def:tail}, which is upper bounded for all $M \ge A + \eta$ by
		\begin{align*}
			R_k(M) \le \Big( {\alpha\e^{(\alpha-1)(A+\eta)}\frac{(A+\eta)^{M}}{M!}} \Big)^k, \qquad \alpha \Let 2^{(k^{-1}-1)}.
		\end{align*}
	\end{myprop}
\begin{proof}
See Appendix~\ref{ap:cutting:poisson}.
\end{proof}
In the following we present an example to illustrate the theoretical results developed in the preceding sections and their performance. Note that for the discrete-time Poisson channel Assumption~\ref{ass:channel:Poisson} clearly holds.
\begin{myex} \label{ex:poisson}
We consider a discrete-time Poisson channel $W$ as defined in \eqref{eq:WpoissonIntro} with a peak-power constraint $A$ and dark current $\eta = 1$. Up to now, the best known lower bound for the capacity is given by \cite[Theorem~4]{lapidoth09}
\begin{equation}
C(W) \geq \frac{1}{\ln 2}\left(\frac{1}{2} \ln A  + \left(\frac{A}{3}+1 \right)\ln \left(1 + \frac{3}{A} \right)-1 -\sqrt{\frac{\eta+\tfrac{1}{12}}{A}}\left(\frac{\pi}{4} + \frac{1}{2}\ln 2 \right) - \frac{1}{2}\ln \frac{\pi e}{2}\right)\!. \label{eq:MoserLB}
\end{equation}
To the best of our knowledge no upper bound for the capacity is known. In \cite{lapidoth09} an asymptotic upper bound is given which includes an unknown error term that is vanishing in the limit $A \to \infty$.
According to Theorems~\ref{thm:C} and \ref{thm:error:bound:capacity:continuous:channel}, the algorithm introduced in this article leads to an approximation error after $n$ iterations that is given by
\begin{align*}
\left|C_{\text{approx}}^{(n)}(W_{M})-C(W) \right| &\leq \left|C_{\text{approx}}^{(n)}(W_{M})-C(W_M) \right| + \left|C(W_M)-C(W) \right| \nonumber\\
&\leq F(\hat{\lambda}) + G(\hat{\lambda}) - \I{\hat{p}}{W} + \mathcal{E} , \label{eq:error:poisson:example}
\end{align*}
where $\mathcal{E} = \frac{2\log(\e)}{\e(1-k)} \Big[ M^{1-k}\big(R_1(M)\big)^k + R_k(M) \Big]$, $R_\ell(M) = \Big( {\alpha\e^{(\alpha-1)(A+\eta)}\frac{(A+\eta)^{M}}{M!}} \Big)^\ell$ and $\alpha \Let 2^{(\ell^{-1}-1)}$ for any $k\in(0,1)$ and $\ell\in(0,1]$. The truncation parameter $M$ was determined as described in Remark~\ref{rmk:optimal:truncation}. This finally leads to the following upper and lower bounds on $C(W)$
\begin{equation} \label{eq:upper:lower:Cp}
	2 \I{\hat{p}}{W}  -  \left( F(\hat{\lambda}) + G(\hat{\lambda}) \right) - \mathcal{E} \leq C(W) \leq 2  \left( F(\hat{\lambda}) + G(\hat{\lambda})  \right) - \I{\hat{p}}{W} +  \mathcal{E}. 
\end{equation}
Figure~\ref{fig:poissonPlott} compares the two bounds \eqref{eq:MoserLB} and \eqref{eq:upper:lower:Cp} for different values of $A$. Further details on the simulation can be found in Appendix~\ref{app:simulation:details}.

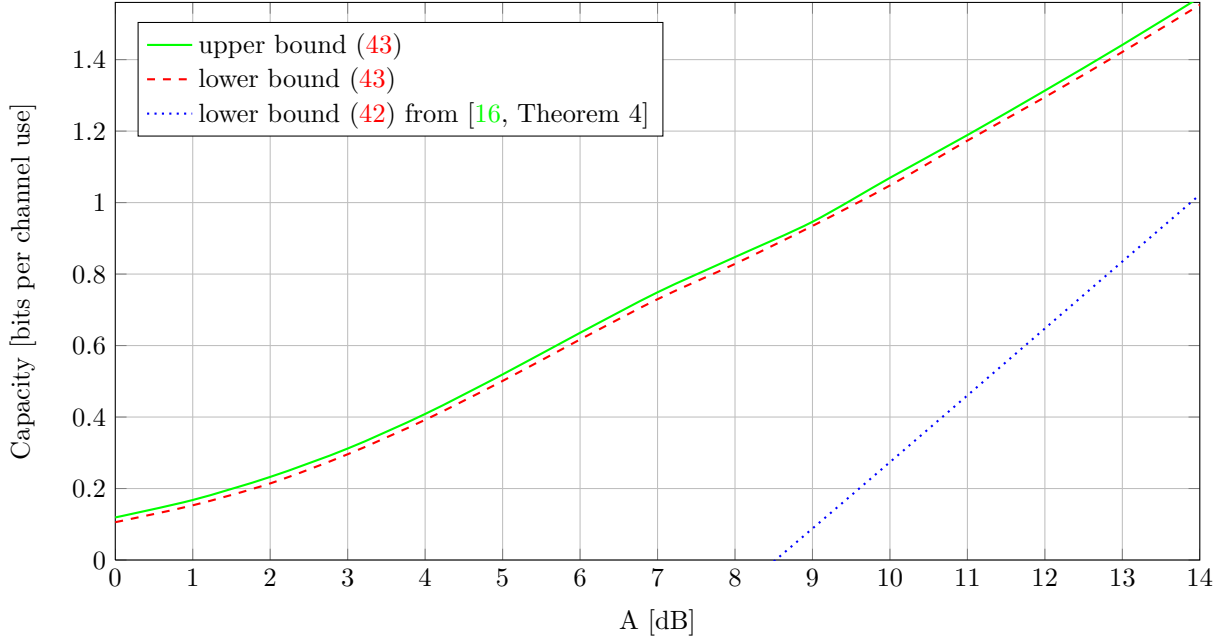
\begin{figure}[!htb]
  \begin{tikzpicture}
	\begin{axis}[
		height=9cm,
		width=16cm,
		grid=major,
		xlabel=A \bracket{dB},
		ylabel=Capacity \bracket{bits per channel use},
		xmin=0,
		xmax=14,
		ymax=1.56,
		ymin=0,
		legend style={at={(0.263,0.965)},anchor=north,legend cell align=left} 
	]


	\addplot[green,thick,smooth] coordinates {
	     (0,0.1191)	
		(1,0.1679)
		(2,0.2325)
		(3,0.3119)
		(4,0.4086)
		(5,0.5192)	
		(6,0.6357)	
		(7,0.7490)	
		(8,0.8476)	
		(9,0.9463)	
		(10,1.0693)
		(11,1.1889)
		(12, 1.3133)	
		(13, 1.4412)
		(14, 1.5743)		
		
		
	};
	\addlegendentry{upper bound \eqref{eq:upper:lower:Cp}}
		
		\addplot[red,thick,smooth,dashed] coordinates {
		(0,0.1058)	
		(1,0.1530)
		(2,0.2143)
		(3,0.2959)
		(4,0.3922)
		(5,0.5010)	
		(6,0.6170)	
		(7,0.7297)	
		(8,0.8285)	
		(9,0.9347)	
		(10,1.0479)		
		(11,1.1734)
		(12, 1.2950)
		(13, 1.4214)
		(14, 1.5532)
		
	};
		\addlegendentry{lower bound \eqref{eq:upper:lower:Cp}}

		\addplot[blue,thick,smooth,dotted] coordinates {
		(8.5,-0.00428557)(8.51,-0.00244138)(8.52,-0.000597025)(8.53,0.00124748)(8.54,0.00309215)(8.55,0.00493697)(8.56,0.00678195)(8.57,0.00862708)(8.58,0.0104724)(8.59,0.0123178)(8.6,0.0141634)(8.61,0.0160092)(8.62,0.0178551)(8.63,0.0197011)(8.64,0.0215473)(8.65,0.0233937)(8.66,0.0252402)(8.67,0.0270869)(8.68,0.0289337)(8.69,0.0307807)(8.7,0.0326278)(8.71,0.0344751)(8.72,0.0363225)(8.73,0.0381701)(8.74,0.0400178)(8.75,0.0418656)(8.76,0.0437137)(8.77,0.0455618)(8.78,0.0474101)(8.79,0.0492586)(8.8,0.0511072)(8.81,0.0529559)(8.82,0.0548048)(8.83,0.0566539)(8.84,0.0585031)(8.85,0.0603524)(8.86,0.0622019)(8.87,0.0640515)(8.88,0.0659012)(8.89,0.0677511)(8.9,0.0696012)(8.91,0.0714514)(8.92,0.0733017)(8.93,0.0751521)(8.94,0.0770028)(8.95,0.0788535)(8.96,0.0807044)(8.97,0.0825554)(8.98,0.0844066)(8.99,0.0862579)(9.,0.0881093)(9.01,0.0899609)(9.02,0.0918126)(9.03,0.0936645)(9.04,0.0955165)(9.05,0.0973686)(9.06,0.0992208)(9.07,0.101073)(9.08,0.102926)(9.09,0.104778)(9.1,0.106631)(9.11,0.108484)(9.12,0.110337)(9.13,0.11219)(9.14,0.114044)(9.15,0.115897)(9.16,0.117751)(9.17,0.119604)(9.18,0.121458)(9.19,0.123312)(9.2,0.125166)(9.21,0.127021)(9.22,0.128875)(9.23,0.13073)(9.24,0.132584)(9.25,0.134439)(9.26,0.136294)(9.27,0.138149)(9.28,0.140004)(9.29,0.141859)(9.3,0.143715)(9.31,0.14557)(9.32,0.147426)(9.33,0.149281)(9.34,0.151137)(9.35,0.152993)(9.36,0.154849)(9.37,0.156706)(9.38,0.158562)(9.39,0.160418)(9.4,0.162275)(9.41,0.164132)(9.42,0.165989)(9.43,0.167846)(9.44,0.169703)(9.45,0.17156)(9.46,0.173417)(9.47,0.175275)(9.48,0.177132)(9.49,0.17899)(9.5,0.180848)(9.51,0.182705)(9.52,0.184564)(9.53,0.186422)(9.54,0.18828)(9.55,0.190138)(9.56,0.191997)(9.57,0.193855)(9.58,0.195714)(9.59,0.197573)(9.6,0.199432)(9.61,0.201291)(9.62,0.20315)(9.63,0.205009)(9.64,0.206868)(9.65,0.208728)(9.66,0.210588)(9.67,0.212447)(9.68,0.214307)(9.69,0.216167)(9.7,0.218027)(9.71,0.219887)(9.72,0.221747)(9.73,0.223608)(9.74,0.225468)(9.75,0.227329)(9.76,0.229189)(9.77,0.23105)(9.78,0.232911)(9.79,0.234772)(9.8,0.236633)(9.81,0.238494)(9.82,0.240355)(9.83,0.242217)(9.84,0.244078)(9.85,0.24594)(9.86,0.247801)(9.87,0.249663)(9.88,0.251525)(9.89,0.253387)(9.9,0.255249)(9.91,0.257111)(9.92,0.258973)(9.93,0.260836)(9.94,0.262698)(9.95,0.264561)(9.96,0.266423)(9.97,0.268286)(9.98,0.270149)(9.99,0.272012)(10.,0.273875)(10.01,0.275738)(10.02,0.277601)(10.03,0.279464)(10.04,0.281328)(10.05,0.283191)(10.06,0.285055)(10.07,0.286918)(10.08,0.288782)(10.09,0.290646)(10.1,0.29251)(10.11,0.294374)(10.12,0.296238)(10.13,0.298102)(10.14,0.299966)(10.15,0.301831)(10.16,0.303695)(10.17,0.30556)(10.18,0.307424)(10.19,0.309289)(10.2,0.311154)(10.21,0.313019)(10.22,0.314884)(10.23,0.316749)(10.24,0.318614)(10.25,0.320479)(10.26,0.322344)(10.27,0.32421)(10.28,0.326075)(10.29,0.327941)(10.3,0.329806)(10.31,0.331672)(10.32,0.333538)(10.33,0.335404)(10.34,0.33727)(10.35,0.339136)(10.36,0.341002)(10.37,0.342868)(10.38,0.344734)(10.39,0.3466)(10.4,0.348467)(10.41,0.350333)(10.42,0.3522)(10.43,0.354066)(10.44,0.355933)(10.45,0.3578)(10.46,0.359667)(10.47,0.361533)(10.48,0.3634)(10.49,0.365267)(10.5,0.367135)(10.51,0.369002)(10.52,0.370869)(10.53,0.372736)(10.54,0.374604)(10.55,0.376471)(10.56,0.378339)(10.57,0.380206)(10.58,0.382074)(10.59,0.383942)(10.6,0.38581)(10.61,0.387677)(10.62,0.389545)(10.63,0.391413)(10.64,0.393281)(10.65,0.39515)(10.66,0.397018)(10.67,0.398886)(10.68,0.400754)(10.69,0.402623)(10.7,0.404491)(10.71,0.40636)(10.72,0.408228)(10.73,0.410097)(10.74,0.411966)(10.75,0.413834)(10.76,0.415703)(10.77,0.417572)(10.78,0.419441)(10.79,0.42131)(10.8,0.423179)(10.81,0.425048)(10.82,0.426917)(10.83,0.428786)(10.84,0.430656)(10.85,0.432525)(10.86,0.434394)(10.87,0.436264)(10.88,0.438133)(10.89,0.440003)(10.9,0.441872)(10.91,0.443742)(10.92,0.445612)(10.93,0.447481)(10.94,0.449351)(10.95,0.451221)(10.96,0.453091)(10.97,0.454961)(10.98,0.456831)(10.99,0.458701)(11.,0.460571)(11.01,0.462441)(11.02,0.464312)(11.03,0.466182)(11.04,0.468052)(11.05,0.469922)(11.06,0.471793)(11.07,0.473663)(11.08,0.475534)(11.09,0.477404)(11.1,0.479275)(11.11,0.481145)(11.12,0.483016)(11.13,0.484887)(11.14,0.486758)(11.15,0.488628)(11.16,0.490499)(11.17,0.49237)(11.18,0.494241)(11.19,0.496112)(11.2,0.497983)(11.21,0.499854)(11.22,0.501725)(11.23,0.503596)(11.24,0.505467)(11.25,0.507338)(11.26,0.50921)(11.27,0.511081)(11.28,0.512952)(11.29,0.514824)(11.3,0.516695)(11.31,0.518566)(11.32,0.520438)(11.33,0.522309)(11.34,0.524181)(11.35,0.526052)(11.36,0.527924)(11.37,0.529796)(11.38,0.531667)(11.39,0.533539)(11.4,0.535411)(11.41,0.537282)(11.42,0.539154)(11.43,0.541026)(11.44,0.542898)(11.45,0.54477)(11.46,0.546642)(11.47,0.548514)(11.48,0.550386)(11.49,0.552258)(11.5,0.55413)(11.51,0.556002)(11.52,0.557874)(11.53,0.559746)(11.54,0.561618)(11.55,0.56349)(11.56,0.565362)(11.57,0.567234)(11.58,0.569107)(11.59,0.570979)(11.6,0.572851)(11.61,0.574723)(11.62,0.576596)(11.63,0.578468)(11.64,0.58034)(11.65,0.582213)(11.66,0.584085)(11.67,0.585958)(11.68,0.58783)(11.69,0.589703)(11.7,0.591575)(11.71,0.593448)(11.72,0.59532)(11.73,0.597193)(11.74,0.599065)(11.75,0.600938)(11.76,0.60281)(11.77,0.604683)(11.78,0.606556)(11.79,0.608428)(11.8,0.610301)(11.81,0.612174)(11.82,0.614046)(11.83,0.615919)(11.84,0.617792)(11.85,0.619665)(11.86,0.621537)(11.87,0.62341)(11.88,0.625283)(11.89,0.627156)(11.9,0.629028)(11.91,0.630901)(11.92,0.632774)(11.93,0.634647)(11.94,0.63652)(11.95,0.638393)(11.96,0.640266)(11.97,0.642138)(11.98,0.644011)(11.99,0.645884)(12.,0.647757)(12.01,0.64963)(12.02,0.651503)(12.03,0.653376)(12.04,0.655249)(12.05,0.657122)(12.06,0.658995)(12.07,0.660868)(12.08,0.662741)(12.09,0.664613)(12.1,0.666486)(12.11,0.668359)(12.12,0.670232)(12.13,0.672105)(12.14,0.673978)(12.15,0.675851)(12.16,0.677724)(12.17,0.679597)(12.18,0.68147)(12.19,0.683343)(12.2,0.685216)(12.21,0.687089)(12.22,0.688962)(12.23,0.690835)(12.24,0.692708)(12.25,0.694581)(12.26,0.696454)(12.27,0.698327)(12.28,0.7002)(12.29,0.702073)(12.3,0.703946)(12.31,0.705819)(12.32,0.707692)(12.33,0.709565)(12.34,0.711438)(12.35,0.713311)(12.36,0.715184)(12.37,0.717057)(12.38,0.718929)(12.39,0.720802)(12.4,0.722675)(12.41,0.724548)(12.42,0.726421)(12.43,0.728294)(12.44,0.730167)(12.45,0.73204)(12.46,0.733913)(12.47,0.735785)(12.48,0.737658)(12.49,0.739531)(12.5,0.741404)(12.51,0.743277)(12.52,0.74515)(12.53,0.747022)(12.54,0.748895)(12.55,0.750768)(12.56,0.752641)(12.57,0.754513)(12.58,0.756386)(12.59,0.758259)(12.6,0.760132)(12.61,0.762004)(12.62,0.763877)(12.63,0.76575)(12.64,0.767622)(12.65,0.769495)(12.66,0.771367)(12.67,0.77324)(12.68,0.775113)(12.69,0.776985)(12.7,0.778858)(12.71,0.78073)(12.72,0.782603)(12.73,0.784475)(12.74,0.786348)(12.75,0.78822)(12.76,0.790093)(12.77,0.791965)(12.78,0.793837)(12.79,0.79571)(12.8,0.797582)(12.81,0.799454)(12.82,0.801327)(12.83,0.803199)(12.84,0.805071)(12.85,0.806943)(12.86,0.808816)(12.87,0.810688)(12.88,0.81256)(12.89,0.814432)(12.9,0.816304)(12.91,0.818176)(12.92,0.820048)(12.93,0.82192)(12.94,0.823792)(12.95,0.825664)(12.96,0.827536)(12.97,0.829408)(12.98,0.83128)(12.99,0.833152)(13.,0.835024) (13.01,0.836896) (13.02,0.838768) (13.03,0.840639) (13.04,0.842511) (13.05,0.844383) (13.06,0.846255) (13.07,0.848126) (13.08,0.849998) (13.09,0.851869) (13.1,0.853741) (13.11,0.855613) (13.12,0.857484) (13.13,0.859356) (13.14,0.861227) (13.15,0.863098) (13.16,0.86497) (13.17,0.866841) (13.18,0.868712) (13.19,0.870584) (13.2,0.872455) (13.21,0.874326) (13.22,0.876197) (13.23,0.878068) (13.24,0.87994) (13.25,0.881811) (13.26,0.883682) (13.27,0.885553) (13.28,0.887424) (13.29,0.889294) (13.3,0.891165) (13.31,0.893036) (13.32,0.894907) (13.33,0.896778) (13.34,0.898649) (13.35,0.900519) (13.36,0.90239) (13.37,0.904261) (13.38,0.906131) (13.39,0.908002) (13.4,0.909872) (13.41,0.911743) (13.42,0.913613) (13.43,0.915483) (13.44,0.917354) (13.45,0.919224) (13.46,0.921094) (13.47,0.922964) (13.48,0.924835) (13.49,0.926705) (13.5,0.928575) (13.51,0.930445) (13.52,0.932315) (13.53,0.934185) (13.54,0.936055) (13.55,0.937925) (13.56,0.939794) (13.57,0.941664) (13.58,0.943534) (13.59,0.945404) (13.6,0.947273) (13.61,0.949143) (13.62,0.951012) (13.63,0.952882) (13.64,0.954751) (13.65,0.956621) (13.66,0.95849) (13.67,0.960359) (13.68,0.962229) (13.69,0.964098) (13.7,0.965967) (13.71,0.967836) (13.72,0.969705) (13.73,0.971574) (13.74,0.973443) (13.75,0.975312) (13.76,0.977181) (13.77,0.97905) (13.78,0.980919) (13.79,0.982787) (13.8,0.984656) (13.81,0.986525) (13.82,0.988393) (13.83,0.990262) (13.84,0.99213) (13.85,0.993999) (13.86,0.995867) (13.87,0.997735) (13.88,0.999603) (13.89,1.00147) (13.9,1.00334) (13.91,1.00521) (13.92,1.00708) (13.93,1.00894) (13.94,1.01081) (13.95,1.01268) (13.96,1.01455) (13.97,1.01642) (13.98,1.01828) (13.99,1.02015) (14.,1.02202)

	};
	\addlegendentry{lower bound \eqref{eq:MoserLB} from \cite[Theorem~4]{lapidoth09}}

	\end{axis}  
\end{tikzpicture}
\label{fig:poissonPlott}

\caption{This plot depicts the capacity of a discrete-time Poisson channel with dark current $\eta=1$ as a function of the peak-power constraint parameter $A$. The red (resp.\ green) line shows the lower (resp.\ upper) bound \eqref{eq:upper:lower:Cp} obtained for a moderate number of iterations, see Appendix~\ref{app:simulation:details}. As a comparison we plot the lower bound of \cite{lapidoth09}, which to the best of our knowledge is the tightest
lower bound available to date (blue line). The parameter $A$ is given in decibels where $A[\textnormal{dB}]=10 \log_{10}(A)$.}

\end{figure}
\end{myex}

\begin{myremark}[AWGN channel with a quantized output]
Another example of a channel that is well studied and can be treated by the proposed method is the discrete-time additive white Gaussian noise (AWGN) channel under output quantization. The output of the channel is described by
\begin{equation*}
Y = \mathsf{Q}(X+N),
\end{equation*}
where $X \in \R$ is the channel input, $N \sim \mathcal{N}(0,\sigma^2)$ for $\sigma^2 > 0$ is white Gaussian noise and $\mathsf{Q}(\cdot)$ is a quantizer that maps the real valued input $X + N$ to one of $M$ bins (where we assume $M < \infty$), which gives $Y \in \{y_1,\ldots,y_MÊ\}$. In addition an average and/or a peak power constraint at the input is considered. More information about this channel model and why it is of interest can be found in \cite{singh09,koch13}. By definition, the AWGN channel with a quantized output has a continuous input alphabet and a discrete output alphabet. Thus, the approximation method discussed in this section can be used to compute the capacity of such channels.  
\end{myremark}

\section{Conclusion and Future Work} \label{sec:conclusion}
We introduced a new approach to approximate the capacity of DMCs possibly having constraints on the input distribution. The dual problem of Shannon's capacity formula turns out to have a particular structure such that the Lagrange dual function admits a closed form solution. Applying smoothing techniques to the non-smooth dual function enables us to solve the dual problem efficiently. This new approach, in the case of no constraints on the input distribution, has a computational complexity per iteration step of $O(MN)$, where $N$ is the input alphabet size and $M$ the size of the output alphabet. In comparison, the Blahut-Arimoto algorithm has the same computational cost of $O(MN)$ per iteration step. More precisely for no input power constraint, the total computational cost to find an $\varepsilon$-close solution is $O(\tfrac{M^2 N \sqrt{\log(N)}}{\varepsilon})$ for the algorithm developed in this article, whereas the Blahut-Arimoto algorithm requires $O(\tfrac{MN \log(N)}{\varepsilon})$. 
A strength of the new approach is that it provides an a posteriori error, i.e., after having run a certain number of iterations we can precisely estimate the actual error in the current approximation. This is computationally appealing as explicit (or a priori) error bounds often are conservative in practice. By exploiting this a posteriori bound we can stop the computation once the desired accuracy has been reached.

As a second contribution, we have shown how similar ideas can be used to approximate the capacity of memoryless channels with continuous bounded input alphabets and countable output alphabets under a mild assumption on the channels tail. This assumption holds, for example, for discrete-time Poisson channels, allowing us to efficiently approximate their capacity. As an example we derived upper and lower bounds for a discrete-time Poisson channel having a peak-power constraint at the input. 

The presented optimization method highly depends on the Lipschitz constant estimate of the objective's gradient. The worse this estimate the more steps the method requires for an a priori $\varepsilon$-precision. For future work, we aim to study the derivation of local Lipschitz constants of the gradient. This technique has recently been shown to be very efficient in practice (up to three orders of magnitude reduction of computation time), while preserving the worst-case complexity \cite{ref:Baes-14}.

In the case of a continuous input alphabet, the proposed method requires to evaluate the gradient $\nabla G_{\nu}(\cdot)$ in every step of Algorithm~1, that requires solving an integral over $\mathbb{A}$. As such the method used to compute those integrals has to be included to the complexity of the proposed algorithm. Therefore, it would be interesting to investigate under which structural properties on the channel the gradient $\nabla G_{\nu}(\cdot)$ can be evaluated efficiently.

The approach introduced in this article can be used to efficiently approximate the capacity of classical-quantum channels, i.e., channels that have classical input and quantum mechanical output, with a discrete or bounded continuous input alphabet. Using the idea of a universal encoder allows us to compute close upper and lower bounds for the Holevo capacity \cite{DavidSutter14}.


\vspace{5mm}





\appendix



\section{Proofs}
This appendix collects the technical proofs omitted above.
\subsection{Proof of Lemma~\ref{lem:equivalent:primal:problem}} \label{ap:LemmaMung}
The mutual information $\I{p}{W}$ can be expressed as
\begin{align*}
\I{p}{W} 	&= 	\sum_{i=1}^{N} \sum_{j=1}^{M} \W_{ij}p_{i}\log\left( \frac{\W_{ij}}{\sum_{k=1}^{N}\W_{kj}p_{k}} \right) \\
		&=	\sum_{i=1}^{N} \sum_{j=1}^{M} \left[ p_{i} \W_{ij}\log(\W_{ij}) - p_{i}\W_{ij}\log\left(\sum_{k=1}^{N}\W_{kj}p_{k} \right) \right].
\end{align*}
By adding the constraint $\sum_{i=1}^{N}p_{i}\W_{ij}=q_{j}$ for all $j=1,\hdots, M$, 
\begin{align*}
\I{p}{W} 	&= 	\sum_{i=1}^{N} \sum_{j=1}^{M} \left[ p_{i} \W_{ij}\log(\W_{ij}) - p_{i}\W_{ij}\log (q_{j}) \right]  \\
		&=	\sum_{i=1}^{N} \sum_{j=1}^{M}  p_{i} \W_{ij}\log(\W_{ij})  -\sum_{j=1}^{M} q_{j }\log(q_{j}) \\
		&= 	- r\transp p + H(q),
\end{align*}
where $p\in\Delta_{N}$. Since $q=\W\transp p$ and $\W\transp$ is a stochastic matrix, this implies $q\in\Delta_{M}$. 
By definition of $S_{\max}$ it is obvious that the input cost constraint $s\transp p \leq S$ is inactive for $S\geq S_{\max}$, leading to the first optimization problem in Lemma~\ref{lem:equivalent:primal:problem}. It remains to show that for
$S<S_{\max}$, the input constraint can be written with equality, leading to the second optimization problem in Lemma~\ref{lem:equivalent:primal:problem}. 
In oder to keep the notation simple we define $\mathsf{C}(S):=C_{S}(W)$ for a fixed channel $W$.
We show that $\mathsf{C}(S)$ is concave in $S$ for $S\in[0,S_{\max}]$. Let $S^{(1)},S^{(2)} \in [0,S_{\max}]$, $0\leq \lambda \leq 1$ and $p^{(i)}$ probability mass functions that achieve $\mathsf{C}(S^{(i)})$ for $i \in \{1,2 \}$. Consider the probability mass function $p^{(\lambda)}=\lambda p^{(1)}+ (1-\lambda) p^{(2)}$. We can write
\begin{align}
s\transp p^{(\lambda)} &= \lambda s\transp p^{(1)} + (1-\lambda) s\transp p^{(2)} \nonumber\\
 &\leq \lambda S^{(1)}+(1-\lambda) S^{(2)} \nonumber\\
 &=: S^{(\lambda)} \label{eq:slam} \in [0,S_{\max}].
\end{align}
Using the concavity of the mutual information in the input distribution, we obtain
\begin{align*}
\lambda \mathsf{C}(S^{(1)}) + (1-\lambda) \mathsf{C}(S^{(2)}) &= \lambda \I{p^{(1)}}{W}+(1-\lambda) \I{p^{(2)}}{W}  \\
&\leq \I{p^{(\lambda)}}{W} \\
&\leq \mathsf{C}(S^{(\lambda)}),
\end{align*}
where the final inequality follows by Shannon's formula for the capacity given in \eqref{eq:shannon48}. $\mathsf{C}(S)$ clearly is non-decreasing in $S$ since enlarging $S$ relaxes the input cost constraint. Furthermore, we show that
\begin{equation} \label{eq:proof:ineq:equality:epsilon:step}
\mathsf{C}(S_{\max}-\varepsilon)<\mathsf{C}(S_{\max}), \quad \text{for all }\varepsilon>0.
\end{equation}
Suppose $\mathsf{C}(S_{\max}-\varepsilon)=\mathsf{C}(S_{\max})$ and denote $\mathsf{C}^{\star}:=\max\limits_{p\in\Delta_{N}}\I{p}{W}$. This then implies that there exists $\bar{p}\in\Delta_{N}$ such that $\I{\bar{p}}{W}=\mathsf{C}^{\star}$ and $s\transp \bar{p}\leq S_{\max}-\varepsilon$, which contradicts the definition of $S_{\max}$. Hence, the concavity of $\mathsf{C}(S)$ together with the non-decreasing property and \eqref{eq:proof:ineq:equality:epsilon:step} imply that $\mathsf{C}(S)$ is strictly increasing in $S$. 
\qed


\subsection{Proof of Lemma~\ref{lem:cover}} \label{ap:cover}
This proof is similar to the proof given in \cite[Theorem~12.1.1]{cover}. Let $q$ satisfy the constraints in \eqref{opt:cover}. Then
\begin{subequations}
\begin{align}
J(q)&= H(q)-c \transp q = -\sum_{i=1}^N q_i \log(q_i) -c\transp q   \nonumber \\
 &= -\sum_{i=1}^N q_i \log\left( \frac{q_i}{p^\star_i}p^\star_i \right) -c\transp q=-\D{q}{p^\star}- \sum_{i=1}^N q_i \log(p^\star_i) - c \transp q  \nonumber \\
 & \leq - \sum_{i=1}^N q_i \log(p^\star_i) - c \transp q \label{eq:ineq}\\
 & = - \sum_{i=1}^N q_i \left(\mu_1 + \mu_2 s_i \right) \label{eq:step1}\\
 &= -\sum_{i=1}^N p^{\star}_i \left(\mu_1 + \mu_2 s_i \right) -c\transp p^\star +c \transp p^\star \label{eq:step2}\\ 
 &= -\sum_{i=1}^N p^\star_i \log(p^\star_i) - c\transp p^\star= J(p^\star). \nonumber
\end{align}
\end{subequations}
The inequality follows form the non-negativity of the relative entropy. Equality \eqref{eq:step1} follows by the definition of $p^\star$ and \eqref{eq:step2} uses the fact that both $p^\star$ and $q$ satisfy the constraints in \eqref{opt:cover}. Note that equality holds in \eqref{eq:ineq} if and only if $q = p^\star$. This proves the uniqueness. \qed


\subsection{Proof of Lemma~\ref{lem:compact:set}} \label{ap:bounding:lambda}
Consider the following two convex optimization problems
\begin{align*} 
\mathsf{P}_{\beta}:\quad  \left\{ \begin{array}{ll}
			\max\limits_{p,q,\varepsilon} 		&-r\transp p + H(q) - \beta\varepsilon\\
			\text{s.t. }						&\|\W\transp p - q\|_{\infty}\leq \varepsilon  \\
										& s\transp p = S \\
			 							& p\in\Delta_{N}, \ q\in\Delta_{M}, \ \varepsilon \in\Rp
	\end{array} \right.  \quad \textnormal{and} \quad
	\quad 
\mathsf{D}_{\beta}: \quad  \left\{ \begin{array}{ll}
			\min\limits_{\lambda} 		&F(\lambda) + G(\lambda) \\
			\text{s.t. } 					& \norm{\lambda}_{1} \leq \beta \\
									& \lambda\in \R^{M}.
	\end{array}\right. 
\end{align*}
\begin{myclaim}
Strong duality holds between $\mathsf{P}_{\beta}$ and $\mathsf{D}_{\beta}$.
\end{myclaim}
\begin{proof}
According to the identity $\norm{\W\transp p - q}_{\infty}=\max_{\norm{\lambda}_{1}\leq 1} \lambda\transp \left( \W\transp p - q \right)$ \cite[p.~7]{holevo_book} the optimization problem $\mathsf{P}_{\beta}$ can be rewritten as
\begin{align*} 
\mathsf{P}_{\beta}:\quad  \left\{ \begin{array}{ll}
			\max\limits_{p,q} 			&-r\transp p + H(q) + \min\limits_{\norm{\lambda}_{1}\leq \beta}\lambda\transp \left( \W\transp p - q \right)\\
			\text{s.t. }					& s\transp p = S \\
			 							& p\in\Delta_{N}, \ q\in\Delta_{M},
	\end{array} \right.
\end{align*}
whose dual program, where strong duality holds according to \cite[Proposition~5.3.1, p.~169]{ref:Bertsekas-09} is given by
\begin{align*}
 \quad  \left\{ \begin{array}{lll}
			\min\limits_{\norm{\lambda}_{1}\leq \beta}  &\max\limits_{p,q} 		&-r\transp p + H(q) + \lambda\transp \left( \W\transp p - q \right) \\
			&\text{s.t. } 				  &s\transp p = S \\
								&	&  p\in\Delta_{N}, \ q\in\Delta_{M},
	\end{array}\right. .
\end{align*}
which clearly is equivalent to $\mathsf{D}_{\beta}$ with $F(\cdot)$ and $G(\cdot)$ as given in \eqref{equation:F:and:G}.
\end{proof}
Denote by $\varepsilon^{\star}(\beta)$ the optimizer of $\mathsf{P}_{\beta}$ with the respective optimal value $J^{\star}_{\beta}$. We show that for a sufficiently large $\beta$ the optimizer $\varepsilon^{\star}(\beta)$ of $\mathsf{P}_{\beta}$ is equal to zero. Hence, in light of the duality relation, the constraint $\norm{\lambda}_{1} \leq \tfrac{\beta}{2}$ in $\mathsf{D}_{\beta}$ is inactive and as such $\mathsf{D}_{\beta}$ is equivalent to $\mathsf{D}$ in equation~\eqref{Lagrange:Dual:Program}. Note that for
\begin{align}  \label{eq:J(eps)}
J(\varepsilon):=  \left\{ \begin{array}{ll}
			\max\limits_{p,q} 		&-r\transp p + H(q) \\
			\text{s.t. }						&\|\W\transp p - q\|_{\infty}\leq \varepsilon \\
									 & s\transp p = S \\
			 							& p\in\Delta_{N}, \ q\in\Delta_{M}
	\end{array} \right. ,
	\end{align}
the mapping $\varepsilon \mapsto J(\varepsilon)$, the so-called perturbation function, is concave \cite[p.~268]{ref:BoyVan-04}. In the next step we write the optimization problem \eqref{eq:J(eps)} in another equivalent form
\begin{align}  \label{eq:J(eps):equiv}
J(\varepsilon)=  \left\{ \begin{array}{ll}
			\max\limits_{p,v} 		&-r\transp p + H(\W\transp p + \varepsilon v) \\
			\text{s.t. }						&\norm{v}_{\infty} \leq 1 \\
								      & s\transp p = S \\
			 							& p\in\Delta_{N}, \ v\in \mathsf{Im}(\W\transp)\subset\R^{M}
	\end{array} \right. .
	\end{align}
By using Taylor's theorem, there exists $y_{\varepsilon}\in[0,\varepsilon]$ such that the entropy term in the objective function of \eqref{eq:J(eps):equiv} can be bounded as
\begin{align}
H(\W\transp p + \varepsilon v) 	&=		H(\W\transp p) - \left( \log(\W\transp p) + \tfrac{1}{\ln 2}\boldsymbol{1} \right)\transp v \varepsilon - \sum_{j=1}^{M}\frac{v_{j}^{2}}{\sum_{i=1}^{N}\W_{ij}p_{i}+y_{\varepsilon}v_{j}}\varepsilon^{2} \tfrac{1}{\ln 2} \nonumber \\
						&\leq		H(\W\transp p) - \left( \log(\W\transp p) + \tfrac{1}{\ln 2}\boldsymbol{1}  \right)\transp v \varepsilon + \frac{M}{\gamma \ln 2}\varepsilon^{2} . \label{eq:Taylor:bound}
\end{align}
Thus, the optimal value of problem $\mathsf{P}_{\beta}$ can be expressed as
\begin{subequations}
\begin{align}
J_{\beta}^{\star} 	&\leq 		\max\limits_{\varepsilon} \left\{ J(\varepsilon)-\beta \varepsilon \right\} \nonumber \\
			&\leq		\max\limits_{\varepsilon} \left\{ \max\limits_{p,v}\left[ -r\transp p + H(\W\transp p) - \left( \log(\W\transp p) + \tfrac{1}{\ln 2} \boldsymbol{1} \right)\transp v \varepsilon :  s\transp p = S  \right]  + \frac{M}{\gamma \ln 2}\varepsilon^{2}  -\beta \varepsilon \right\} \label{eq:proof:compact:first:step}\\
			&\leq		\max\limits_{\varepsilon} \left\{ \max\limits_{p,v}\left[ -r\transp p + H(\W\transp p) \ : \ s\transp p = S  \right] + (\rho - \beta)\varepsilon + \frac{M}{\gamma \ln 2}\varepsilon^{2} \right\} \label{eq:proof:compact:second:step}\\
			&=		J(0) + \max\limits_{\varepsilon} \left\{  (\rho - \beta)\varepsilon +\frac{M}{\gamma \ln 2}\varepsilon^{2} \right\}, \label{eq:proof:compact:third:step}
\end{align}
\end{subequations}
where $\rho = M \left( \log(\gamma^{-1}) \vee \tfrac{1}{\ln 2}\right)$. Note that \eqref{eq:proof:compact:first:step} follows from $\eqref{eq:J(eps):equiv}$ and \eqref{eq:Taylor:bound}. The equation \eqref{eq:proof:compact:second:step} uses the fact that for $\|v\|_\infty \leq 1$, $- \left( \log(\W\transp p) + \tfrac{1}{\ln 2}\boldsymbol{1} \right)\transp v \leq M \left( \log(\gamma^{-1}) \vee \tfrac{1}{\ln 2}\right)$. Thus, for $\beta>\rho$ and $\varepsilon_{1}=\frac{\gamma\ln 2}{M}(\beta - \rho)$, we have $\max\limits_{\varepsilon\leq \varepsilon_{1}} \left\{  (\rho - \beta)\varepsilon + \tfrac{M}{\gamma \ln 2}\varepsilon^{2}\right\} = 0.$ Therefore, \eqref{eq:proof:compact:third:step} together with the concavity of the mapping $\varepsilon\mapsto J(\varepsilon)$ imply that $J(0)$ is the global optimum of $J(\varepsilon)$ and as such $\varepsilon^{\star}(\beta)=0$ for $\beta>\rho$, indicating that $\mathsf{P}_{\beta}$ is equivalent to $\mathsf{P}$ in the sense that $J^{\star}_{\beta}=J^{\star}_{0}$. By strong duality this implies that the constraint $\norm{\lambda}_{1} \leq \beta$ in $\mathsf{D}_{\beta}$ is inactive. Finally, $\norm{\lambda}_{2}\leq\norm{\lambda}_{1}$ concludes the proof.  \qed


\subsection{Proof of Theorem~\ref{thm:C}} \label{ap:cutting:thm}
To prove Theorem \ref{thm:C} we need a preliminary lemma.
	
	\begin{mylem}
	\label{lem:log}
		Given $k \in (0,1)$ and $p\in[0,1]$, we have for all $x\in[0,1-p]$ 
		\begin{align*}
			\big| (p+x)\log(p+x) - p\log(p) \big| \le \frac{\log(\e)}{\e(1-k)} x^k.
		\end{align*}
	\end{mylem}

	\begin{proof}
		Note that for a fixed $x \in [0,1]$, the mapping $p \mapsto  (p+x)\log(p+x) - p\log(p)$ is non-decreasing; observe that the derivative of the mapping is non-negative for all $x \in [0,1]$. Therefore, it suffices to verify the claim for $p \in \{0,1\}$. For $p = 1$ and accordingly $x = 0$, Lemma~\ref{lem:log} holds trivially. Let $p = 0$ and $h(x) \Let \frac{\log(\e)}{\e(1-k)} x^{k-1} + \log(x)$. Note that $h(1) = \tfrac{\log(\e)}{\e(1-k)} >0$ and $h(x) \rightarrow \infty$ as $x \rightarrow 0$. Hence, by setting $\frac{\diff }{\diff x}h(x^\star) = 0$, it can be easily seen that 
			$$\min_{x \in (0,1]}h(x) = h(x^\star)= 0, \qquad x^\star \Let \e^{\frac{1}{k-1}}.$$
		Thus $h(x)\ge 0$, and consequently $xh(x)\ge 0$ for all $x \in (0,1]$, which concludes the proof. 
	\end{proof}

	\begin{proof}[Proof of Theorem \ref{thm:C}]
	We bound the mutual information difference uniformly in the input probability distribution $p \in \meas(\X)$. Observe that
	\begin{align*}
		\big|  I(p &, W)  - I(p,W_{M}) \big| \\
		& = \bigg| \int_{\X} \Big[-h\big(W(\cdot,x)\big) + h\big(W_M(\cdot,x)\big)\Big] p(\diff x) + h\Big(\int_{\X}W(\cdot,x) p(\diff x) \Big) -  h\Big(\int_{\X}W_M(\cdot,x) p(\diff x) \Big) \bigg| \\
		& = \bigg|  \int_{\X} \Big[ \sum\limits_{i \in \N} W(i|x)\log(W(i|x)) - W_M(i|x)\log(W_M(i|x))\Big] p(\diff x) \\
		& \qquad \qquad + \sum\limits_{i \in \N}  -\Big(\int_{\X} W(i|x)p(\diff x)  \Big)\log \Big( \int_{\X} W(i|x)p(\diff x) \Big) \\
		& \qquad \qquad + \Big(\int_{\X} W_M(i|x) p(\diff x) \Big)\log \Big( \int_{\X} W_M(i|x) p(\diff x) \Big) \bigg|.
	\end{align*}
	By the definition of the truncated channel in \eqref{W_M} and applying Lemma \ref{lem:log} to the above relation, we have
	\begin{align*}
		\big|  I(p , W)   - I(p,W_{M}) \big| & \le \frac{\log(\e)}{\e(1-k)} \bigg(  \int_{\X} \bigg[ \sum\limits_{i < M } \Big(\frac{1}{M}\sum\limits_{j \ge M}W(j|x)\Big)^k + \sum\limits_{i \ge M} \big(W(i|x)\big)^k \bigg] p(\diff x) \\ 
		& \qquad \qquad  + \sum\limits_{i < M } \Big(\frac{1}{M}\sum\limits_{j \ge M}  \int_{\X} W(j|x)p(\diff x) \Big)^k + \sum\limits_{i \ge M}   \Big( \int_{\X} W(i|x)p(\diff x) \Big)^k  \bigg) \\
		& \le \frac{2\log(\e)}{\e(1-k)} \bigg( M \Big(\frac{R_1(M)}{M}\Big)^k + R_k(M)\bigg),
	\end{align*}
	which concludes the proof.
	\end{proof}


\subsection{Proof of Proposition~\ref{lem:density:dense}} \label{app:proof:density:restriction}
We show that the optimization problem \eqref{eq:inf:channel:primal} is equivalent to
\begin{equation*}
C_{\A,S}(W_{M}) = \sup\limits_{p\in\mathfrak{D}(\A)} \left\{ I(p,W_{M}) \ : \ \E{s(X)}\leq S \right\},
\end{equation*}
where $\mathfrak{D}(\A)$ is the space of probability measures that are absolutely continuous with respect to the Lebesgue measure. This completes the proof since optimizing over $\mathfrak{D}(\A)$ is equivalent to optimizing over the space of probability densities $\mathcal{D}(\A)$ according to the Radon-Nikod\'ym Theorem \cite[Theorem~3.8, p.~90]{ref:Folland-99}.

It is known that the mapping $p \mapsto I(p,W_{M})$ is weakly lower semicontinuous \cite{ref:verdu-12}. It then suffices to show that $\mathfrak{D}(\A)$ is weakly dense in $\mathcal{P}(\A)$. Let $\B$ be a countable dense subset of $\A$, and $\Delta(\B)$ be the family of probability measures whose supports are finite subsets of $\B$. It is well known that $\Delta(\B)$ is weakly dense in $\mathcal{P}(\A)$, i.e., $\mathcal{P}(\A) = \overline{{\Delta}(\B)}$ \cite[Theorem~4, p.~237]{ref:billingsley-68}, where $\overline{\Delta}$ is the weak closure of $\Delta$. Moreover, thanks to the Lebesgue differentiation theorem \cite[Theorem~3.21, p.~98]{ref:Folland-99}, we know that for any $b \in \B$ the point measure $\delta_{\{b\}} \in \Delta(\B)$ can be arbitrarily weakly approximated by measures in $\mathfrak{D}(\A)$, i.e., $\delta_{\{b\}} \in \overline{\mathfrak{D}(\A)}$. Hence, we have $\Delta(\B) = \overline{\mathfrak{D}(\A)}$, which in light of the preceding assertion implies $\mathcal{P}(\A) = \overline{\mathfrak{D}(\A)}$. \qed

\subsection{Proof of Lemma~\ref{lem:strong:convexity:cts}} \label{app:proof:strong:convexity:poisson}
The proof follows the ideas of \cite{nesterov05}. It can easily be shown that for $d(p):=-\Hdiff{p}+\log(\rho)$
\begin{equation*}
\inprod{d''(p)\cdot g}{g} = \int_{\A}\frac{g(x)^{2}}{p(x)}\drv x.
\end{equation*}
Cauchy-Schwarz then implies
\begin{align*}
\inprod{d''(p)\cdot g}{g} \geq \frac{\left( \int_{\A} g(x) \drv x \right)^{2}}{\int_{\A} p(x) \drv x} = \norm{g}^{2}.
\end{align*} \qed


\subsection{Proof of Proposition~\ref{prop:Lipschitz:cts:channel}} \label{ap:lipschitz:cts:gradient:cts}
It is known, according to Theorem~5.1 in \cite{ref:devolder-12}, that $G_{\nu}(\lambda)$ is well defined and continuously differentiable at any $\lambda\in \R^{M}$ and that  this function is convex and its gradient $\nabla G_{\nu}(\lambda)=\WW^{\star} p_{\nu}^{\lambda}$ is Lipschitz continuous with constant $L_{\nu} =\tfrac{1}{\nu}\norm{\WW}^{2}$, where we have also used Lemma~\ref{lem:strong:convexity:cts}.
The operator norm can be simplified to
\begin{align}
\norm{\WW} 	&=		\sup\limits_{\lambda\in\R^{M}\!, \, p\in\Lp{1}(\A)} \left\{ \inprod{p}{\WW \lambda} \ : \ \norm{\lambda}_{2}=1, \ \norm{p}_{1}=1 \right\} \nonumber \\
			&\leq 	\sup\limits_{\lambda\in\R^{M}\!, \, p\in\Lp{1}(\A)} \left\{ \norm{\WW^{\star}p}_{2} \norm{\lambda}_{2} \ : \ \norm{\lambda}_{2}=1, \ \norm{p}_{1}=1\right\}\label{eq:norm:W:proof:step:CS} \\
			&\leq 	\sup\limits_{ p\in\Lp{1}(\A)} \left\{ \norm{\WW^{\star}p}_{1} \ : \ \norm{p}_{1}=1\right\}\nonumber \\
			&= 		\sup\limits_{ p\in\Lp{1}(\A)} \left\{ \sum_{i=0}^{M-1}  \int_{\mathcal{X}} W_{M}(i|x) p(x) \drv x   \ : \ \norm{p}_{1}=1\right\}\nonumber \\
			&= 		\sup\limits_{ p\in\Lp{1}(\A)} \left\{   \int_{\mathcal{X}} \norm{ W_{M}(\cdot|x)}_{1} p(x) \drv x   \ : \ \norm{p}_{1}=1\right\} \nonumber\\
			&\leq 	\sup\limits_{x\in \A} \norm{ W_{M}(\cdot|x)}_{1} \nonumber \\
			&\leq 	1, \nonumber
\end{align}
where \eqref{eq:norm:W:proof:step:CS} is due to Cauchy-Schwarz.\qed

\subsection{Proof of Lemma~\ref{lem:Lf}} \label{ap:proof:Lf}
Let $x_1,x_2 \in \mathcal{X}$, then by definition of $f_{\lambda}(\cdot)$ we obtain
\begin{subequations}
\begin{align}
&|f_{\lambda}(x_1)-f_{\lambda}(x_2)| \nonumber \\
&\hspace{5mm}= \left| \sum_{i=1}^M W_{M}(i-1|x_1) \lambda_i + \sum_{j=1}^{M} W_{M}(j-1|x_1) \log W_{M}(j-1|x_1) \right. \nonumber \\
&\hspace{11mm}\left. - \sum_{i=1}^M W_{M}(i-1|x_2) \lambda_i - \sum_{j=1}^{M} W_{M}(j-1|x_2) \log W_{M}(j-1|x_2) \right| \nonumber \\
&\hspace{5mm}\leq \left| \sum_{i=1}^M \left(W_{M}(i-1|x_1)-W_{M}(i-1|x_2) \right)\lambda_i \right| + \left|\Hh{W_{M}(\cdot|x_1)}-\Hh{W_{M}(\cdot|x_2)} \right| \label{eq:triang}\\
&\hspace{5mm}\leq  \sum_{i=1}^M \left|\left(W_{M}(i-1|x_1)-W_{M}(i-1|x_2) \right)\lambda_i \right| + \left|\Hh{W_{M}(\cdot|x_1)}-\Hh{W_{M}(\cdot|x_2)} \right| \label{eq:ass1}\\
&\hspace{5mm}\leq L M \|\lambda \|_1 |x_1-x_2 | + \left|\Hh{W_{M}(\cdot|x_1)}-\Hh{W_{M}(\cdot|x_2)} \right|\label{eq:LambaSet}\\
&\hspace{5mm}\leq L M^2  \left(\log \frac{1}{\gamma_M}\vee \frac{1}{\ln 2} \right)  |x_1-x_2| + \left|\Hh{W_{M}(\cdot|x_1)}-\Hh{W_{M}(\cdot|x_2)} \right| \label{eq:LambdaSet}\\
&\hspace{5mm}\leq L M^2  \left(\log \frac{1}{\gamma_M}\vee \frac{1}{\ln 2} \right)  |x_1-x_2| + ML \left|\log \frac{1}{\gamma_M}- \frac{1}{\ln 2} \right|  |x_1-x_2|. \label{eq:fini}
\end{align}
\end{subequations}
Inequalities~\eqref{eq:triang} and \eqref{eq:ass1} use the triangle inequality. Inequality~\eqref{eq:LambaSet} follows by Assumption~\ref{a:channel}\eqref{ass:channel:ii} and \eqref{eq:LambdaSet} can be derived by following the proof of Lemma~\ref{lem:compact:set:Poisson}, which is similar to the one of Lemma~\ref{lem:compact:set}. Finally, \eqref{eq:fini} follows from the fact that the function $\Delta_n \ni x^n \mapsto \Hh{x^n} \in \Rp $ with $\min_{1\leq i \leq n} x_i < c$ is Lipschitz continuous with constant $n \left|\log \tfrac{1}{c}-\tfrac{1}{\ln 2} \right|$ and from Assumption~\ref{a:channel}\eqref{ass:channel:ii}. 
\qed

\subsection{Proof of Lemma~\ref{lem:iota:new}} \label{ap:proof:iota}
We start by the following definitions that simplify the proof below
\begin{alignat*}{3}
f_{\lambda, \nu}(x)&:=\WW\lambda(x)-r(x) + \nu \mu_{\nu}s(x), \qquad  &\bar{f}_{\lambda,\nu}&:= \max\limits_{x\in\A}f_{\lambda,\nu}(x) \\
B_{\lambda,\nu}(\varepsilon)&:=\left\{ x\in\A \ | \ \bar{f}_{\lambda,\nu} - f_{\lambda,\nu}(x) <\varepsilon \right\}, \qquad  &\eta_{\lambda,\nu}(\varepsilon) &:= \int_{B_{\lambda,\nu}(\varepsilon)}\drv x.
\end{alignat*}
By the Lipschitz continuity of $f_{\lambda}(\cdot)$ and $s(\cdot)$ we get the uniform lower bound 
\begin{equation} \label{eq:eta:lower:bound}
\eta_{\lambda,\nu}(\varepsilon) \geq  \frac{\varepsilon}{L_{f}+|\nu\mu_{\nu}|L_{s}}\wedge \rho.
\end{equation}
By using the solution to $G_{\nu}(\lambda)$, according to \eqref{eq:finite:optimizer:pmu:cts} we can write
\begin{subequations}
\begin{align}
G_{\nu}(\lambda) 	&= 		-\nu \log(\rho) + \nu \log\left( \int_{\A}2^{\tfrac{1}{\nu}f_{\lambda,\nu}(x)}\drv x \right) \label{eq:proof:Gnu}\\
					&\leq	\inf\limits_{\ell\in\R}\max\limits_{x\in\A} \left\{ f_{\lambda}(x)+\ell s(x) \right\} \label{eq:step:Gnu:smaller:G0} \\
					&= G(\lambda), \label{eq:step:strong:duality}
\end{align}
\end{subequations}	
where the equality \eqref{eq:step:strong:duality} follows as \eqref{eq:step:Gnu:smaller:G0} is the dual program to $G(\lambda)$ and strong duality holds. The inequality \eqref{eq:step:Gnu:smaller:G0} then is due to $G_{\nu}(\lambda)\leq G(\lambda)$ for any $\lambda$, see \eqref{eq:uniform:bound:cts}. Therefore,

\begin{subequations}
\begin{align}
&G(\lambda)-G_{\nu}(\lambda)		\leq		\bar{f}_{\lambda,\nu}-G_{\nu}(\lambda) \label{eq:iota:G} \\
		&\qquad =		\nu \left( -\log \left(  \int_{B_{\lambda,\nu}(\varepsilon)}2^{\frac{1}{\nu}\left( f_{\lambda,\nu}(x) - \bar{f}_{\lambda,\nu}\right)} \drv x +  \int_{B_{\lambda,\nu}^{\setC}(\varepsilon)}2^{\frac{1}{\nu}\left( f_{\lambda,\nu}(x) - \bar{f}_{\lambda,\nu}\right)} \drv x\right) + \log(\rho)  \right) \label{eq:iota:Gnu}\\
		&\qquad \leq 	\nu \left( -\log \left(  \int_{B_{\lambda,\nu}(\varepsilon)}2^{\frac{1}{\nu}\left( f_{\lambda,\nu}(x) - \bar{f}_{\lambda,\nu}\right)} \drv x \right) + \log(\rho)  \right) \nonumber \\
		&\qquad \leq 	\nu \left( -\log \left( \eta_{\lambda,\nu}(\varepsilon) 2^{-\frac{\varepsilon}{\nu}} \right) + \log(\rho)  \right) \label{eq:iota:B:eta} \\
		&\qquad \leq 	\nu \left( -\log \left( \tfrac{\varepsilon}{L_{f}+|\nu\mu_{\nu}|L_{s}}\vee \rho \right) +\frac{\varepsilon}{\nu} + \log(\rho)  \right) \label{eq:iota:lipschitz} \\
		&\qquad =		\nu  \log \left( \tfrac{(L_{f}+|\nu\mu_{\nu}|L_{s})\rho}{\varepsilon} \vee 1 \right) + \varepsilon, \nonumber
\end{align}
\end{subequations}
where \eqref{eq:iota:G} follows from \eqref{eq:step:strong:duality} and \eqref{eq:iota:Gnu} is due to \eqref{eq:proof:Gnu}. The inequality \eqref{eq:iota:B:eta} results from the definitions of $B_{\lambda,\nu}(\varepsilon)$ and $\eta_{\lambda,\nu}(\varepsilon)$ above and \eqref{eq:iota:lipschitz} is implied by \eqref{eq:eta:lower:bound}. Finally, it can be seen that for $\nu<(L_{f}+|\nu\mu_{\nu}|L_{s})\rho$, the optimal choice for $\varepsilon$ is $\nu$, which leads to
\begin{equation} \label{eq:proof:itoa:almost:done}
G(\lambda)-G_{\nu}(\lambda) \leq \nu \left( 1 +  \log \left( \tfrac{(L_{f}+|\nu\mu_{\nu}|L_{s})\rho}{\nu} \vee 1 \right) \right). 
\end{equation}
It remains to upper bound the term $|\nu\mu_{\nu}|$. Define $\underline{f}:=\min_{x,\lambda}f_{\lambda}(x)$, $\overline{f}:=\max_{x,\lambda}f_{\lambda}(x)$, $\Delta_{f}:=\overline{f}-\underline{f}$ and note that $\Delta_{f}\leq L_{f}\rho$. By \eqref{eq:proof:Gnu}, \eqref{eq:uniform:bound:cts} and the fact that adding an additional constraint to a maximization problem cannot increase its objective value
\begin{equation*}
G_{\nu}(\lambda) =  \nu \log\left( \int_{\A}2^{\tfrac{1}{\nu}\left( f_{\lambda}(x)+\nu\mu_{\nu}s(x)\right)}\drv x \right) -\nu \log(\rho) \leq \overline{f}=\nu \log \left( 2^{\tfrac{1}{\nu}\overline{f}} \right),
\end{equation*}
which is equivalent to
 $\int_{\A}2^{\tfrac{1}{\nu}\left( f_{\lambda}(x) - \overline{f} +\nu\mu_{\nu}s(x)\right)}\drv x \leq \rho $ and implies
 \begin{equation} \label{eq:proof:iota:inbetween}
 \int_{\A} 2^{\mu_{\nu}s(x)} \drv x \leq \rho \, 2^{\tfrac{\Delta_{f}}{\nu}}.
 \end{equation}
From \eqref{eq:proof:iota:inbetween} two bounds can be derived. First, \eqref{eq:proof:iota:inbetween} implies that $\left(\rho \wedge \tfrac{\varepsilon}{L_{s}}\right) 2^{\mu_{\nu}(\overline{s}-\varepsilon)}\leq \rho 2^{\tfrac{\Delta_{f}}{\nu}}$, which by choosing $\varepsilon=\tfrac{\overline{s}}{2}$ leads to $2^{\mu_{\nu}\tfrac{\overline{s}}{2}}\leq \left( \tfrac{2L_{s}\rho}{\overline{s}}\vee 1 \right)2^{\tfrac{\Delta_{f}}{\nu}}$ and finally, 
\begin{equation} \label{eq:itoa:upper:bound}
\nu\mu_{\nu}\leq \tfrac{2}{\overline{s}}\log\left( \tfrac{2L_{s}\rho}{\overline{s}}\vee 1 \right)\nu + \tfrac{2\Delta_{f}}{\overline{s}}.
\end{equation}
Similarly one can derive a lower bound
\begin{equation} \label{eq:itoa:lower:bound}
\nu\mu_{\nu}\geq \tfrac{2}{\underline{s}}\log\left( \tfrac{2L_{s}\rho}{-\underline{s}}\vee 1 \right)\nu +  \tfrac{2\Delta_{f}}{\underline{s}}.
\end{equation}
Equation \eqref{eq:proof:itoa:almost:done} together with \eqref{eq:itoa:upper:bound} and \eqref{eq:itoa:lower:bound} complete the proof.
\qed


\subsection{Proof of Theorem~\ref{thm:error:bound:capacity:continuous:channel}} \label{ap:thm:error:bound:capacity:continuous:channel}

Following \cite{nesterov05} and using Lemma~\ref{lem:compact:set:Poisson}, Lemma~\ref{lem:strong:duality:poisson}, Propostion~\ref{prop:Lipschitz:cts:channel} and Lemma~\ref{lem:iota:new}, after $n$ iterations of Algorithm~\hyperlink{algo:1}{1} the following approximation error is obtained
\begin{equation}\label{eq:proof:rate:error:term}
0\leq F(\hat{\lambda}) + G(\hat{\lambda}) - \I{\hat{p}}{W} \leq \iota(\nu) + \frac{4D_{1}}{\nu (n+1)^{2}}+\frac{4D_{1}}{(n+1)^{2}}=:\mathsf{err}(\nu,n),
\end{equation}
where for $\nu<\tfrac{T_{1}}{1-T_{2}}$ or $T_{2}>1$ we have $\iota(\nu)=\nu \left( \log\left( \tfrac{T_{1}}{\nu} + T_{2} \right) +1 \right)$, which is strictly increasing in $\nu$. Let us redefine the smoothing term by $\nu:=\tfrac{\delta}{\log\left( \delta^{-1} \right)}$ for $\delta \in (0,1)$ and define the function $g(\delta):=\left( \tfrac{\log\left( T_{1} \log\left( \delta^{-1} \right) + T_{2}\delta \right)+1}{\log\left( \delta^{-1} \right)} +1 \right)$. One can see that $\iota(\nu)=\delta g(\delta)$ and that $\lim_{\delta\to 0}g(\delta)=1$. Furthermore $\delta\leq 2^{-1} \wedge 2^{-\tfrac{1}{T_{1}+T_{2}}}$ implies 
\begin{equation} \label{eq:bound:g:delta}
g(\delta) - 1 \leq \frac{\log \left( 2(T_{1}+T_{2})\log\left( \delta^{-1} \right) \right) }{\log\left( \delta^{-1} \right)} \leq T_{1}+T_{2},
\end{equation}
where the first inequality is due to $\delta\leq 2^{-1}$ and the second follows from $\delta\leq 2^{-\tfrac{1}{T_{1}+T_{2}}}$. We seek for a lower bound of $n$ and upper bound $\delta$ such that the error term \eqref{eq:proof:rate:error:term} is smaller than the preassigned $\varepsilon>0$, i.e.,
\begin{equation} \label{eq:proof:leq:vareps}
\mathsf{err}(\tfrac{\delta}{\log\left( \delta^{-1} \right)},n) = g(\delta) \delta + \frac{4D_{1}}{(n+1)^{2}} \left( \frac{\log\left( \delta^{-1} \right)}{\delta} +1 \right) \leq \varepsilon
\end{equation}
To this end, we introduce an auxiliary variable $\zeta \in (0,1)$ such that such that $g(\delta) \delta = (1-\zeta)\varepsilon$ and $\tfrac{4D_{1}}{(n+1)^{2}} \left( \tfrac{\log\left( \delta^{-1} \right)}{\delta} +1 \right)\leq \zeta \varepsilon$, which implies \eqref{eq:proof:leq:vareps}. Observe that $g(\delta) \delta = (1-\zeta)\varepsilon$ is equivalent to $\delta = \tfrac{(1-\zeta)}{g(\delta)}\varepsilon=:\beta \varepsilon$. Hence $\zeta=1-\beta g(\delta)$ for $\beta \in [0,\tfrac{1}{g(\delta)}]$. Moreover,
\begin{align*}
\frac{4D_{1}}{(n+1)^{2}} \left( \frac{\log\left( \delta^{-1} \right)}{\delta} +1 \right) =  \frac{4D_{1}}{(n+1)^{2}} \left( \frac{\log\left( (\beta \varepsilon)^{-1} \right)}{\beta \varepsilon} +1 \right)\leq \zeta \varepsilon \nonumber
\end{align*}
is equivalent to
\begin{align}
 4 D_{1} \left(  \frac{\log\left( (\beta \varepsilon)^{-1} \right) + \beta \varepsilon}{\beta (1-g(\delta)\beta)\varepsilon^{2}}\right)
= 4 D_{1} \left( \frac{\log(\varepsilon^{-1}) + \log 2 g(\delta) + \tfrac{\varepsilon}{2 g(\delta)}}{\tfrac{\varepsilon^{2}}{4g(\delta)}} \right) \leq (n+1)^{2} \label{eq:proof:no:beta},
\end{align}
where we have chosen $\beta=\tfrac{1}{2g(\delta)}$ and as such is equivalent to
\begin{align}
\frac{4}{\varepsilon} \sqrt{D_{1}\left( g(\delta) \log\left( \varepsilon^{-1} \right) + g(\delta)\log\left( 2 g(\delta) \right) + \tfrac{\varepsilon}{2} \right)} \leq n+1 \nonumber
\end{align}
 Finally, using \eqref{eq:bound:g:delta} implies for $\nu = \tfrac{\varepsilon / \alpha}{\log \left( \alpha / \varepsilon \right)}$, where $\alpha := 2(T_1 + T_2 + 1)$
\begin{equation*}
\mathsf{err}(\nu,n) \leq \varepsilon \quad \text{ for } \quad n\geq \tfrac{1}{\varepsilon} \sqrt{8 D_{1}\alpha}\sqrt{\log(\varepsilon^{-1}) + \log(\alpha) + \tfrac{1}{4}}.
\end{equation*} \qed


\subsection{Proof of Proposition~\ref{prop:poisson_tail}} \label{ap:cutting:poisson}
To prove Proposition \ref{prop:poisson_tail}, we need two lemmas. 
	
	\begin{mylem}\label{lem:ab}
		For any $k\in(0,1]$ and $a, b \ge 0$
		\begin{align*}
			a^k+b^k \le 2^{1-k}(a+b)^k.
		\end{align*}
	\end{mylem}
	 
	 \begin{proof}
		 Let $g(x) \Let 2^{1-k}(1+x)^k - x^k$. By setting $\frac{\diff}{\diff x}g(x^\star) = 0$, one can easily see that $x^\star = 1$ is the minimizer of function $g$ over the interval $[0,1]$, i.e., $g(x) \ge g(1) = 1$ for all $x \in [0,1]$. Suppose, without loss of generality, that $a \ge b$. By virtue of the preceding result of function $g$, we know that 
		 	$$ 1 \le g\left( \frac{b}{a} \right) =  2^{1-k}\left( 1+\frac{b}{a} \right)^k - \left(\frac{b}{a}\right)^k, $$
		 where by multiplying $a^k$ it readily leads to the desired assertion. 
	 \end{proof}
	 
	\begin{mylem}\label{lem:ai}
		Let $(a_i)_{i \in \mathbb{N}}$ be a non-negative sequence of real numbers. For any $k \in (0,1]$
		\begin{align*}
			\sum\limits_{i \in \mathbb{N}} a_i^k \le \Big(\sum\limits_{i \in \mathbb{N}} \alpha^{i}a_i\Big)^k, \qquad \alpha \Let 2^{(k^{-1}-1)}.
		\end{align*}
	\end{mylem}
	\begin{proof}
		For the proof we make use of an induction argument. Note that for any $a_1 \ge 0$ it trivially holds that $a_1^k \le 2^{1 - k}a_1^k$. We now assume that for any sequence $(a_i)_{i=1}^{N} \subset \R_{\ge 0}$ we have
		\begin{align}
		\label{eq:0}
			\sum\limits_{i=1}^{N} a_i^k \le \Big(\sum\limits_{i=1}^{N} 2^{(k^{-1}-1)i}a_i\Big)^k. 
		\end{align}
		Let $(a_i)_{i=1}^{N+1} \subset \R_{\ge 0}$. Then, 
		\begin{align}
			\label{eq:1} \sum\limits_{i=1}^{N+1} a_i^k & = a_1^k + \sum\limits_{i=2}^{N+1} a_i^k  \le a_1^k + \Big(\sum\limits_{i=2}^{N+1} 2^{(k^{-1}-1){(i-1)}}a_i\Big)^k \le 2^{1-k} \Big(a_1 + \sum\limits_{i=2}^{N+1} 2^{(k^{-1}-1){(i-1)}}a_i\Big)^k \\
			& \notag = \Big(2^{(k^{-1}-1)}a_1 + \sum\limits_{i=2}^{N+1} 2^{(k^{-1}-1){i}}a_i\Big)^k = \Big(\sum\limits_{i=1}^{N+1} 2^{(k^{-1}-1){i}}a_i\Big)^k,
		\end{align}
		where the first (resp.\ second) inequality in \eqref{eq:1} follows from \eqref{eq:0} (resp.\ Lemma \ref{lem:ab}). 
	\end{proof}

	\begin{proof}[Proof of Proposition \ref{prop:poisson_tail}]
		It is straightforward to see that
		\begin{align}
		\label{A,n}
			\max_{x \in [0,A]} \e^{-x} x^i = \e^{-\min\{A,i\}} \big(\min\{A,i\}\big)^{i}. 
		\end{align} 
		Moreover, based on a Taylor series expansion, it is well known that for all $M \in \mathbb{N}$ and $x \in \R_{\geq 0}$
		\begin{align}
		\label{taylor}
			\sum\limits_{i \ge M} \frac{x^i}{i!} \le \frac{\e^x}{M!} x^{M}. 
		\end{align}
		Therefore, it follows that
		\begin{subequations}
		\begin{align}
			\label{eq:3} R_k(M) & \Let \sum\limits_{i\ge M} \Big(\sup_{x \in [0,A]} \e^{-(x + \eta)} \frac{(x+\eta)^{i}}{i!} \Big)^k \le \sum\limits_{i\ge M} \Big( \e^{-(A + \eta)} \frac{(A+\eta)^{i}}{i!} \Big)^k \\
			\label{eq:4} & \le \e^{-k(A + \eta)}\Big( \sum\limits_{i \ge M} \alpha^{(i-M+1)} \frac{(A+\eta)^{i}}{i!} \Big)^k = \frac{\e^{-k(A + \eta)}}{\alpha^{k(M-1)}} \Big( \sum\limits_{i \ge M}  \frac{\big(\alpha(A+\eta)\big)^{i}}{i!} \Big)^k \\
			\label{eq:5} & \le \frac{\e^{-k(A + \eta)}}{\alpha^{k(M-1)}} \Big( \frac{\e^{\alpha(A+\eta)}}{M!} \alpha^{M} (A+\eta)^{M} \Big)^k = \Big( {\alpha\e^{(\alpha-1)(A+\eta)}\frac{(A+\eta)^{M}}{M!}} \Big)^k,
		\end{align}
		\end{subequations} 
		where \eqref{eq:3} results from \eqref{A,n} and the assumption $M \ge A + \eta$, and \eqref{eq:4} (resp.\ \eqref{eq:5}) follows from Lemma \ref{lem:ai} (resp.\ \eqref{taylor}). 
	\end{proof}


\section{Simulation Details}\label{app:simulation:details}	
This section provides some further details on the simulation in Example~\ref{ex:poisson}. The parameters considered are $k=\tfrac{1}{2}$, $L_{f}=0$ and $M$ is chosen according to Table~\ref{tab:poisson:details}. All the simulations in this section are performed on a 2.3 GHz Intel Core i7 processor with 8 GB RAM with Matlab.

\begin{table}[!htb]
\centering 
\caption{Simulation details to Example~\ref{ex:poisson}}
\label{tab:poisson:details}
\vspace{1mm}
  \begin{tabular}{c | c c c c c c c c}
 $A$ [dB] \hspace{1mm}  & \hspace{1mm}    0  &1 & 2 & 3 & 4 & 5 & 6 & 7 \\ 
 $M$ \hspace{1mm}  & \hspace{1mm}    16  &17 & 19 & 20 & 22 & 25 & 28 & 31 \\ 
 Iterations $n$ & \hspace{1mm} 4$\cdot 10^{4}$ & 4$\cdot 10^{4}$   & 4$\cdot 10^{4}$  & 5$\cdot 10^{4}$  & 6$\cdot 10^{4}$ & 7$\cdot 10^{4}$ & 9$\cdot 10^{4}$ &1.2$\cdot 10^{5}$ \\
 $\nu$ \hspace{1mm}  & \hspace{1mm} 0.0026 & 0.0029 & 0.0036 & 0.0029 & 0.0027 & 0.0029 & 0.0026 &  0.0022 \\ 
 $F(\hat{\lambda}) + G(\hat{\lambda})$ & \hspace{1mm} 0.1144 & 0.1626  & 0.2263 & 0.3063 & 0.4029 & 0.5129 & 0.6293 & 0.7423 \\
 $\I{\hat{p}}{W}$ & \hspace{1mm} 0.1105 & 0.1583  & 0.2206  & 0.3015  & 0.3979 & 0.5072 & 0.6234 &  0.7365 \\ 
 $\mathcal{E}$ & \hspace{1mm} 9.3$\cdot 10^{-4}$ &9.7$\cdot 10^{-4}$  & 4.8$\cdot 10^{-4}$  & 8.5$\cdot 10^{-4}$  & 8.2$\cdot 10^{-4}$  & 4.9 $\cdot 10^{-4}$ & 5.0$\cdot 10^{-4}$ & 9.5$\cdot 10^{-4}$ \\ 
  \end{tabular}
  
  \vspace{7mm}
  \phantom{bla}
  \hspace{-13.1mm}
    \begin{tabular}{c | c c c c c c c}
 $A$ [dB] \hspace{1mm}  & \hspace{1mm}  8  & 9   & 10  & 11  & 12 & 13 & 14   \\ 
 $M$ \hspace{1mm}  & \hspace{1mm}   36   & 42  & 49  & 59 &  71 & 85 & 104    \\ 
 Iterations $n$ & \hspace{1mm}   2$\cdot 10^{5}$ & 5$\cdot 10^{5}$   &2$\cdot 10^{6}$   & 3$\cdot 10^{6}$ & 4$\cdot 10^{6}$ & 9$\cdot 10^{6}$ & 1.5$\cdot 10^{7}$\\
 $\nu$ \hspace{1mm}  & \hspace{1mm} 0.0016 & 7.1 $\cdot 10^{-5}$ & 8.0$\cdot 10^{-4}$ &8.3$\cdot 10^{-4}$& 9.7$\cdot 10^{-4}$ &  6.2$\cdot 10^{-4}$ & 5.8$\cdot 10^{-4}$ \\
 $F(\hat{\lambda}) + G(\hat{\lambda})$ & \hspace{1mm}   0.8410 & 0.9422  & 1.0591  & 1.1835  & 1.3070 & 1.4343 & 1.5671  \\
 $\I{\hat{p}}{W}$ & \hspace{1mm} 0.8351 & 0.9388  & 1.0547  & 1.1788 & 1.3013 &1.4219 & 1.5605 \\ 
 $\mathcal{E}$ & \hspace{1mm}    7.5$\cdot 10^{-4}$ &  7.1$\cdot 10^{-4}$ & 8.0$\cdot 10^{-4}$  & 6.2$\cdot 10^{-4}$ & 5.2 $\cdot 10^{-4}$ & 9.0$\cdot 10^{-4}$ & 6.7$\cdot 10^{-4}$ \\ 
  \end{tabular}
  
\end{table}	


\section*{Acknowledgment}
The authors thank Yurii Nesterov, Renato Renner and Stefan Richter for helpful discussions and pointers to references.

\bibliography{../bibtex/header,../bibtex/bibliofile}

\begin{thebibliography}{44}%
\makeatletter
\providecommand \@ifxundefined [1]{%
 \@ifx{#1\undefined}
}%
\providecommand \@ifnum [1]{%
 \ifnum #1\expandafter \@firstoftwo
 \else \expandafter \@secondoftwo
 \fi
}%
\providecommand \@ifx [1]{%
 \ifx #1\expandafter \@firstoftwo
 \else \expandafter \@secondoftwo
 \fi
}%
\providecommand \natexlab [1]{#1}%
\providecommand \enquote  [1]{``#1''}%
\providecommand \bibnamefont  [1]{#1}%
\providecommand \bibfnamefont [1]{#1}%
\providecommand \citenamefont [1]{#1}%
\providecommand \href@noop [0]{\@secondoftwo}%
\providecommand \href [0]{\begingroup \@sanitize@url \@href}%
\providecommand \@href[1]{\@@startlink{#1}\@@href}%
\providecommand \@@href[1]{\endgroup#1\@@endlink}%
\providecommand \@sanitize@url [0]{\catcode `\\12\catcode `\$12\catcode
  `\&12\catcode `\#12\catcode `\^12\catcode `\_12\catcode `\%12\relax}%
\providecommand \@@startlink[1]{}%
\providecommand \@@endlink[0]{}%
\providecommand \url  [0]{\begingroup\@sanitize@url \@url }%
\providecommand \@url [1]{\endgroup\@href {#1}{\urlprefix }}%
\providecommand \urlprefix  [0]{URL }%
\providecommand \Eprint [0]{\href }%
\providecommand \doibase [0]{http://dx.doi.org/}%
\providecommand \selectlanguage [0]{\@gobble}%
\providecommand \bibinfo  [0]{\@secondoftwo}%
\providecommand \bibfield  [0]{\@secondoftwo}%
\providecommand \translation [1]{[#1]}%
\providecommand \BibitemOpen [0]{}%
\providecommand \bibitemStop [0]{}%
\providecommand \bibitemNoStop [0]{.\EOS\space}%
\providecommand \EOS [0]{\spacefactor3000\relax}%
\providecommand \BibitemShut  [1]{\csname bibitem#1\endcsname}%
\let\auto@bib@innerbib\@empty
\bibitem [{\citenamefont {Shannon}(1948)}]{shannon48}%
  \BibitemOpen
  \bibfield  {author} {\bibinfo {author} {\bibfnamefont {Claude~E.}\
  \bibnamefont {Shannon}},\ }\bibfield  {title} {\enquote {\bibinfo {title} {A
  mathematical theory of communication},}\ }\href
  {http://cm.bell-labs.com/cm/ms/what/shannonday/shannon1948.pdf} {\bibfield
  {journal} {\bibinfo  {journal} {Bell System Technical Journal}\ }\textbf
  {\bibinfo {volume} {27}},\ \bibinfo {pages} {379--423} (\bibinfo {year}
  {1948})}\BibitemShut {NoStop}%
\bibitem [{\citenamefont {Blahut}(1972)}]{blahut72}%
  \BibitemOpen
  \bibfield  {author} {\bibinfo {author} {\bibfnamefont {Richard~E.}\
  \bibnamefont {Blahut}},\ }\bibfield  {title} {\enquote {\bibinfo {title}
  {Computation of channel capacity and rate-distortion functions},}\ }\href
  {\doibase 10.1109/TIT.1972.1054855} {\bibfield  {journal} {\bibinfo
  {journal} {IEEE Transactions on Information Theory}\ }\textbf {\bibinfo
  {volume} {18}},\ \bibinfo {pages} {460--473} (\bibinfo {year}
  {1972})}\BibitemShut {NoStop}%
\bibitem [{\citenamefont {Gallager}(1968)}]{gallager68}%
  \BibitemOpen
  \bibfield  {author} {\bibinfo {author} {\bibfnamefont {Robert~G.}\
  \bibnamefont {Gallager}},\ }\href@noop {} {\emph {\bibinfo {title}
  {Information Theory and Reliable Communication}}}\ (\bibinfo  {publisher}
  {John Wiley \& Sons},\ \bibinfo {year} {1968})\BibitemShut {NoStop}%
\bibitem [{\citenamefont {Arimoto}(1972)}]{arimoto72}%
  \BibitemOpen
  \bibfield  {author} {\bibinfo {author} {\bibfnamefont {Suguru}\ \bibnamefont
  {Arimoto}},\ }\bibfield  {title} {\enquote {\bibinfo {title} {An algorithm
  for computing the capacity of arbitrary discrete memoryless channels},}\
  }\href {\doibase 10.1109/TIT.1972.1054753} {\bibfield  {journal} {\bibinfo
  {journal} {IEEE Transactions on Information Theory}\ }\textbf {\bibinfo
  {volume} {18}},\ \bibinfo {pages} {14--20} (\bibinfo {year}
  {1972})}\BibitemShut {NoStop}%
\bibitem [{\citenamefont {Abou-Faycal}\ \emph {et~al.}(2001)\citenamefont
  {Abou-Faycal}, \citenamefont {Trott},\ and\ \citenamefont
  {Shamai}}]{shamai01}%
  \BibitemOpen
  \bibfield  {author} {\bibinfo {author} {\bibfnamefont {Ibrahim~C.}\
  \bibnamefont {Abou-Faycal}}, \bibinfo {author} {\bibfnamefont {Mitchell~D.}\
  \bibnamefont {Trott}}, \ and\ \bibinfo {author} {\bibfnamefont {Shlomo}\
  \bibnamefont {Shamai}},\ }\bibfield  {title} {\enquote {\bibinfo {title} {The
  capacity of discrete-time memoryless {R}ayleigh-fading channels},}\ }\href
  {\doibase 10.1109/18.923716} {\bibfield  {journal} {\bibinfo  {journal} {IEEE
  Transactions on Information Theory}\ }\textbf {\bibinfo {volume} {47}},\
  \bibinfo {pages} {1290--1301} (\bibinfo {year} {2001})}\BibitemShut {NoStop}%
\bibitem [{\citenamefont {Sayir}(2000)}]{sayir00}%
  \BibitemOpen
  \bibfield  {author} {\bibinfo {author} {\bibfnamefont {Jossy}\ \bibnamefont
  {Sayir}},\ }\bibfield  {title} {\enquote {\bibinfo {title} {Iterating the
  {A}rimoto-{B}lahut algorithm for faster convergence},}\ }\href {\doibase
  10.1109/ISIT.2000.866533} {\bibfield  {journal} {\bibinfo  {journal}
  {Proceedings IEEE International Symposium on Information Theory (ISIT)}\ ,\
  \bibinfo {pages} {235}} (\bibinfo {year} {2000})}\BibitemShut {NoStop}%
\bibitem [{\citenamefont {Matz}\ and\ \citenamefont {Duhamel}(2004)}]{matz04}%
  \BibitemOpen
  \bibfield  {author} {\bibinfo {author} {\bibfnamefont {Gerald}\ \bibnamefont
  {Matz}}\ and\ \bibinfo {author} {\bibfnamefont {Pierre}\ \bibnamefont
  {Duhamel}},\ }\bibfield  {title} {\enquote {\bibinfo {title} {Information
  geometric formulation and interpretation of accelerated
  {B}lahut-{A}rimoto-type algorithms},}\ }\href {\doibase
  10.1109/ITW.2004.1405276} {\bibfield  {journal} {\bibinfo  {journal}
  {Proceedings Information Theory Workshop (ITW)}\ ,\ \bibinfo {pages}
  {66--70}} (\bibinfo {year} {2004})}\BibitemShut {NoStop}%
\bibitem [{\citenamefont {Yu}(2010)}]{yaming10}%
  \BibitemOpen
  \bibfield  {author} {\bibinfo {author} {\bibfnamefont {Yaming}\ \bibnamefont
  {Yu}},\ }\bibfield  {title} {\enquote {\bibinfo {title} {Squeezing the
  {A}rimoto-{B}lahut algorithm for faster convergence},}\ }\href {\doibase
  10.1109/TIT.2010.2048452} {\bibfield  {journal} {\bibinfo  {journal} {IEEE
  Transactions on Information Theory}\ }\textbf {\bibinfo {volume} {56}},\
  \bibinfo {pages} {3149--3157} (\bibinfo {year} {2010})}\BibitemShut {NoStop}%
\bibitem [{\citenamefont {Dauwels}(2005)}]{dauwels05}%
  \BibitemOpen
  \bibfield  {author} {\bibinfo {author} {\bibfnamefont {Justin}\ \bibnamefont
  {Dauwels}},\ }\bibfield  {title} {\enquote {\bibinfo {title} {Numerical
  computation of the capacity of continuous memoryless channels},}\ }\href@noop
  {} {\bibfield  {journal} {\bibinfo  {journal} {Proceedings of the 26th
  Symposium on Information Theory in the BENELUX}\ ,\ \bibinfo {pages}
  {221--228}} (\bibinfo {year} {2005})}\BibitemShut {NoStop}%
\bibitem [{\citenamefont {Cao}\ \emph {et~al.}(2013)\citenamefont {Cao},
  \citenamefont {Hranilovic},\ and\ \citenamefont {Chen}}]{ref:Chen-13}%
  \BibitemOpen
  \bibfield  {author} {\bibinfo {author} {\bibfnamefont {Jihai}\ \bibnamefont
  {Cao}}, \bibinfo {author} {\bibfnamefont {S.}~\bibnamefont {Hranilovic}}, \
  and\ \bibinfo {author} {\bibfnamefont {Jun}\ \bibnamefont {Chen}},\
  }\bibfield  {title} {\enquote {\bibinfo {title} {Capacity and nonuniform
  signaling for discrete-time poisson channels},}\ }\href {\doibase
  10.1364/JOCN.5.000329} {\bibfield  {journal} {\bibinfo  {journal} {Optical
  Communications and Networking, IEEE/OSA Journal of}\ }\textbf {\bibinfo
  {volume} {5}},\ \bibinfo {pages} {329--337} (\bibinfo {year}
  {2013})}\BibitemShut {NoStop}%
\bibitem [{\citenamefont {Cao}\ \emph {et~al.}(2014{\natexlab{a}})\citenamefont
  {Cao}, \citenamefont {Hranilovic},\ and\ \citenamefont
  {Chen}}]{ref:Chen-14-1}%
  \BibitemOpen
  \bibfield  {author} {\bibinfo {author} {\bibfnamefont {Jihai}\ \bibnamefont
  {Cao}}, \bibinfo {author} {\bibfnamefont {S.}~\bibnamefont {Hranilovic}}, \
  and\ \bibinfo {author} {\bibfnamefont {Jun}\ \bibnamefont {Chen}},\
  }\bibfield  {title} {\enquote {\bibinfo {title} {Capacity-achieving
  distributions for the discrete-time poisson channel - part 1: General
  properties and numerical techniques},}\ }\href {\doibase
  10.1109/TCOMM.2013.112513.130142} {\bibfield  {journal} {\bibinfo  {journal}
  {Communications, IEEE Transactions on}\ }\textbf {\bibinfo {volume} {62}},\
  \bibinfo {pages} {194--202} (\bibinfo {year}
  {2014}{\natexlab{a}})}\BibitemShut {NoStop}%
\bibitem [{\citenamefont {Cao}\ \emph {et~al.}(2014{\natexlab{b}})\citenamefont
  {Cao}, \citenamefont {Hranilovic},\ and\ \citenamefont
  {Chen}}]{ref:Chen-14-2}%
  \BibitemOpen
  \bibfield  {author} {\bibinfo {author} {\bibfnamefont {Jihai}\ \bibnamefont
  {Cao}}, \bibinfo {author} {\bibfnamefont {S.}~\bibnamefont {Hranilovic}}, \
  and\ \bibinfo {author} {\bibfnamefont {Jun}\ \bibnamefont {Chen}},\
  }\bibfield  {title} {\enquote {\bibinfo {title} {Capacity-achieving
  distributions for the discrete-time poisson channel - part 2: Binary
  inputs},}\ }\href {\doibase 10.1109/TCOMM.2013.112513.130143} {\bibfield
  {journal} {\bibinfo  {journal} {Communications, IEEE Transactions on}\
  }\textbf {\bibinfo {volume} {62}},\ \bibinfo {pages} {203--213} (\bibinfo
  {year} {2014}{\natexlab{b}})}\BibitemShut {NoStop}%
\bibitem [{\citenamefont {Mung}\ and\ \citenamefont {Boyd}(2004)}]{chiang04}%
  \BibitemOpen
  \bibfield  {author} {\bibinfo {author} {\bibfnamefont {Chiang}\ \bibnamefont
  {Mung}}\ and\ \bibinfo {author} {\bibfnamefont {Stephen}\ \bibnamefont
  {Boyd}},\ }\bibfield  {title} {\enquote {\bibinfo {title} {Geometric
  programming duals of channel capacity and rate distortion},}\ }\href
  {\doibase 10.1109/TIT.2003.822581} {\bibfield  {journal} {\bibinfo  {journal}
  {IEEE Transactions on Information Theory}\ }\textbf {\bibinfo {volume}
  {50}},\ \bibinfo {pages} {245--258} (\bibinfo {year} {2004})}\BibitemShut
  {NoStop}%
\bibitem [{\citenamefont {Huang}\ and\ \citenamefont {Meyn}(2005)}]{meyn05}%
  \BibitemOpen
  \bibfield  {author} {\bibinfo {author} {\bibfnamefont {Jianyi}\ \bibnamefont
  {Huang}}\ and\ \bibinfo {author} {\bibfnamefont {Sean~P.}\ \bibnamefont
  {Meyn}},\ }\bibfield  {title} {\enquote {\bibinfo {title} {Characterization
  and computation of optimal distributions for channel coding},}\ }\href
  {\doibase 10.1109/TIT.2005.850108} {\bibfield  {journal} {\bibinfo  {journal}
  {IEEE Transactions on Information Theory}\ }\textbf {\bibinfo {volume}
  {51}},\ \bibinfo {pages} {2336--2351} (\bibinfo {year} {2005})}\BibitemShut
  {NoStop}%
\bibitem [{\citenamefont {Nesterov}(2005)}]{nesterov05}%
  \BibitemOpen
  \bibfield  {author} {\bibinfo {author} {\bibfnamefont {Yurii}\ \bibnamefont
  {Nesterov}},\ }\bibfield  {title} {\enquote {\bibinfo {title} {Smooth
  minimization of non-smooth functions},}\ }\href {\doibase
  10.1007/s10107-004-0552-5} {\bibfield  {journal} {\bibinfo  {journal}
  {Mathematical Programming}\ }\textbf {\bibinfo {volume} {103}},\ \bibinfo
  {pages} {127--152} (\bibinfo {year} {2005})}\BibitemShut {NoStop}%
\bibitem [{\citenamefont {Lapidoth}\ and\ \citenamefont
  {Moser}(2009)}]{lapidoth09}%
  \BibitemOpen
  \bibfield  {author} {\bibinfo {author} {\bibfnamefont {Amos}\ \bibnamefont
  {Lapidoth}}\ and\ \bibinfo {author} {\bibfnamefont {Stefan~M.}\ \bibnamefont
  {Moser}},\ }\bibfield  {title} {\enquote {\bibinfo {title} {On the capacity
  of the discrete-time {P}oisson channel},}\ }\href {\doibase
  10.1109/TIT.2008.2008121} {\bibfield  {journal} {\bibinfo  {journal} {IEEE
  Transactions on Information Theory}\ }\textbf {\bibinfo {volume} {55}},\
  \bibinfo {pages} {303--322} (\bibinfo {year} {2009})}\BibitemShut {NoStop}%
\bibitem [{\citenamefont {Ben-Tal}\ and\ \citenamefont
  {Teboulle}(1988)}]{benTal88}%
  \BibitemOpen
  \bibfield  {author} {\bibinfo {author} {\bibfnamefont {Aharon}\ \bibnamefont
  {Ben-Tal}}\ and\ \bibinfo {author} {\bibfnamefont {Marc}\ \bibnamefont
  {Teboulle}},\ }\bibfield  {title} {\enquote {\bibinfo {title} {Extension of
  some results for channel capacity using a generalized information measure},}\
  }\href {\doibase 10.1007/BF01448363} {\bibfield  {journal} {\bibinfo
  {journal} {Applied Mathematics and Optimization}\ }\textbf {\bibinfo {volume}
  {17}},\ \bibinfo {pages} {121--132} (\bibinfo {year} {1988})}\BibitemShut
  {NoStop}%
\bibitem [{\citenamefont {Mung}(2005)}]{chiang05}%
  \BibitemOpen
  \bibfield  {author} {\bibinfo {author} {\bibfnamefont {Chiang}\ \bibnamefont
  {Mung}},\ }\bibfield  {title} {\enquote {\bibinfo {title} {Geometric
  programming for communication systems},}\ }\href {\doibase
  10.1516/0100000005} {\bibfield  {journal} {\bibinfo  {journal} {Foundations
  and Trends in Communications and Information Theory}\ }\textbf {\bibinfo
  {volume} {2}},\ \bibinfo {pages} {1--154} (\bibinfo {year}
  {2005})}\BibitemShut {NoStop}%
\bibitem [{\citenamefont {Nesterov}(2004)}]{ref:nesterov-book-04}%
  \BibitemOpen
  \bibfield  {author} {\bibinfo {author} {\bibfnamefont {Yurii}\ \bibnamefont
  {Nesterov}},\ }\href@noop {} {\emph {\bibinfo {title} {Introductory Lectures
  on Convex Optimization: A Basic Course}}},\ Applied Optimization\ (\bibinfo
  {publisher} {Springer},\ \bibinfo {year} {2004})\BibitemShut {NoStop}%
\bibitem [{\citenamefont {Bertsekas}(2009)}]{ref:Bertsekas-09}%
  \BibitemOpen
  \bibfield  {author} {\bibinfo {author} {\bibfnamefont {Dimitri~P.}\
  \bibnamefont {Bertsekas}},\ }\href@noop {} {\emph {\bibinfo {title} {Convex
  Optimization Theory}}},\ Athena Scientific optimization and computation
  series\ (\bibinfo  {publisher} {Athena Scientific},\ \bibinfo {year}
  {2009})\BibitemShut {NoStop}%
\bibitem [{\citenamefont {Witsenhausen}(1980)}]{witsenhausen}%
  \BibitemOpen
  \bibfield  {author} {\bibinfo {author} {\bibfnamefont {H.S.}\ \bibnamefont
  {Witsenhausen}},\ }\bibfield  {title} {\enquote {\bibinfo {title} {Some
  aspects of convexity useful in information theory},}\ }\href {\doibase
  10.1109/TIT.1980.1056173} {\bibfield  {journal} {\bibinfo  {journal} {IEEE
  Transactions on Information Theory}\ }\textbf {\bibinfo {volume} {26}},\
  \bibinfo {pages} {265--271} (\bibinfo {year} {1980})}\BibitemShut {NoStop}%
\bibitem [{\citenamefont {Borwein}\ and\ \citenamefont
  {Lewis}(1991)}]{ref:Borwein-91}%
  \BibitemOpen
  \bibfield  {author} {\bibinfo {author} {\bibfnamefont {J.~M.}\ \bibnamefont
  {Borwein}}\ and\ \bibinfo {author} {\bibfnamefont {A.~S.}\ \bibnamefont
  {Lewis}},\ }\bibfield  {title} {\enquote {\bibinfo {title} {Duality
  relationships for entropy-like minimization problems},}\ }\href {\doibase
  10.1137/0329017} {\bibfield  {journal} {\bibinfo  {journal} {SIAM J. Control
  Optim.}\ }\textbf {\bibinfo {volume} {29}},\ \bibinfo {pages} {325--338}
  (\bibinfo {year} {1991})}\BibitemShut {NoStop}%
\bibitem [{\citenamefont {Lasserre}(2009)}]{ref:Lasserre-11}%
  \BibitemOpen
  \bibfield  {author} {\bibinfo {author} {\bibfnamefont {Jean~B.}\ \bibnamefont
  {Lasserre}},\ }\href@noop {} {\emph {\bibinfo {title} {Moments, Positive
  Polynomials and Their Applications}}},\ Imperial College Press optimization
  series\ (\bibinfo  {publisher} {Imperial College Press},\ \bibinfo {year}
  {2009})\BibitemShut {NoStop}%
\bibitem [{\citenamefont {Leung}\ and\ \citenamefont {Smith}(2009)}]{leung09}%
  \BibitemOpen
  \bibfield  {author} {\bibinfo {author} {\bibfnamefont {Debbie}\ \bibnamefont
  {Leung}}\ and\ \bibinfo {author} {\bibfnamefont {Graeme}\ \bibnamefont
  {Smith}},\ }\bibfield  {title} {\enquote {\bibinfo {title} {Continuity of
  quantum channel capacities},}\ }\href {\doibase 10.1007/s00220-009-0833-1}
  {\bibfield  {journal} {\bibinfo  {journal} {Communications in Mathematical
  Physics}\ }\textbf {\bibinfo {volume} {292}},\ \bibinfo {pages} {201--215}
  (\bibinfo {year} {2009})}\BibitemShut {NoStop}%
\bibitem [{\citenamefont {Cover}\ and\ \citenamefont {Thomas}(2006)}]{cover}%
  \BibitemOpen
  \bibfield  {author} {\bibinfo {author} {\bibfnamefont {Thomas~M.}\
  \bibnamefont {Cover}}\ and\ \bibinfo {author} {\bibfnamefont {Joy~A.}\
  \bibnamefont {Thomas}},\ }\href@noop {} {\emph {\bibinfo {title} {Elements of
  Information Theory}}}\ (\bibinfo  {publisher} {Wiley Interscience},\ \bibinfo
  {year} {2006})\BibitemShut {NoStop}%
\bibitem [{\citenamefont {Moser}(2005)}]{moser_phd}%
  \BibitemOpen
  \bibfield  {author} {\bibinfo {author} {\bibfnamefont {Stefan~M.}\
  \bibnamefont {Moser}},\ }\bibfield  {title} {\enquote {\bibinfo {title}
  {Duality-based bounds on channel capacity},}\ }\href@noop {} {\bibfield
  {journal} {\bibinfo  {journal} {PhD thesis, ETH Zurich}\ } (\bibinfo {year}
  {2005})}\BibitemShut {NoStop}%
\bibitem [{\citenamefont {Shamai}(1990)}]{shamai90}%
  \BibitemOpen
  \bibfield  {author} {\bibinfo {author} {\bibfnamefont {Shlomo}\ \bibnamefont
  {Shamai}},\ }\bibfield  {title} {\enquote {\bibinfo {title} {Capacity of a
  pulse amplitude modulated direct detection photon channel},}\ }\href@noop {}
  {\bibfield  {journal} {\bibinfo  {journal} {IEE Proceedings on
  Communications, Speech and Vision}\ }\textbf {\bibinfo {volume} {137}},\
  \bibinfo {pages} {424--430} (\bibinfo {year} {1990})}\BibitemShut {NoStop}%
\bibitem [{\citenamefont {Fremlin}(2010)}]{ref:fremlin-03}%
  \BibitemOpen
  \bibfield  {author} {\bibinfo {author} {\bibfnamefont {D.~H.}\ \bibnamefont
  {Fremlin}},\ }\href@noop {} {\emph {\bibinfo {title} {Measure theory. {V}ol.
  2}}}\ (\bibinfo  {publisher} {Torres Fremlin, Colchester},\ \bibinfo {year}
  {2010})\ pp.\ \bibinfo {pages} {563+12 pp. (errata)},\ \bibinfo {note} {broad
  foundations, Second edition January 2010}\BibitemShut {NoStop}%
\bibitem [{\citenamefont {Anderson}\ and\ \citenamefont
  {Nash}(1987)}]{anderson87}%
  \BibitemOpen
  \bibfield  {author} {\bibinfo {author} {\bibfnamefont {Edward~J.}\
  \bibnamefont {Anderson}}\ and\ \bibinfo {author} {\bibfnamefont {Peter}\
  \bibnamefont {Nash}},\ }\href@noop {} {\emph {\bibinfo {title} {Linear
  programming in infinite-dimensional spaces: theory and applications}}},\
  Wiley-Interscience Series in Discrete Mathematics and Optimization\ (\bibinfo
   {publisher} {Wiley},\ \bibinfo {year} {1987})\BibitemShut {NoStop}%
\bibitem [{\citenamefont {Mitter}(2008)}]{mitter08}%
  \BibitemOpen
  \bibfield  {author} {\bibinfo {author} {\bibfnamefont {Sanjoy~K.}\
  \bibnamefont {Mitter}},\ }\bibfield  {title} {\enquote {\bibinfo {title}
  {Convex optimization in infinite dimensional spaces},}\ }in\ \href {\doibase
  10.1007/978-1-84800-155-8_12} {\emph {\bibinfo {booktitle} {Recent advances
  in learning and control}}},\ \bibinfo {series} {Lecture Notes in Control and
  Inform. Sci.}, Vol.\ \bibinfo {volume} {371}\ (\bibinfo  {publisher}
  {Springer, London},\ \bibinfo {year} {2008})\ pp.\ \bibinfo {pages}
  {161--179}\BibitemShut {NoStop}%
\bibitem [{\citenamefont {Devolder}\ \emph {et~al.}(2013)\citenamefont
  {Devolder}, \citenamefont {Glineur},\ and\ \citenamefont
  {Nesterov}}]{ref:Devolver-13}%
  \BibitemOpen
  \bibfield  {author} {\bibinfo {author} {\bibfnamefont {Olivier}\ \bibnamefont
  {Devolder}}, \bibinfo {author} {\bibfnamefont {Franois}\ \bibnamefont
  {Glineur}}, \ and\ \bibinfo {author} {\bibfnamefont {Yurii}\ \bibnamefont
  {Nesterov}},\ }\bibfield  {title} {\enquote {\bibinfo {title} {First-order
  methods of smooth convex optimization with inexact oracle},}\ }\href
  {\doibase 10.1007/s10107-013-0677-5} {\bibfield  {journal} {\bibinfo
  {journal} {Mathematical Programming}\ ,\ \bibinfo {pages} {1--39}} (\bibinfo
  {year} {2013})}\BibitemShut {NoStop}%
\bibitem [{\citenamefont {Sutter}\ \emph {et~al.}(2014)\citenamefont {Sutter},
  \citenamefont {Sutter}, \citenamefont {Mohajerin~Esfahani},\ and\
  \citenamefont {Renner}}]{DavidSutter14}%
  \BibitemOpen
  \bibfield  {author} {\bibinfo {author} {\bibfnamefont {David}\ \bibnamefont
  {Sutter}}, \bibinfo {author} {\bibfnamefont {Tobias}\ \bibnamefont {Sutter}},
  \bibinfo {author} {\bibfnamefont {Peyman}\ \bibnamefont
  {Mohajerin~Esfahani}}, \ and\ \bibinfo {author} {\bibfnamefont {Renato}\
  \bibnamefont {Renner}},\ }\bibfield  {title} {\enquote {\bibinfo {title}
  {Efficient approximation of quantum channel capacities},}\ }\href@noop {} {\
  (\bibinfo {year} {2014})},\ \bibinfo {note} {available at
  \texttt{arXiv:\href{http://arxiv.org/abs/1407.8202}{1407.8202}}}\BibitemShut
  {NoStop}%
\bibitem [{\citenamefont {Brady}\ and\ \citenamefont
  {Verd\'u}(1990)}]{brady90}%
  \BibitemOpen
  \bibfield  {author} {\bibinfo {author} {\bibfnamefont {David}\ \bibnamefont
  {Brady}}\ and\ \bibinfo {author} {\bibfnamefont {Sergio}\ \bibnamefont
  {Verd\'u}},\ }\bibfield  {title} {\enquote {\bibinfo {title} {The asymptotic
  capacity of the direct detection photon channel with a bandwidth
  constraint},}\ }in\ \href@noop {} {\emph {\bibinfo {booktitle} {28th Annual
  Allerton Conference on Communication, Control, and Computing}}}\ (\bibinfo
  {year} {1990})\ pp.\ \bibinfo {pages} {691--700}\BibitemShut {NoStop}%
\bibitem [{\citenamefont {Martinez}(2007)}]{martinez07}%
  \BibitemOpen
  \bibfield  {author} {\bibinfo {author} {\bibfnamefont {Alfonso}\ \bibnamefont
  {Martinez}},\ }\bibfield  {title} {\enquote {\bibinfo {title} {Spectral
  efficiency of optical direct detection},}\ }\href {\doibase
  10.1364/JOSAB.24.000739} {\bibfield  {journal} {\bibinfo  {journal} {J. Opt.
  Soc. Am. B}\ }\textbf {\bibinfo {volume} {24}},\ \bibinfo {pages} {739--749}
  (\bibinfo {year} {2007})}\BibitemShut {NoStop}%
\bibitem [{\citenamefont {Lapidoth}\ \emph {et~al.}(2011)\citenamefont
  {Lapidoth}, \citenamefont {Shapiro}, \citenamefont {Venkatesan},\ and\
  \citenamefont {Wang}}]{lapidoth11}%
  \BibitemOpen
  \bibfield  {author} {\bibinfo {author} {\bibfnamefont {Amos}\ \bibnamefont
  {Lapidoth}}, \bibinfo {author} {\bibfnamefont {Jeffrey~H.}\ \bibnamefont
  {Shapiro}}, \bibinfo {author} {\bibfnamefont {Vinodh}\ \bibnamefont
  {Venkatesan}}, \ and\ \bibinfo {author} {\bibfnamefont {Ligong}\ \bibnamefont
  {Wang}},\ }\bibfield  {title} {\enquote {\bibinfo {title} {The discrete-time
  {P}oisson channel at low input powers},}\ }\href {\doibase
  10.1109/TIT.2011.2134430} {\bibfield  {journal} {\bibinfo  {journal} {IEEE
  Transactions on Information Theory}\ }\textbf {\bibinfo {volume} {57}},\
  \bibinfo {pages} {3260--3272} (\bibinfo {year} {2011})}\BibitemShut {NoStop}%
\bibitem [{\citenamefont {Singh}\ \emph {et~al.}(2009)\citenamefont {Singh},
  \citenamefont {Dabeer},\ and\ \citenamefont {Madhow}}]{singh09}%
  \BibitemOpen
  \bibfield  {author} {\bibinfo {author} {\bibfnamefont {J.}~\bibnamefont
  {Singh}}, \bibinfo {author} {\bibfnamefont {O.}~\bibnamefont {Dabeer}}, \
  and\ \bibinfo {author} {\bibfnamefont {U.}~\bibnamefont {Madhow}},\
  }\bibfield  {title} {\enquote {\bibinfo {title} {On the limits of
  communication with low-precision analog-to-digital conversion at the
  receiver},}\ }\href {\doibase 10.1109/TCOMM.2009.12.080559} {\bibfield
  {journal} {\bibinfo  {journal} {IEEE Transactions on Communications}\
  }\textbf {\bibinfo {volume} {57}},\ \bibinfo {pages} {3629--3639} (\bibinfo
  {year} {2009})}\BibitemShut {NoStop}%
\bibitem [{\citenamefont {Koch}\ and\ \citenamefont {Lapidoth}(2013)}]{koch13}%
  \BibitemOpen
  \bibfield  {author} {\bibinfo {author} {\bibfnamefont {Tobias}\ \bibnamefont
  {Koch}}\ and\ \bibinfo {author} {\bibfnamefont {Amos}\ \bibnamefont
  {Lapidoth}},\ }\bibfield  {title} {\enquote {\bibinfo {title} {At low {SNR},
  asymmetric quantizers are better},}\ }\href {\doibase
  10.1109/TIT.2013.2262919} {\bibfield  {journal} {\bibinfo  {journal} {IEEE
  Transactions on Information Theory}\ }\textbf {\bibinfo {volume} {59}},\
  \bibinfo {pages} {5421--5445} (\bibinfo {year} {2013})}\BibitemShut {NoStop}%
\bibitem [{\citenamefont {Baes}\ and\ \citenamefont
  {B\"urgisser}(2014)}]{ref:Baes-14}%
  \BibitemOpen
  \bibfield  {author} {\bibinfo {author} {\bibfnamefont {Michel}\ \bibnamefont
  {Baes}}\ and\ \bibinfo {author} {\bibfnamefont {Michael}\ \bibnamefont
  {B\"urgisser}},\ }\bibfield  {title} {\enquote {\bibinfo {title} {An
  acceleration procedure for optimal first-order methods},}\ }\href {\doibase
  10.1080/10556788.2013.835812} {\bibfield  {journal} {\bibinfo  {journal}
  {Optimization Methods and Software}\ }\textbf {\bibinfo {volume} {29}},\
  \bibinfo {pages} {610--628} (\bibinfo {year} {2014})}\BibitemShut {NoStop}%
\bibitem [{\citenamefont {Holevo}(2012)}]{holevo_book}%
  \BibitemOpen
  \bibfield  {author} {\bibinfo {author} {\bibfnamefont {Alexander~S.}\
  \bibnamefont {Holevo}},\ }\href@noop {} {\emph {\bibinfo {title} {Quantum
  Systems, Channels, Information}}}\ (\bibinfo  {publisher} {De Gruyter Studies
  in Mathematical Physics 16},\ \bibinfo {year} {2012})\BibitemShut {NoStop}%
\bibitem [{\citenamefont {Boyd}\ and\ \citenamefont
  {Vandenberghe}(2004)}]{ref:BoyVan-04}%
  \BibitemOpen
  \bibfield  {author} {\bibinfo {author} {\bibfnamefont {Stephen}\ \bibnamefont
  {Boyd}}\ and\ \bibinfo {author} {\bibfnamefont {Lieven}\ \bibnamefont
  {Vandenberghe}},\ }\href@noop {} {\emph {\bibinfo {title} {Convex
  {O}ptimization}}}\ (\bibinfo  {publisher} {Cambridge University Press},\
  \bibinfo {address} {Cambridge},\ \bibinfo {year} {2004})\ pp.\ \bibinfo
  {pages} {xiv+716},\ \bibinfo {note} {sixth printing with corrections,
  2008}\BibitemShut {NoStop}%
\bibitem [{\citenamefont {Folland}(1999)}]{ref:Folland-99}%
  \BibitemOpen
  \bibfield  {author} {\bibinfo {author} {\bibfnamefont {Gerald~B.}\
  \bibnamefont {Folland}},\ }\href@noop {} {\emph {\bibinfo {title} {Real
  analysis: modern techniques and their applications}}},\ Pure and applied
  mathematics\ (\bibinfo  {publisher} {Wiley},\ \bibinfo {year}
  {1999})\BibitemShut {NoStop}%
\bibitem [{\citenamefont {Wu}\ and\ \citenamefont
  {Verd\'u}(2012)}]{ref:verdu-12}%
  \BibitemOpen
  \bibfield  {author} {\bibinfo {author} {\bibfnamefont {Yihong}\ \bibnamefont
  {Wu}}\ and\ \bibinfo {author} {\bibfnamefont {Sergio}\ \bibnamefont
  {Verd\'u}},\ }\bibfield  {title} {\enquote {\bibinfo {title} {Functional
  properties of minimum mean-square error and mutual information},}\ }\href
  {\doibase 10.1109/TIT.2011.2174959} {\bibfield  {journal} {\bibinfo
  {journal} {IEEE Transactions on Information Theory}\ }\textbf {\bibinfo
  {volume} {58}},\ \bibinfo {pages} {1289--1301} (\bibinfo {year}
  {2012})}\BibitemShut {NoStop}%
\bibitem [{\citenamefont {Billingsley}(1968)}]{ref:billingsley-68}%
  \BibitemOpen
  \bibfield  {author} {\bibinfo {author} {\bibfnamefont {Patrick}\ \bibnamefont
  {Billingsley}},\ }\href@noop {} {\emph {\bibinfo {title} {Convergence of
  probability measures}}},\ Wiley Series in probability and Mathematical
  Statistics: Tracts on probability and statistics\ (\bibinfo  {publisher}
  {Wiley},\ \bibinfo {year} {1968})\BibitemShut {NoStop}%
\bibitem [{\citenamefont {Devolder}\ \emph {et~al.}(2012)\citenamefont
  {Devolder}, \citenamefont {Glineur},\ and\ \citenamefont
  {Nesterov}}]{ref:devolder-12}%
  \BibitemOpen
  \bibfield  {author} {\bibinfo {author} {\bibfnamefont {Olivier}\ \bibnamefont
  {Devolder}}, \bibinfo {author} {\bibfnamefont {Fran\c{c}ois}\ \bibnamefont
  {Glineur}}, \ and\ \bibinfo {author} {\bibfnamefont {Yurii}\ \bibnamefont
  {Nesterov}},\ }\bibfield  {title} {\enquote {\bibinfo {title} {Double
  smoothing technique for large-scale linearly constrained convex
  optimization},}\ }\href {\doibase 10.1137/110826102} {\bibfield  {journal}
  {\bibinfo  {journal} {SIAM Journal on Optimization}\ }\textbf {\bibinfo
  {volume} {22}},\ \bibinfo {pages} {702--727} (\bibinfo {year}
  {2012})}\BibitemShut {NoStop}%
\end{thebibliography}%


\clearpage



\end{document}